\documentclass[reqno]{amsart}
\usepackage[lite]{amsrefs}
\usepackage{amssymb}
\usepackage{bm}
\usepackage[all,cmtip]{xy}

\newcommand{\T}{\mathbb{T}}

\newcommand{\R}{\mathbb{R}}
\newcommand{\C}{\mathbb{C}}
\newcommand{\Z}{\mathbb{Z}}
\newcommand{\N}{\mathbb{N}}

\newcommand{\E}{\mathbb{E}}

\newcommand{\Om}{\Omega}
\newcommand{\eps}{\varepsilon}

\newcommand{\mcN}{\mathcal{N}}
\newcommand{\mcNt}{\mcN_\tau}
\newcommand{\mcA}{\mathcal{A}}

\newcommand{\mcW}{\mathcal{W}}
 
\newcommand{\mcWt}{\mathcal{W}_\tau}
\newcommand{\mfh}{\mathfrak{h}}

\newcommand{\mbt}{\mathbf{t}}
\newcommand{\mbx}{\mathbf{x}}
\newcommand{\mby}{\mathbf{y}}

\newcommand{\mcB}{\mathcal{B}}
\newcommand{\mcP}{\mathcal{P}}
\newcommand{\mfP}{\mathfrak{P}}
\newcommand{\etavstar}{\left(\e^{t_\alpha h /\tau}\varphi^*_\tau\right)}

\newcommand{\g}{g^\xi_{\tau,m}(\mbt)}

\newcommand{\mbk}{\mathbf{k}}

\newcommand{\Jte}{\mathcal{J}_{\tau,e}}
\newcommand{\mcJ}{\mathcal{J}}
\newcommand{\vph}{\varphi}

\newcommand{\lk}{\lambda_k}
\newcommand{\dd}{\mathrm{d}}
\newcommand{\Tr}{\text{Tr}}
\newcommand{\La}{\Lambda}
\newcommand{\mbP}{\mathbf{P}}
\newcommand{\e}{\mathrm{e}}
\newcommand{\zG}{z_{\mathrm{Gibbs}}}
\newcommand{\mfS}{\mathfrak{S}}

\newcommand{\mfhp}{\mathfrak{h}^{(p)}}
\newcommand{\mcCp}{\mathcal{C}_p}
\newcommand{\mc}{\mathcal}
\newcommand{\mf}{\mathfrak}
\newcommand{\om}{\omega}
\newcommand{\Th}{\Theta}
\newcommand{\vphts}{\varphi_\tau^*}
\newcommand{\vpht}{\varphi_\tau}
\newcommand{\Hf}{H_{\tau,0}}
\newcommand{\rt}{\rho_\tau}
\newcommand{\Zf}{Z_{\tau,0}}
\newcommand{\mbb}{\mathbb}
\newcommand{\tm}{\tau,m}
\newcommand{\tM}{\tau,M}
\newcommand{\tr}{\tilde{\rho}}

\newcommand{\HnuN}{\left(H_{\tau,0} + \nu\mcNt\right)}
\newcommand{\mbs}{\mathbf{s}}

\newcommand{\tinfty}{\tau \to \infty}
\newcommand{\mcL}{\mathcal{L}}
\newcommand{\mthc}{\mathcal}
\numberwithin{equation}{section}

\theoremstyle{plain}

\theoremstyle{definition}

\theoremstyle{remark}

\usepackage{graphicx, color}

\definecolor{darkred}{rgb}{0.9,0,0.3}
\definecolor{darkblue}{rgb}{0,0.3,0.9}

\newcommand{\vertiii}[1]{{\left\vert\kern-0.25ex\left\vert\kern-0.25ex\left\vert #1 
		\right\vert\kern-0.25ex\right\vert\kern-0.25ex\right\vert}}

\theoremstyle{plain} %plain, definition, remark
\newtheorem{theorem}{Theorem}[section]
\newtheorem*{theorem*}{Theorem}
\newtheorem{lemma}[theorem]{Lemma}
\newtheorem*{lemma*}{Lemma}
\newtheorem{corollary}[theorem]{Corollary}
\newtheorem*{corollary*}{Corollary}
\newtheorem{proposition}[theorem]{Proposition}
\newtheorem*{proposition*}{Proposition}

\newtheorem*{conjecture*}{Conjecture}

\theoremstyle{definition} %plain, definition, remark
\newtheorem{definition}[theorem]{Definition}
\newtheorem*{definition*}{Definition}

\newtheorem*{example*}{Example}
\newtheorem{remark}[theorem]{Remark}
\newtheorem*{remark*}{Remark}
\newtheorem{assumption}[theorem]{Assumption}
\newtheorem*{assumption*}{Assumption}

\begin{document}
	\title[Derivation of Gibbs measures for the 1D focusing quintic NLS]{A microscopic derivation of Gibbs measures for the 1D focusing quintic nonlinear Schr\"odinger equation}
	
	\author{Andrew Rout and Vedran Sohinger}
	\maketitle
	
\begin{abstract}
In this work, we obtain a microscopic derivation of Gibbs measures for the focusing quintic nonlinear Schr\"{o}dinger equation (NLS) on $\mathbb{T}$ from many-body quantum Gibbs states. On the quantum many-body level, the quintic nonlinearity corresponds to a three-body interaction.
This is a continuation of our previous work \cite{RS22}. In the aforementioned work, we studied the cubic problem, which corresponds to a two-body interaction on the quantum many-body level.

In our setup, we truncate the mass of the classical free field in the classical setting and the rescaled particle number in the quantum setting. 
Our methods are based on a perturbative expansion previously developed in the work of Fr\"{o}hlich, Knowles, Schlein, and the second author \cite{FKSS17}.
We prove results both in the time-independent and time-dependent setting. This is the first such known result in the three-body regime. Furthermore, this gives the first microscopic derivation of time-dependent correlation functions for Gibbs measures corresponding to the quintic NLS, as studied in the work of Bourgain \cite{Bou94}.
\end{abstract}

\section{Introduction}
\subsection{The general setup}
We work on the one-dimensional torus $\Lambda := \T^1 \equiv \R/\Z \equiv [-1/2,1/2)$. The \emph{one-particle space} is given by $\mfh := L^2(\T;\C)$, with the inner product convention
\begin{equation}
\label{inner_product}
\langle f_1,f_2 \rangle := \int \dd x \, \overline{f_1}(x)f_2(x)\,.
\end{equation}
In \eqref{inner_product} and in the sequel, we abbreviate $\int_{\Lambda} \dd x$ by $\int \dd x$.
We define the \emph{one-body Hamiltonian} as
\begin{equation}
	\label{q_1-body_Hamn_defn}
h:= -\Delta + \kappa,
\end{equation}
where $\kappa > 0$ is a fixed quantity throughout. Note that $h$ is a positive self-adjoint densely-defined operator on $\mfh$. We can write $h$ spectrally as
\begin{equation}
	\label{q_Hamn_spectral_decomp}
h=\sum_{k \in \N} \lk e_k e^*_k\,,
\end{equation}
where
\begin{equation}
\label{lambda_k}
\lk := 4 \pi^2 |k|^2 + \kappa,
\end{equation}
are the eigenvalues of $h$ and 
\begin{equation}
	\label{q_eigenvalues_defn}
e_k(x) := \e^{2\pi \mathrm{i} k x}
\end{equation}
are the corresponding normalised eigenvalues of $h$ in $\mfh$. It follows from \eqref{q_Hamn_spectral_decomp}--\eqref{lambda_k} that
\begin{equation}
	\label{q_trace_finite}
\Tr_{\mfh}(h^{-1}) = \sum_{k \in \N} \frac{1}{\lk} < \infty\,.
\end{equation}
%In this paper, when it is clear from context, we will omit the space we are taking traces over. Where convenient, we also denote by $g(x)$ by $g_x$. {\color{blue} (Vedran: I think that we can avoid using this convention in this part of the paper. Maybe we can use it at a later point if it is really necessary.)}

Let us now state the assumptions of the interaction potential in our model.

\begin{assumption}
\label{w_assumption_bounded}
Let $w: \Lambda \rightarrow \R$ be even and such that $w \in L^{\infty}(\Lambda)$.
\end{assumption}

In certain parts of our analysis, we can study more singular interaction potentials.

\begin{assumption}
\label{w_assumption}
We consider one of the following forms of interaction potentials.
\begin{itemize}
\item[(i)] $w: \Lambda \rightarrow \R$ is even and $w \in L^{1}(\Lambda)$.
\item[(ii)] $w=c\delta$ for some $c \in \R$.
\end{itemize}
\end{assumption}

In Assumptions \ref{w_assumption_bounded}--\ref{w_assumption} above, we make no condition on the (pointwise) sign of $w$ or $\hat{w}$.

\textbf{Convention:} When working in the classical setting, we consider $w$ as in Assumption \ref{w_assumption}.
When working in the quantum setting, we consider $w$ as in Assumption \ref{w_assumption_bounded}. 

\subsubsection*{The classical model}
We now introduce the classical model that we will analyse.
For $w$ as in Assumption \ref{w_assumption}, we study the following form of the {\it nonlinear Schr\"odinger equation (NLS)}.

\begin{align}
\notag
\mathrm{i}\partial_t u + (\Delta - \kappa)u =&-\frac{1}{3}\,\int \dd y \, \dd z \, w(x-y)\,w(x-z)\,|u(y)|^2\,|u(z)|^2\,u(x)
\\
\label{q_Hartree_defn}
&-\frac{2}{3}\,\int \dd y \, \dd z \, w(x-y)\,w(y-z)\,|u(y)|^2\,|u(z)|^2\,u(x)\,.
\end{align}
%{\color{red} Vedran: I have removed references to Remark \ref{q_interaction_potential_remark} and replaced them with references to Assumption \ref{w_assumption}. Later, we should delete Remark \ref{q_interaction_potential_remark}.
%\begin{remark}
	%\label{q_interaction_potential_remark}
%{\color{blue}In} this paper, we will consider even interaction potentials $w:\La \to \R$ of the form
	%\begin{enumerate}
		%\item \label{q_interaction_assumption_1} $w \in L^\infty(\La)$,
%\end{enumerate}
%\end{remark}	
%}
We refer to \eqref{q_Hartree_defn} as the \emph{quintic Hartree equation} or \emph{nonlocal quintic NLS}. 
Let us note that when $w=\delta$ (as in Assumption \ref{w_assumption} (ii)), \eqref{q_Hartree_defn} is the \emph{focusing quintic NLS}
\begin{equation}
\label{focusing_quintic_NLS}
\mathrm{i}\partial_t u + (\Delta - \kappa)u=-|u|^4\,u\,.
\end{equation}

On the space of fields $u: \La \to \C$, we consider the Poisson structure where the Poisson bracket is given by
\begin{equation}
\label{Poisson_structure}
\{u(x),\bar{u}(y)\} = \mathrm{i} \delta(x-y), \qquad \{u(x),u(y)\} = \{\bar{u}(x),\bar{u}(y)\} = 0\,.
\end{equation}

\begin{lemma}
\label{Hamiltonian_Lemma}
With Poisson structure given by \eqref{Poisson_structure}, we have that \eqref{q_Hartree_defn} corresponds to the Hamiltonian equation of motion associated with Hamiltonian
\begin{multline}
\label{q_Hamn_defn}
H(u) = \int \dd x\, \left(|\nabla u(x)|^2 + \kappa |u(x)|^2\right) \\ 
- \frac{1}{3}\int \dd x \, \dd y \, \dd z \, w(x-y)\,w(x-z)\,|u(x)|^2\,|u(y)|^2\,|u(z)|^2\,.
\end{multline}
\end{lemma}

\begin{proof}
By a direct calculation using \eqref{Poisson_structure} and Assumption \ref{w_assumption}, we obtain that
\begin{multline}
\label{Poisson_bracket_H}
\{H,u\}(x)=\mathrm{i} (\Delta u(x)-\kappa u(x))+\frac{\mathrm{i}}{3} \int \dd y\,\dd z\, w(x-y)\,w(x-z)\,|u(y)|^2\,|u(z)|^2\,u(x)
\\
+\frac{\mathrm{2i}}{3} \int \dd y\,\dd z\, w(x-y)\,w(y-z)\,|u(y)|^2\,|u(z)|^2\,u(x)\,.
\end{multline}
\end{proof}

%\begin{remark}
%We emphasise that the calculation \eqref{Poisson_bracket_H} does not rely on the evenness of $w$. This is in contrast to the cubic Hartree equation
%\begin{equation*}
%\mathrm{i}\partial_t u + (\Delta - \kappa)u =\int \dd y \, w(x-y)\,|u(y)|^2\,u(x)\,,
%\end{equation*}
%which is a Hamiltonian equation of motion associated with Hamiltonian
%\begin{equation*}
%H(u) = \int \dd x\, \left(|\nabla u(x)|^2 + \kappa |u(x)|^2\right) \\ 
%+ \frac{1}{2}\int \dd x \, \dd y \,  |u(x)|^2\,w(x-y)\,|u(y)|^2\,,
%\end{equation*}
%and Poisson structure is given by \eqref{Poisson_structure} if and only if $w$ is even. In fact, for the entire analysis of the classical problem, we can omit the assumption that $w$ is even in Assumption \ref{w_assumption} above. When considering the $n$-body Hamiltonian given in \eqref{q_n_body_Hamn} below, it is realistic to assume that $w$ is even in the sense that the interaction between particle $x_i$ and $x_j$ is the same as the interaction between the particle $x_j$ and $x_i$ for $1 \leq i<j \leq n$. Hence, we will consider $w$ even in our results.
%\end{remark}
It was shown in \cite{Bou93} that \eqref{focusing_quintic_NLS} is locally well-posed for $w$ as in Assumption \ref{w_assumption} (ii). We show in Section \ref{Cauchy problem} that \eqref{q_Hartree_defn} that the analogous result holds for $w$ as in Assumption \ref{w_assumption} (i).
Here, $H^s(\La)$ denotes the $L^2$-based Sobolev space of order $s$ on $\La$ with norm given by
\begin{equation*}
\|g\|_{H^s(\La)} := \|(1+|k|)^s|\hat{g}(k)|\|_{\ell^2_k}\,,
\end{equation*}
where we use the following convention for the Fourier transform of $g \in L^1(\La)$.
\begin{equation*}
\hat{g}(k) := \int \dd x \, g(x)\,\e^{-2 \pi  \mathrm{i} k x}, \quad k \in \mathbb{Z}.
\end{equation*}

We study suitable \emph{Gibbs measures} associated with \eqref{q_Hartree_defn}. Let us recall the following general notion. The Gibbs measure associated with an abstract positive-definite Hamiltonian $H$ is the probability measure $\dd\mbP_{\mathrm{Gibbs}}$ on the space of fields $u : \La \to \C$ formally defined by
\begin{equation}
	\label{q_gibbs_measure_defn}
\dd\mbP_{\mathrm{Gibbs}}(u) := \frac{1}{z_{\mathrm{Gibbs}}} \e^{-H(u)} \,\dd u,
\end{equation}
where $z_{\mathrm{Gibbs}}$ is the {\it partition function}, i.e.\ the normalisation constant making $\dd\mbP_{\mathrm{Gibbs}}$ a probability measure, and $\dd u$ is the formally defined Lebesgue measure on the space of fields. The rigorous construction of such measures was first studied in the constructive quantum field theory literature, see for example the works \cite{GJ87,Nel73a,Nel73b,Simon74}. They were further studied in \cite{BFS83,BG20,BS96,CFL16,GH21,LRS88,MV97a,MV97b,OST22} and the references therein. The invariance of measures of the form of \eqref{q_gibbs_measure_defn} under the flow of the nonlinear Schr\"{o}dinger equation was first rigorously proven by Bourgain in \cite{Bou94,Bou96,Bou97}, with some preliminary results known to Zhidkov in \cite{Zhi91}. This invariance has been used to construct global solutions for equations with rough initial conditions. For further details, we refer the reader to the expository works \cite{BTT18,NS19,OT18}.
For more recent developments, we also refer the reader to \cite{Bringmann1,Bringmann2,Bringmann_Deng_Nahmod_Yue, Deng_Nahmod_Yue_1,Deng_Nahmod_Yue, Gunaratnam_Oh_Tzvetkov_Weber, OST22} and the references given therein.
 
When $H$ is not positive definite, it is in general not possible to normalise $\e^{-H(u)}\,\dd u$. In other words, one formally has $\zG=\int \dd\mbP_{\mathrm{Gibbs}}(u) =  \infty$. In this case, one considers instead the modification of \eqref{q_gibbs_measure_defn} given by
\begin{equation}
\label{q_gibbs_cutoff_defn}
\dd \mbP^f_{\mathrm{Gibbs}}(u) := \frac{1}{z_{\mathrm{Gibbs}}^f} \,\e^{-H(u)}\, f(\|u\|_{L^2}^2) \,\dd u,
\end{equation}
where $f \in C_0^\infty(\R)$ is a suitable cut-off function and 
\begin{equation*}
z_{\mathrm{Gibbs}}^f=\int  \,\dd u\, \e^{-H(u)}\, f(\|u\|_{L^2}^2)
\end{equation*}
is the formal normalisation constant that makes \eqref{q_gibbs_cutoff_defn} a probability measure. Such measures were considered in \cite{Ammari_Sohinger,Bou94,CFL16,DR23,DRTW23,GJ87,Tolomeo_Weber}. For a non positive-definite Hamiltonian $H$, we say the Hamiltonian equation associated with $H$ is in the {\it focusing} (or {\it thermodynamically unstable}) regime.

Since we are considering interaction potentials $w$ given by Assumptions \ref{w_assumption_bounded}-\ref{w_assumption} above, it follows that the Hamiltonian \eqref{q_Hamn_defn} is not necessarily positive-definite.  As such the Gibbs measures we work with in this paper will always be of the cut-off form given by \eqref{q_gibbs_cutoff_defn}. When it is clear, we will omit the $f$ dependence from $\dd\mbP^f_{\mathrm{Gibbs}} \equiv \dd\mbP_{\mathrm{Gibbs}}$.

In Section \ref{Cauchy problem}, we show that \eqref{q_Hartree_defn} is globally well-posed for initial conditions in the support of the Gibbs measure \eqref{q_gibbs_cutoff_defn}. For a precise statement, see Proposition \ref{Cauchy_problem_2} below. In particular, there is a well-defined solution map $S_t$ that maps any initial condition $u_0$ in the support of \eqref{q_gibbs_cutoff_defn} to the solution of \eqref{q_Hartree_defn} at time $t$ given by
\begin{equation}
\label{q_evolution_defn}
u(\cdot) \equiv u(\cdot,t) := S_t(u_0)\,.
\end{equation}
The map $S_t$ preserves regularity in the sense that for $s \in (0,\frac{1}{2})$ and $u_0 \in H^s$, we have $S_t(u_0) \in H^s$.

Let $\dd\mbP_{\mathrm{Gibbs}}^f$ be as in \eqref{q_gibbs_cutoff_defn} and let be $S_t$ as in \eqref{q_evolution_defn} above. %we can consider their {\it time-dependent correlation functions}. Letting $X_j \in C^\infty(\mfh)$ and $t_j \in \R$ for $ j \leq m \in \N^*$, we define the  
Given $m \in \N^*$, we consider the $m$-particle time-dependent correlation functions associated with $H$. More precisely, given $X_j \in C^\infty(\mfh)$ and $t_j \in \R$ for $j=1,\ldots,m$, we define 
$\mathbb{Q}_{\mathrm{Gibbs}} \equiv \mathbb{Q}_{\mathrm{Gibbs}}^f$ as
\begin{equation}
\label{q_time_correlation_defn}
\mathbb{Q}_{\mathrm{Gibbs}} (X_1,\ldots, X_m,t_1,\ldots,t_m) := \int \dd \mbP_{\mathrm{Gibbs}}^f(u) \prod^m_{j=1} X_j(S_{t_j}u)\,.
\end{equation} 

\subsubsection*{The quantum model}
We can view \eqref{q_Hartree_defn} as the {\it classical limit} of many-body quantum dynamics. At this point, according to our convention, we take $w$ as in Assumption \ref{w_assumption_bounded}. Given $n \in \N$, we consider the $n$-body Hamiltonian given by
\begin{equation}
\label{q_n_body_Hamn}
H^{(n)} := \sum_{i=1}^n (-\Delta_i + \kappa) -\frac{\lambda}{3} \mathop{\sum_{i,j,k}^{n}}_{i \neq j \neq k \neq i} w(x_i-x_j)\,w(x_i-x_k)\,,
\end{equation}
which is defined on the {\it bosonic $n$-particle space} 
\begin{equation}
\label{bosonic_n-particle_space}
\mfh^{(n)}:=\bigl\{u \in \mfh^n,\, u(x_1,\ldots,x_n) = u(x_{\pi(1)}, \ldots, x_{\pi(n)})
\,\,\forall (x_1,\ldots,x_n) \in \Lambda^n\,\,\forall \pi \in S_n\bigr\}\,.
\end{equation}
Here $\Delta_i$ denotes the action of  $\Delta$ in the $i^{th}$ component $x_i$.
In order for the two terms in \eqref{q_n_body_Hamn} to have the same order, we henceforth consider the interaction strength $\lambda \sim \frac{1}{n^2}$. %Note that since $w$ is always assumed to be even, %the second term in the Hamiltonian \eqref{q_Hamn_defn} can be written as
%begin{multline*}
%=-\frac{1}{9} \int \dd y \, \dd z \, \Bigl(w(x-y)\,w(y-z)+w(y-z)\,w(z-x)+w(z-x)\,w(x-z)\Bigr)
%\\
%\times
%|u(x)|^2\,|u(y)|^2\,|u(z)|^2
%\end{multline*}
%without the assumption that $w$ is even (as in Assumption \ref{w_assumption}).
%In this case, 
Note that the second term (i.e.\ the interaction term) in \eqref{q_n_body_Hamn} can be written as 
\begin{equation}
\label{symmetry_Fock_space}
-\frac{\lambda}{3} \sum_{1\leq i<j<k\leq n} \Bigl( w(x_i-x_j)\,w(x_i-x_k)+ w(x_j-x_k)\,w(x_j-x_i)+w(x_k-x_i)\,w(x_k-x_j) \Bigr)\,.
\end{equation}
As a multiplication operator, \eqref{symmetry_Fock_space} leaves $\mfh^{(n)}$ invariant.

The $n$-body Schr\"odinger equation corresponding to \eqref{q_n_body_Hamn} is defined as
\begin{equation}
\label{q_n-body_Schro}
\mathrm{i} \partial_{t} \Psi_{n,t} = H^{(n)}\Psi_{n,t}\,,
\end{equation}
where $H^{(n)}$ is as in \eqref{q_n_body_Hamn}. One can in general compare the dynamics of \eqref{q_n-body_Schro} with the dynamics of \eqref{q_Hartree_defn} for suitably chosen initial conditions as $n \to \infty$. The first rigorous results in this direction (for different, but closely related models) were proved in \cite{GV79,Hep74,Spo80}. For (suitably) factorised initial data $\Psi_{n,0} \sim u_0^{\otimes n}$ in the appropriate $n$-body Schr\"{o}dinger equation, it is shown that the factorisation property persists in the form $\Psi_{n,t} \sim u_t^{\otimes n}$, where $u_t$ solves solves the appropriate NLS with initial data $u_0$. This type of result sometimes goes under the name of \emph{propagation of chaos}. We do not address this question in our paper, although this heuristic provides a useful guiding principle for our further analysis. For further details and more recent results on related problems, we refer the reader to the expository work \cite{Schlein_2008}.

At temperature $\tau>0$, the \emph{(normalised) canonical Gibbs state} associated with \eqref{q_n-body_Schro} is given by
\begin{equation*}
P_\tau^{(n)} \equiv P_\tau^{f,(n)} := \frac{\e^{-H^{(n)}/\tau}f\left(\frac{n}{\tau}\right)}{Z^{(n)}_{\tau}}\,,\qquad Z^{(n)}_{\tau} \equiv Z^{f,(n)}_{\tau}:=\mathrm{Tr}_{\mfh^{(n)}} \Bigl(\e^{-H^{(n)}/\tau}f\left(\frac{n}{\tau}\right)\Bigr)\,.
\end{equation*}
This is an equilibrium state of \eqref{q_n-body_Schro}.
In our setting, we will vary the particle number $n$. Hence, we work in the \emph{grand canonical ensemble} and consider the \emph{(normalised) grand canonical Gibbs state} given by
\begin{equation}
\label{normalised_grand_canonical_Gibbs_state}
P_\tau^{\mathrm{norm}} \equiv P_\tau^{f,\mathrm{norm}}
:= \frac{1}{Z_{\tau}}\,\bigoplus_{n=0}^{\infty} \e^{-H^{(n)}/\tau}f\left(\frac{n}{\tau}\right)\,,
\end{equation}
where
\begin{equation*}
Z_{\tau} \equiv Z^{f}_{\tau}:=
\sum_{n=0}^{\infty} \mathrm{Tr}_{\mfh^{(n)}} \Bigl(\e^{-H^{(n)}/\tau}f\left(\frac{n}{\tau}\right)\Bigr)\,.
\end{equation*}
In this paper, we show that the grand canonical Gibbs state \eqref{normalised_grand_canonical_Gibbs_state} converges to the classical Gibbs state $\E_{\mbP^f_{\mathrm{Gibbs}}}(\cdot)$ (see \eqref{q_classical_state_defn} below) as $\tau \rightarrow \infty$ in the sense of correlation functions and partition function. Furthermore, we show time-dependent variants of this result. 
The precise results are stated in Section \ref{Results} below. We emphasise that the main novelty of our work consists in analysing three-body interactions in the quantum setting, which in turn lead to a quintic nonlinearity in the classical setting.
All of the previous works on this problem have focused on two-body interactions in the quantum setting which lead to a quartic nonlinearity in the classical setting. We recall the previously known results in detail in Section \ref{Previously known results} below. In what follows, we explain the setup more precisely.

\subsubsection*{Notation and Conventions}
Let us define the notation and give the conventions that we will use in the sequel.
We write $\N = \{0,1,2,\ldots\}$ and $\N^* = \{1,2,\ldots\}$. Throughout the paper, $C$ is used to denote a generic positive constant that can change from line to line. If $C$ depends on a set of parameters $\{a_1,\ldots,a_n\}$, we write $C(a_1,\ldots,a_n)$. We sometimes write $a \lesssim b$ to denote $a \leq Cb$ for some $C>0$. Moreover, we write $a \lesssim_{a_{1},\ldots,a_{n}}b $ if $a \leq C(a_1,\ldots,a_n) b$ for some constant $C(a_1,\ldots,a_n)>0$ depending on the parameters $a_1,\ldots,a_n$.
We denote by $\chi_A$ the characteristic function of the set $A$, which is defined by

\begin{equation*}
\chi_A(x) := 
\begin{cases}
0 & \text{if } x \notin A \\
1 & \text{if } x \in A.
\end{cases}
\end{equation*}

Consider a separable Hilbert space $\mathcal{H}$.
We denote by $\mathbf{1}$ the identity operator on $\mathcal{H}$. Furthermore, $\mathcal{L}(\mathcal{H})$ denotes the space of bounded linear operators on $\mathcal{H}$. For $q \in [1,\infty]$, the Schatten space $\mfS^q(\mathcal{H})$ is defined to be the set of $\mcA \in \mathcal{L}(\mathcal{H})$ with $\|\mathcal{A}\|_{\mfS^q(\mathcal{H})} < \infty$, where
\begin{equation*}
\|\mathcal{A}\|_{\mfS^q(\mathcal{H})} := \begin{cases}
\bigl(\Tr_{\mathcal{H}} |\mcA|^q\bigr)^{1/q} & \text{if } q< \infty \\
\sup \,\rm{spec}\, |\mathcal{A}| & \text{if } q = \infty\,.
\end{cases}
\end{equation*}
Here $|\mcA| := \sqrt{\mathcal{A}^*\mcA}$ and $\rm{spec}$ denotes the spectrum of an operator. For $p \in \N$, we also define
\begin{align}
	\label{q_mathfrakB_defn}
\mathfrak{B}_p &:= \{\xi \in \mfS^2(\mfhp) : \|\xi\|_{\mathfrak{S}^2(\mfhp)} \leq 1\}\,, \\
\label{q_mathcalC_defn}
\mcCp &:= \mathfrak{B}_p \cup \mathbf{1}_p\,,
\end{align}
where $\mathbf{1}_p$ is the identity operator on $\mfhp$. Here, we recall \eqref{bosonic_n-particle_space}.

If $\xi$ is a closed linear operator on $\mfhp$, we can identify it with its Schwartz kernel, which we write as $\xi(x_1,\ldots,x_p;y_1,\ldots,y_p)$; see for example \cite[Corollary V.4.4]{RS80}.

\begin{remark}
\label{Setup_1}
Our analysis would apply if we considered the nonlinearity of the opposite sign in \eqref{q_Hartree_defn}. In this case, $w=\delta$ would yield the defocusing quintic NLS, instead of \eqref{focusing_quintic_NLS}. Our methods require a truncation of the mass in the classical problem. This is a natural condition to define the Gibbs measure in the focusing regime, as is explained in Section \ref{Gibbs measures and the classical system} below. Such a truncation is unnatural in the defocusing regime. We conjecture that analogues of our main results hold in the defocusing regime without a truncation. 
\end{remark}

\subsection{Gibbs measures and the classical system}
\label{Gibbs measures and the classical system}
For $k \in \N$, we take $\mu_k$ to be independent standard complex Gaussian measures. Namely, we take $\mu_k := \frac{1}{\pi}\,\e^{-|z|^2}\dd z$, where $\dd z$ denotes the standard Lebesgue measure on $\C$. We consider the product probability space $(\C^{\mathbb{N}}, \mc{G}, \mu)$, where
\begin{equation}
	\label{q_rigorous_gibbs_defn}
\mu := \bigotimes_{k \in \mathbb{N}} \mu_k\,.
\end{equation}
We write $\omega = (\omega_k)_{k \in \mathbb{N}}$ for elements of the probability space $\C^{\N}$. The {\it classical free field} $\vph \equiv \vph^{\om}$ is defined as
\begin{equation}
\label{q_class_free_field}
\vph := \sum_{k \in \mathbb{N}} \frac{\om_k}{\sqrt{\lk}}\, e_k\,.
\end{equation}
Here, we recall \eqref{lambda_k}.
Then $\eqref{q_trace_finite}$ implies that
\begin{equation}
	\label{q_regularity_vph}
\vph \in H^{\frac{1}{2}-\varepsilon}(\Lambda)\,,
\end{equation}
$\mu$-almost surely for $\eps > 0$. 
Letting $\vph(g) := \langle g,\vph\rangle$ and $\overline{\vph}(g) := \langle \vph,g\rangle$, we have
\begin{equation}
	\label{q_covariance_ Wiener}
\mathbb{E}_\mu \bigl( \overline{\vph}(g) \vph(\tilde{g})\bigr) = \langle \tilde{g},h^{-1}g \rangle, \quad \mathbb{E}_\mu\bigl( \vph(g) \vph(\tilde{g})\bigr) = \mathbb{E}_\mu \bigl(\overline{\vph}(g) \overline{\vph}(\tilde{g})\bigr)= 0 
\end{equation} 
for any $g,\tilde{g} \in H^{-\frac{1}{2} + \eps}$. Here the Green function $h^{-1}$ is the covariance of $\mu$. Then $\mu$ satisfies the following form of Wick's theorem; see \cite[Lemma 2.4]{FKSS20} for a self-contained summary.
\begin{proposition}[Wick's theorem]
	\label{q_Wick_thm}
Suppose $\vph$ is defined as in \eqref{q_class_free_field}. Let $n \in \mathbb{N}^*$ and $g_i \in H^{-\frac{1}{2}+\varepsilon}$ for $i \in \{1,\ldots,n\}$ be given. Denote by $\vph^\#$ either $\vph$ or $\overline{\vph}$. Then, we have
\begin{equation}
	\label{q_Wick_thm_computation}
\mathbb{E}_\mu \left(\prod_{i=1}^n \vph^{\#_i}(g_i)\right) = \sum_{\Pi \in M(n)} \prod_{(j,k) \in \Pi} \mathbb{E}_\mu \left(\vph^{\#_j}(g_j)\vph^{\#_k}(g_k)\right)\,.
\end{equation}
Here, $M(n)$ is the set of complete pairings of $\{1,\ldots,n\}$, and $(j,k)$ is the set of edges in $\Pi$.
\end{proposition}
In particular, using Proposition \ref{q_Wick_thm}, we can simplify \eqref{q_Wick_thm_computation} using the identities in \eqref{q_covariance_ Wiener}, keeping only the terms terms which contain a $(\overline{\vph},\vph)$ pair.  As in \cite[Remark 1.3]{FKSS17}, using a suitable pushforward one can also view $\mu$ as a probability measure on $H^s$. 
More precisely, it is the complex Gaussian measure on $H^s$ (for $s<1/2$) with covariance $h^{-1}$. We also refer to $\mu$ as the \emph{Wiener measure}.
%We do not use this form and instead work with the form of the measure $\mu$ as defined above.

The {\it mass} $\mathcal{N} \equiv \mathcal{N}^{\omega}$ of the classical free field \eqref{q_class_free_field} is defined to be
\begin{equation}
\label{q_mass}
\mc{N} := \|\vph\|_{\mfh}^2\,.
\end{equation}
With $w$ as in Assumption \ref{w_assumption}, the \emph{classical interaction} $\mathcal{W} \equiv \mathcal{W}^{\omega}$ is given by
\begin{multline}
\label{q_classical_interaction}
\mc{W} := -\frac{1}{3} \int \dd x \, \dd y \, \dd z \, w(x-y)\,w(x-z)\, |\vph(x)|^2\,|\vph(y)|^2\,|\vph(z)|^2
\\
=-\frac{1}{3} \int \dd x\, \bigl(w*|\vph|^2\bigr)^2(x)\,|\vph(x)|^2\,.
\end{multline}
Using Sobolev embedding, we get that $\varphi \in H^{\frac{1}{3}} \subset L^6$ $\mu$-almost surely. Therefore, by Lemma \ref{Multlinear_estimates_1} below, it follows that $\mc{W}$ is finite $\mu$-almost surely.
We also define the \emph{classical free Hamiltonian} $H_0 \equiv H_0^{\omega}$ as
\begin{equation}
\label{q_classical_free_Hamn}
H_0 := \int \dd x \, \dd y \,  \overline{\vph}(x)\,h(x;y)\,\vph(y)\,,
\end{equation}
where $h(x;y)$ is the kernel corresponding to \eqref{q_1-body_Hamn_defn}. The {\it classical interacting Hamiltonian} $H \equiv H^{\omega}$ is given by
\begin{equation}
\label{q_classical_interacting_Hamn}
H := H_0 + \mc{W}\,.
\end{equation}

\begin{assumption}
\label{f_assumption}
For the remainder of the paper, we fix $f \in C^\infty_0(\R)$ a cut-off function, which is not identically zero, satisfying $0 \leq f \leq 1$, and
\begin{equation}
	\label{q_cut-off_support}
f(s)=0 \,\,\, \text{for}  \,\,\, s>K\,,
\end{equation}
where $K>0$ is a sufficiently small positive constant.
\end{assumption}

\begin{remark}
\label{f_assumption_remark}
The choice of $K$ is dictated by Proposition \ref{Cauchy_problem_2} (i) below. Throughout, our estimates will depend on the $K$ in \eqref{q_cut-off_support}, but we will not keep explicit track of this dependence.
\end{remark}

With $\mu$ as in \eqref{q_rigorous_gibbs_defn}, $\mathcal{N}$ as in \eqref{q_mass}, $\mathcal{W}$ as in \eqref{q_classical_interaction}, and $f$ as in Assumption \ref{f_assumption}, we now define the \emph{classical Gibbs measure}
\begin{equation}
\label{q_gibbs_cutoff_defn_rigorous}
\mbP^f_{\mathrm{Gibbs}} := \frac{1}{z_{\mathrm{Gibbs}}^f} \,\e^{-\mathcal{W}}\,f(\mathcal{N})\,\mu\,,
\end{equation}
where the {\it classical partition function} $z \equiv z_{\mathrm{Gibbs}}^f$ is the normalisation constant
\begin{equation}
\label{q_classical_partn_functn}
z:= \int \dd \mu \, \e^{-\mc{W}} f(\mc{N})\,.
\end{equation}
We note that the probability measure $\mbP^f_{\mathrm{Gibbs}}$ in \eqref{q_gibbs_cutoff_defn_rigorous} and the classical partition function \eqref{q_classical_state_defn} are well-defined for sufficiently small $K$ in \eqref{q_cut-off_support} by Proposition \ref{Cauchy_problem_2} (i) below. In particular, this gives us a rigorous construction of \eqref{q_gibbs_cutoff_defn} above.

For a random variable $X \equiv X^\om $, the \emph{classical state} $\rho \equiv \rho^f$ is given by
\begin{equation}
\label{q_classical_state_defn}
\rho(X) := \E_{\mbP^f_{\mathrm{Gibbs}}}(X) = \frac{\int \dd\mu \, X \e^{-\mc{W}}f(\mc{N})}{\int \dd\mu \, \e^{-\mc{W}}f(\mc{N})}\,.
\end{equation}
 %The classical state $\rho$ can be characterised completely by its moments, so 
We define the \emph{classical $p$-particle correlation function} $\gamma_p \equiv \gamma_p^f$ as the operator on $\mfhp$ with kernel given by
\begin{equation}
	\label{q_classical_correln_fn_defn}
\gamma_p(x_1,\ldots,x_p;y_1,\ldots,y_p) = \rho(\bar{\vph}(y_1) \ldots \overline{\vph}(y_p) \vph(x_1)\ldots \vph(x_p))\,.
\end{equation}
For a closed densely-defined linear operator $\xi$ on $\mfhp$, we will consider the random variable $\Th(\xi) \equiv \Th^\om(\xi)$ defined as
\begin{multline}
	\label{q_random_var_defn}
\Th(\xi) := \int \dd x_1 \ldots \dd x_p \, \dd y_1 \ldots \dd y_p \, \xi(x_1,\ldots,x_p;y_1,\ldots,y_p) \\
\times \overline{\vph}(x_1) \ldots \overline{\vph}(x_p) \vph(y_1) \ldots \vph(y_p)\,.
\end{multline}
By writing \eqref{q_classical_interaction} as
\begin{multline}
\label{q_classical_interaction_symmetric}
-\frac{1}{9} \int \dd x\, \dd y \, \dd z \, \Bigl(w(x-y)\,w(x-z)+w(y-z)\,w(y-x)+w(z-x)\,w(z-y)\Bigr)
\\
\times
|\vph(x)|^2\,|\vph(y)|^2\,|\vph(z)|^2\,,
\end{multline}
it follows that we can write \eqref{q_classical_interacting_Hamn} as
\begin{equation}
\label{H_random_variable}
H = \Th(h) +\frac{1}{3} \,\Th(W)\,,
\end{equation}
where $W$ is the operator on $\mfh^{(3)}$ and operator kernel
\begin{multline}
\label{W_kernel}
W(x_1,x_2,x_3;y_1,y_2,y_3)
\\
:=-\frac{1}{3} \Bigl(w(x_1-x_2)\,w(x_1-x_3)+w(x_2-x_3)\,w(x_2-x_1)+w(x_3-x_1)\,w(x_3-x_2)\Bigr)\,
\\
\times \delta(x_1-y_1)\,\delta(x_2-y_2)\,\delta(x_3-y_3)\,.
\end{multline}
In particular, $W$ acts as multiplication by 
\begin{equation*}
-\frac{1}{3} \Bigl(w(x_1-x_2)\,w(x_1-x_3)+w(x_2-x_3)\,w(x_2-x_1)+w(x_3-x_1)\,w(x_3-x_2)\Bigr)\,.
\end{equation*}
Note that $W$ indeed acts on $\mathfrak{h}^{(3)}$ as densely-defined linear operator\footnote{
In order to write $W$ as an operator on $\mfh^{(3)}$, we needed to apply the symmetrisation step in \eqref{q_classical_interaction_symmetric} above. Moreover, we note that $W$ is bounded if $w \in L^{\infty}$ as in Assumption \ref{w_assumption_bounded}.}. As was noted earlier, $\Theta(W)=3 \mathcal{W}$ is finite $\mu$-almost surely. %This follows from Lemma \ref{Multlinear_estimates_1} below, as well as from the fact that $\varphi \in H^{\frac{1}{3}}$, $\mu$-almost surely.

With $S_t$ defined as in \eqref{q_evolution_defn} and $\xi$ a closed densely-defined linear operator on $\mfhp$, we define for $t \in \R$ the random variable
\begin{multline}
\label{q_time_evolved_rv}
\Psi^t\Th(\xi) := \int \dd x_1 \ldots \dd x_p \, \dd y_1 \ldots \dd y_p \, \xi(x_1,\ldots,x_p;y_1,\ldots,y_p) \\
\times \overline{S_t\vph}(x_1) \ldots \overline{S_t\vph}(x_p) S_t\vph(y_1) \ldots S_t\vph(y_p)\,,
\end{multline}
which corresponds to the time evolution of \eqref{q_random_var_defn}.
\subsection{The quantum system}
We work with the {\it bosonic Fock space}
\begin{equation*}
\mc{F} \equiv \mc{F}(\mfh) := \bigoplus_{p \in \mathbb{N}} \mfhp\,.
\end{equation*}
Vectors in $\mc{F}$ are denoted by $\Phi = (\Phi^{(p)})_{p \in \N}$. For $g \in \mfh$, we denote by $b^*(g)$ and $b(g)$ the  {\it creation and annihilation operators} on $\mc{F}$, given by
\begin{align}
\label{Operator_b^*}
\left(b^*(g)\Phi\right)^{(p)}(x_1,\ldots,x_p) &:= \frac{1}{\sqrt{p}} \sum_{j=1}^p g(x_j) \Phi^{(p-1)}(x_1,\ldots,x_{i-1},x_{i+1}\,,\ldots,x_p), \\
\label{Operator_b}
\left(b(g)\Phi\right)^{(p)}(x_1,\ldots,x_p) &:= \sqrt{p+1} \int \dd x \, \overline{g}(x) \Phi^{(p+1)}(x,x_1,\ldots,x_p)\,.
\end{align}
These operators are closed, densely-defined, and adjoints of each other. They also satisfy the {\it canonical commutation relations}, meaning that for all $g,\tilde{g} \in \mfh$, we have
\begin{equation}
\label{CCR}
[b(g),b^*(\tilde{g})] = \langle g,\tilde{g}\rangle_{\mfh}, \quad [b(g),b(\tilde{g})] = [b^*(g),b^*(\tilde{g})] = 0\,,
\end{equation}
where $[X,Y]=XY-YX$ denotes the commutator.
In what follows, we work with the {\it rescaled creation and annihilation operators}. For $g \in \mfh$, we define 
\begin{equation}
\label{varphi_tau^sharp}
\vph^*_\tau(g) := \tau^{-1/2}\,b^*(g)\,,\qquad \vph_\tau(g) := \tau^{-1/2}\,b(g)\,.
\end{equation}
We think of $\vph^*_\tau$ and $\vph_\tau$ as operator-valued distributions and denote their distribution kernels by 
\begin{equation}
\label{varphi_tau^sharp_2}
\vph^*_\tau(x) := \vph^*_\tau(\delta_x)\,,\qquad \vph_\tau(x) := \vph_\tau(\delta_x)\,.
\end{equation}
%$\vph^*_\tau(x) := \vph^*_\tau(\delta_x)$ and $\vph_\tau(x) := \vph_\tau(\delta_x)$ 
By a formal calculation from \eqref{CCR}, we obtain that for $x,y \in \Lambda$ 
\begin{equation}
\label{CCR_formal}
[\varphi_{\tau}(x),\varphi^{*}_{\tau}(y)]=\frac{1}{\tau}\,\delta(x-y)\,,\qquad [\varphi_{\tau}(x),\varphi_{\tau}(y)]=[\varphi^{*}_{\tau}(x),\varphi^{*}_{\tau}(y)]=0\,.
\end{equation}
%The limit $\tau \to \infty$ can be considered a high-density limit, we direct the reader to \cite[Section 1.1]{FKSS17} for a more detailed discussion on this interpretation. 
By analogy with the classical field $\vph$ defined in \eqref{q_class_free_field}, we call $\vph_\tau$ the {\it quantum field}.
By letting $\tau \rightarrow \infty$ in the formal identity \eqref{CCR_formal}, the quantum fields formally commute in this limit.

For $w$ as in Assumption \ref{w_assumption_bounded}, the {\it quantum interaction} is defined as 
\begin{multline}
\label{Quantum_interaction}
\mc{W}_\tau := -\frac{1}{3} \int \dd x \, \dd y \, \dd z \, w(x-y)\,w(x-z)\,\vphts(x)\vphts(y)\vphts(z)\vpht(x)\vpht(y)\vpht(z)
\\
=-\frac{1}{9}  \int \dd x \, \dd y \, \dd z \Bigl(w(x-y)\,w(x-z)+w(y-z)\,w(y-x)+w(z-x)\,w(z-y)\Bigr)
\\
\times \vphts(x)\vphts(y)\vphts(z)\vpht(x)\vpht(y)\vpht(z)
\,.
\end{multline}
For $h$ as in \eqref{q_1-body_Hamn_defn} the {\it quantum free Hamiltonian} is defined as
\begin{equation}
\label{Quantum_free_Hamiltonian}
\Hf := \int \dd x \, \dd y \,  \vphts(x)\,h(x,y)\, \vpht(y)\,.
\end{equation}
Then the {\it quantum interacting Hamiltonian} is given by
\begin{equation}
	\label{q_quantum_interacting_Hamn}
H_\tau := \Hf + \mc{W}_\tau\,.
\end{equation}
%On the $n^{th}$ sector of Fock space, the action of $H_\tau$ agrees with \eqref{q_n_body_Hamn}. 
By using \eqref{Operator_b^*}--\eqref{Operator_b}, \eqref{varphi_tau^sharp}--\eqref{varphi_tau^sharp_2}, and \eqref{Quantum_interaction}--\eqref{q_quantum_interacting_Hamn}, we have
\begin{equation*}
H_{\tau}=\bigoplus_{n=0}^{\infty} H_{\tau}^{(n)}\,,
\end{equation*}
where 
\begin{equation}
\label{H^n_tau}
H^{(n)}_{\tau} =
\frac{1}{\tau}\,\sum_{i=1}^n (-\Delta_i + \kappa) - \frac{1}{3\tau^3} \mathop{\sum_{i,j,k}^{n}}_{i \neq j \neq k \neq i} w(x_i-x_j)\,w(x_i-x_k)\,.
\end{equation}
Note that $H^{(n)}_{\tau}=\frac{1}{\tau} H^{(n)}$ for $H^{(n)}$ as in \eqref{q_n_body_Hamn} with $\lambda=\frac{1}{\tau^2}$.

We define the \emph{rescaled particle number} as
\begin{equation}
\label{N_tau}
\mc{N}_{\tau} := \int \dd x \, \vphts(x)\,\vpht(x)\,.
\end{equation}
By \eqref{Operator_b^*}--\eqref{Operator_b} and \eqref{varphi_tau^sharp}--\eqref{varphi_tau^sharp_2}, it follows that \eqref{N_tau}
acts on the $n^{th}$ sector of Fock space as multiplication by $\frac{n}{\tau}$. The \emph{(untruncated and unnormalised) grand canonical ensemble} is given by 
\begin{equation*}
P_\tau := \e^{-H_\tau}=\bigoplus_{n=0}^{\infty} \e^{-H^{(n)}_{\tau}}\,,
\end{equation*}
with $H^{(n)}_{\tau}$ as in \eqref{H^n_tau}.
For a closed operator $\mc{A}: \mc{F} \to \mc{F}$, the \emph{quantum state} $\rho_\tau \equiv \rt^f$ is defined as
\begin{equation}
\label{q_quantum_state_defn}
\rt(\mc{A}) := \frac{\Tr_{\mc{F}}\left(\mc{A}P_\tau f(\mc{N}_\tau)\right)}{\Tr_{\mc{F}}\left(P_\tau f(\mc{N}_\tau)\right)}\,.
\end{equation}
We define the {\it quantum partition function} and {\it quantum free partition function} $Z_\tau \equiv Z_\tau^f$ and $\Zf$ respectively as
\begin{equation}
	\label{q_quantum_partition_functions}
Z_\tau := \Tr\left(P_\tau f(\mc{N}_\tau)\right), \qquad \Zf := \Tr\left(\e^{-\Hf} \right),
\end{equation}
and the {\it quantum relative partition function} $\mathcal{Z}_\tau \equiv \mathcal{Z}_\tau^f$ as
\begin{equation}
\label{q_relative quantum_partition_fn}
\mc{Z}_\tau := \frac{Z_\tau}{\Zf}\,.
\end{equation}
%Like in the classical case, we can characterise the quantum state through its correlation functions. 
For $p \in \N^*$ we define the \emph{$p$-particle quantum correlation function} $\gamma_{\tau,p} \equiv \gamma_{\tau,p}^f$ as operator which acts on $\mfhp$, with kernel given by
\begin{equation}
\label{q_quantum_correlation_fn_defn}
\gamma_{\tau,p}(x_1,\ldots, x_p, y_1 , \ldots, y_p) := \rt(\vphts(y_1)\ldots\vphts(y_p)\vpht(x_1)\ldots\vpht(x_p))\,.
\end{equation}
Note that \eqref{q_quantum_correlation_fn_defn} is a quantum analogue of \eqref{q_classical_correln_fn_defn}. 

By analogy with \eqref{q_random_var_defn}, for a closed linear operator $\xi \in \mc{L}(\mfhp)$, we define the lift of $\xi$ to Fock space as
\begin{multline}
	\label{q_quantum_lift}
\Th_\tau(\xi) := \int \dd x_1  \ldots \dd x_p \, \dd y_1 \ldots \dd y_p \, \xi(x_1,\ldots,x_p;y_1,\ldots,y_p) \\ 
\times \vphts(x_1)\ldots\vphts(x_p)\vpht(y_1)\ldots\vpht(y_p)\,.
\end{multline}
Analogously to \eqref{H_random_variable}, in the quantum setting we have
\begin{equation*}
H_\tau = \Th_\tau(h) + \frac{1}{3} \,\Th_\tau(W)\,,
\end{equation*}
where $W$ is the $3$-body operator with kernel \eqref{W_kernel}. We note that, by Assumption \ref{w_assumption_bounded}, $\Theta_{\tau}(W)$ is a bounded operator on Fock space. 
For an operator $\mc{A}:\mc{F} \to \mc{F}$, we define the {\it quantum time evolution}
\begin{equation}
\label{q_quantum_time_evolution}
\Psi^t_\tau \mc{A} := \e^{\mathrm{i} t\tau H_\tau} \mc{A}\, \e^{-\mathrm{i}t\tau H_\tau}\,.
\end{equation}
In this paper, when working with an object $X = \rho, \gamma_p, \Th, \ldots$, we use the notation $X_{\#}$ to denote either the classical object $X$ or quantum object $X_{\tau}$.
\subsection{Results}
\label{Results}

We now state the main results of our work. In Section \ref{Time-independent results}, we state the results on the time-independent problem. In Section \ref{Time-dependent results}, we state the results on time-dependent correlations. Throughout, we fix a cut-off function $f \in C_0^{\infty}(\R)$ as in Assumption \ref{f_assumption} above.

\subsubsection{Time-independent results}
\label{Time-independent results}

We first analyse bounded interaction potentials (as in Assumption \ref{w_assumption_bounded}).
\begin{theorem}[Convergence for bounded interaction potentials]
	\label{q_bounded_state_convergence_thm}
Let $w$ be as in Assumption \ref{w_assumption_bounded}. Let $p \in \mathbb{N}^*$ be given. Consider $\gamma_{p}$ and $\gamma_{\tau,p}$ defined as in \eqref{q_classical_correln_fn_defn} and \eqref{q_quantum_correlation_fn_defn} respectively. Then, we have
\begin{equation}
\label{correlation_function_convergence_bounded}
\lim_{\tau \to \infty}\|\gamma_{\tau,p} - \gamma_p\|_{\mfS^1(\mfhp)} = 0\,.
\end{equation}
Moreover, for $z$ and $\mc{Z}_{\tau}$ defined as in \eqref{q_classical_partn_functn} and \eqref{q_relative quantum_partition_fn} respectively, we have
\begin{equation}
\label{partition_function_convergence_bounded}
\lim_{\tau \to \infty}\mc{Z}_{\tau} = z\,.
\end{equation}

\end{theorem}
To obtain a result for $w$ as in Assumption \ref{w_assumption}, we use an approximation argument. For a suitable approximation  $w^\eps$ to $w$, as defined below, any object with a superscript $\eps$ will be the corresponding object defined using $w^\eps$ instead of $w$. We first state the result for interactions as in Assumption \ref{w_assumption} (i).
\begin{theorem}[Convergence for $L^1$ interaction potentials]
	\label{q_L1_state_convergence_thm}
Let $w$ be as in Assumption \ref{w_assumption} (i). Suppose that $w^\eps$ is a sequence of interaction potentials as in Assumption \ref{w_assumption_bounded} converging to $w$ in $L^1$. %Let $\gamma_{p}$, $\gamma^\eps_{\tau,p}$, $z$, and $\mc{Z}^\eps_{\tau}$ be defined as in \eqref{q_classical_correln_fn_defn}, \eqref{q_quantum_correlation_fn_defn}, \eqref{q_classical_partn_functn}, and \eqref{q_relative quantum_partition_fn} respectively. 
With objects defined analogously as for Theorem \ref{q_bounded_state_convergence_thm}, there exists a sequence $\eps_\tau \rightarrow 0$ as $\tau \to \infty$ such that for any $p \in \mathbb{N}^*$, we have
\begin{equation}
\label{correlation_function_convergence_L1}
	\lim_{\tau \to \infty}\|\gamma^{\varepsilon_\tau}_{\tau,p} - \gamma_p\|_{\mfS^1(\mfhp)} = 0\,,
\end{equation}
and such that
\begin{equation}
\label{partition_function_convergence_L1}
	\lim_{\tau \to \infty} \mc{Z}^{\varepsilon_\tau}_{\tau} = z\,.
\end{equation}
\end{theorem}

We prove an analogous result for the local problem, i.e.\ for $w=c\delta$ be as in Assumption \ref{w_assumption} (ii). In order to state the result, we define $w^{\varepsilon}$. Let $\Phi :\R \rightarrow \R$ be a continuous even function with compact support such that  %$\mathrm{supp}\,\Phi \subset \T$ such that
\begin{equation}
\label{int_Phi}
\int_{\R} \dd x\,\Phi(x)=c\,.
\end{equation}
For $\varepsilon \in (0,1)$, we define $w^{\varepsilon}: \T \rightarrow \R$ as
\begin{equation}
\label{w_epsilon_local_problem_definition}
w^{\varepsilon}(x):=\frac{1}{\varepsilon}\,\Phi\biggl(\frac{[x]}{\varepsilon}\biggr)\,,
\end{equation}
where $[x]$ denotes the unique element of $(x+\mathbb{Z}) \cap \mathbb{T}$. The function $w^{\varepsilon}$ defined as in \eqref{w_epsilon_local_problem_definition} is even and belongs to $L^{\infty}(\T)$.  Moreover, $w^{\varepsilon}$ converges to $c \delta$ weakly, with respect to continuous functions, i.e.\
\begin{equation*}
\int_{\T}\dd x\, w^{\varepsilon}(x)\,g(x) \rightarrow c g(0)\,\quad \text{as }\, \varepsilon \rightarrow 0
\end{equation*}
for all $g: \T \rightarrow \R$ continuous.

\begin{theorem}[Convergence for delta function interaction potentials]
\label{delta_function_state_convergence_thm}
Let $w$ be as in Assumption \ref{w_assumption} (ii). Let $w^{\varepsilon}$ be defined as in \eqref{w_epsilon_local_problem_definition} above.
%Let $\gamma_{p}$, $\gamma^\eps_{\tau,p}$, $z$, and $\mc{Z}^\eps_{\tau}$ be defined as in \eqref{q_classical_correln_fn_defn}, \eqref{q_quantum_correlation_fn_defn}, \eqref{q_classical_partn_functn}, and \eqref{q_relative quantum_partition_fn} respectively. 
%With objects defined analogously as for Theorem \ref{q_bounded_state_convergence_thm}, 
There exists a sequence $\eps_\tau \rightarrow 0$ as $\tau \to \infty$ such that for any $p \in \mathbb{N}^*$, we have
\begin{equation}
\label{correlation_function_convergence_delta}
	\lim_{\tau \to \infty}\|\gamma^{\varepsilon_\tau}_{\tau,p} - \gamma_p\|_{\mfS^1(\mfhp)} = 0\,,
\end{equation}
and such that
\begin{equation}
\label{partition_function_convergence_delta}
	\lim_{\tau \to \infty} \mc{Z}^{\varepsilon_\tau}_{\tau} = z\,.
\end{equation}
\end{theorem}

\begin{remark}
In \cite[Appendix B]{RS22}, it was proved that one can obtain an analogue of \cite[Theorem 1.4]{RS22} with a cut-off function of the form of $f(x)=e^{-c|x|^2}$. This is not possible in the quintic case because there is no known corresponding analogue of Theorem \ref{q_bounded_state_convergence_thm} for non-negative interaction potentials without a cut-off function.
\end{remark}

\subsubsection{Time-dependent results}
\label{Time-dependent results}
We now state our time-dependent results.
Let us recall the definitions of $\rho$ and $\rt$ in \eqref{q_classical_state_defn} and \eqref{q_quantum_state_defn} respectively. Furthermore, we recall the definitions of $\Psi^t \Theta(\xi)$ and $\Psi^t_{\tau} \Theta_{\tau}(\xi)$ in \eqref{q_time_evolved_rv} and \eqref{q_quantum_lift}--\eqref{q_quantum_time_evolution}  respectively. We first state the result for bounded interaction potentials.
\begin{theorem}[Convergence for bounded potentials]
	\label{q_bounded_time_thm}
Let $w$ be as in Assumption \ref{w_assumption_bounded}. Let $m \in \mathbb{N}^*$, $p_1,\ldots,p_m \in \N^*$, $\xi_1 \in \mc{L}(\mathfrak{h}^{(p_1)}),\ldots,\xi_m \in \mc{L}(\mathfrak{h}^{(p_m)})$, and $t_1,\ldots,t_m \in \R$ be given. Then
\begin{equation}
\label{Theorem_1.12}
\lim_{\tau \to \infty}\rt(\Psi^{t_1}_\tau\Th_\tau(\xi_1) \ldots \Psi^{t_m}_\tau\Th_\tau(\xi_m)) = \rho(\Psi^{t_1}\Th(\xi_1) \ldots \Psi^{t_m}\Th(\xi_m))\,.
\end{equation}
\end{theorem}
As in the time-independent problem, we can use an approximation result to prove results for $w$ as in Assumption \ref{w_assumption}.

\begin{theorem}[Convergence for $L^1$ potentials]
\label{q_L1_time_thm}
Let $w$ be as in Assumption \ref{w_assumption} (i). Let $w^\eps$ be defined as in Theorem \ref{q_L1_state_convergence_thm} above. Then there is a sequence $\varepsilon_\tau \rightarrow 0$ as $\tau \to \infty$ such that, for all $m \in \mathbb{N}^*$, $p_1,\ldots,p_m \in \N^*$, $\xi_1 \in \mc{L}(\mathfrak{h}^{(p_1)}),\ldots,\xi_m \in \mc{L}(\mathfrak{h}^{(p_m)})$, and $t_1,\ldots,t_m \in \R$, we have
\begin{equation*}
\lim_{\tau \rightarrow \infty} \rt^{\eps_\tau}(\Psi^{t_1}_\tau\Th_\tau(\xi_1) \ldots \Psi^{t_m}_\tau\Th_\tau(\xi_m)) = \rho(\Psi^{t_1}\Th(\xi_1) \ldots \Psi^{t_m}\Th(\xi_m))\,.
\end{equation*}
\end{theorem}

\begin{theorem}[Convergence for delta function potentials]
\label{q_delta_time_thm}
Let $w$ be as in Assumption \ref{w_assumption} (i). Let $w^\eps$ be defined as in Theorem \ref{delta_function_state_convergence_thm} above. Then there is a sequence $\varepsilon_\tau \rightarrow 0$ as $\tau \to \infty$ such that, for all $m \in \mathbb{N}^*$, $p_1,\ldots,p_m \in \N^*$, $\xi_1 \in \mc{L}(\mathfrak{h}^{(p_1)}),\ldots,\xi_m \in \mc{L}(\mathfrak{h}^{(p_m)})$, and $t_1,\ldots,t_m \in \R$, we have
\begin{equation*}
\lim_{\tau \rightarrow \infty} \rt^{\eps_\tau}(\Psi^{t_1}_\tau\Th_\tau(\xi_1) \ldots \Psi^{t_m}_\tau\Th_\tau(\xi_m)) = \rho(\Psi^{t_1}\Th(\xi_1) \ldots \Psi^{t_m}\Th(\xi_m))\,.
\end{equation*}
\end{theorem}

We make the following remarks about the time-dependent results given in Theorems \ref{q_bounded_time_thm}--\ref{q_delta_time_thm} above.

\begin{remark}
\label{PDE_relevance_remark}
Theorems \ref{q_bounded_time_thm}--\ref{q_delta_time_thm} give the first microscopic derivation of time-dependent correlation functions for a quintic NLS.
From the PDE point of view, the study of time-dependent correlation functions is more relevant for quintic nonlinearities than for cubic ones. Namely, in the latter case, one can study the global well-posedness theory in $\mathfrak{h}$ without using Gibbs measures \cite{Bou93}. When studying the quintic problem, local well-posedness is known in $H^s(\mathbb{T})$ only for $s>0$. Therefore, one needs to use the invariance of the Gibbs measure as the substitute for a conservation law that allows us to obtain (almost sure) global solutions \cite{Bou94}.
\end{remark}
\begin{remark}
By arguing as in \cite[Remark 1.4]{FKSS18}, we can also recover the invariance of the Gibbs measure \eqref{q_gibbs_cutoff_defn_rigorous} for \eqref{q_Hartree_defn} from Theorem \ref{q_L1_time_thm}.
\end{remark}
\begin{remark}
We note that Theorems \ref{q_bounded_time_thm}, \ref{q_L1_time_thm}, and \ref{q_delta_time_thm} are generalisations of Theorems \ref{q_bounded_state_convergence_thm}, \ref{q_L1_state_convergence_thm}, and  \ref{delta_function_state_convergence_thm} respectively. This can be seen by taking $m=1$ and $t_1 = 0$, and arguing by duality as in \cite[Remark 1.7 (5)]{RS22}. Alternatively, see \eqref{duality_1}--\eqref{duality_3} below.
\end{remark}

\subsection{Previously known results}
\label{Previously known results}

One refers to the results in Section \ref{Time-independent results} as a microscopic derivation of a Gibbs measure for the nonlinear Schr\"{o}dinger equation from quantum many-body Gibbs states. Physically, one can interpret this as a high-density limit where the mass (or temperature) of the system tends to infinity. One can also interpret this as a semiclassical (mean-field) limit with semiclassical parameter $1/\tau$; see \cite[Section 1.1]{FKSS20} for a detailed explanation.
%The first microscopic derivation of a Gibbs measure (in the sense of the results of Section \ref{Time-independent results} above) for a nonlinear Schr\"{o}dinger equation 
The first such result was obtained by Lewin, Nam, and Rougerie \cite{LNR15}. In this work, the authors study the one-dimensional problem with positive translation-invariant interaction, and variants in higher dimensions, which do not require renormalisation. The method in \cite{LNR15} is based on a variational approach and the quantum de Finetti theorem. Further extensions in one dimension were given in \cite{LNR18}. The problem for translation-invariant (Wick-ordered) interactions when the dimension $d=2,3$ was subsequently studied by Fr\"{o}hlich, Knowles, Schlein, and the second author using a perturbative expansion in the interaction and Borel resummation techniques \cite{FKSS17}. Here, the authors could prove the result under a suitable modification of the grand canonical ensemble. The methods of \cite{FKSS17} were extended to positive interaction potentials $w \in L^p$ for optimal $p$ (in the sense of the work of Bourgain \cite{Bou97}) by the second author \cite{Soh19}. 

The full result when $d=2,3$ without the modification of the grand canonical ensemble introduced in \cite{FKSS17} was later shown simultaneously and using different methods by Lewin, Nam, and Rougerie \cite{LNR21} and Fr\"{o}hlich, Knowles, Schlein, and the second author \cite{FKSS20}. The $d=2$ result in \cite{LNR21} had previously been announced in \cite{LNR18b}. The method in \cite{LNR18b,LNR21} is a highly non-trivial extension of that used in \cite{LNR15}. The method in \cite{FKSS20} is based on a functional integral formulation. In \cite{FKSS22}, the $d=2$ result of \cite{FKSS20} was extended to the complex Euclidean $\Phi^4_2$ theory, which corresponds to taking $w=\delta$.

Results on time-dependent correlations were first considered when $d=1$ in \cite{FKSS18}. Related problems on the lattice were considered in \cite{FKSS20_3,Kno09,Salmhofer}. For more details on the previously known literature, we refer the reader to the expository works \cite{LNR19} and \cite{FKSS20_2}, as well as to the introduction of \cite{RS22}.

All of the aforementioned results deal with positive (defocusing) interactions. The first known result for focusing interactions (or more generally for interactions without any assumption on their sign) was recently obtained by the authors when $d=1$ in \cite{RS22}. In the quantum setting, we considered a model involving two-body interactions which led to a quartic nonlinearity in the classical setting. In this framework, we had to truncate the mass of the classical free field as in \eqref{q_gibbs_cutoff_defn_rigorous}, \eqref{q_classical_state_defn}, and \eqref{q_quantum_state_defn} above.
In \cite{RS22}, we could also study time-dependent correlations.

The present paper gives the first result for three-body interactions, which in turn yield a quintic nonlinearity in the limit $\tau \rightarrow \infty$. Our methods rely crucially on the presence of the cut-off $f$. This is a natural assumption in the focusing regime; see \cite{Bou94,CFL16,LRS88,Li_Liang_Wang,Liang_Wang,OST22,RSTW,Xian}. In the defocusing regime, one would expect the result to hold without a cut-off (see Remark \ref{Setup_1} above).

It is crucial for our analysis that the limiting problem \eqref{q_Hartree_defn} is Hamiltonian; see Lemma \ref{Hamiltonian_Lemma} above. Our goal is to obtain the microscopic derivation of the Gibbs measure for a quintic classical effective evolution, which we show is possible in our setting.  For related works on three-body interactions in the quantum many-body problem and their classical limits, we refer the reader to \cite{CP10,Che12,Lee,Nam_Ricaud_Triay_1,Nam_Ricaud_Triay_2,Nam_Ricaud_Triay_3,Xie15}. %Our convention for three-body interactions in the quantum setting \eqref{Quantum_interaction} differs from that of the aforementioned works. It would be interesting to see if our analysis applies in this framework. We do not consider this in the current paper.

%The classical effective evolution equation corresponding to the problem studied in \cite{CP10} and \cite{Xie15} is the local NLS, whereas the one is in \cite{Che12} is a nonlocal quintic NLS of a different form from that given in \eqref{q_Hartree_defn} above. We do not study the former nonlocal model in the current work.

%{\color{blue} These kind of three-body potentials were studied in \cite{CP10,Che12,Xie15}.}

 %{\color{red} AR: It may be worth including the following reference for three body interactions: \url{https://arxiv.org/abs/2201.13440}.}
%{\color{magenta} 
%In the case of an interaction potential with no positivity assumption, states of the form of $P^{(n)}_\tau$ without cut-off were studied in \cite{FKSS17,FKSS18,FKSS20,FKSS20_3,FKSS22,LNR15,LNR18,LNR18b,LNR19,LNR21,Soh19}. An overview of these results is given in \cite{FKSS20_2}. Cut-off grand canonical states were studied in \cite{RS22}.}

\subsection{Outline of proof}
We begin by proving local well-posedness, existence of the Gibbs measure, and almost sure global well-posedness for \eqref{q_Hartree_defn}; see Propositions \ref{Cauchy_problem_1} and \ref{Cauchy_problem_2} for precise statements. This is done by proving suitable multilinear estimates in Lemmas \ref{Multlinear_estimates_1} and \ref{Multlinear_estimates_2} which allow us to apply the argument from \cite{Bou94}. Here, we need to exploit the convolution structure of the nonlinearity, which is more complicated than that for the cubic problem studied previously in \cite{RS22}.

For the quantum many-body problem, we use a perturbative approach, similarly as in \cite{FKSS17,FKSS18,RS22}. Namely, we use a Duhamel expansion to perform a perturbative expansion of $\e^{-H_\#}$. As in \cite{RS22}, the  cut-off $f(\mc{N}_\#)$ in the definitions \eqref{q_classical_state_defn} and \eqref{q_quantum_state_defn} means the resulting explicit terms of the series both converge absolutely, meaning we avoid needing to use the Borel resummation techniques used in \cite{FKSS17,FKSS18}. 

The remainder terms of the series expansion in the quantum setting are analysed using the Feynman-Kac formula, which allows us to rewrite the remainder term using the explicit terms. We thus obtain a series with infinite radius of convergence of the resulting series. %This technique is only possible in one dimension, since in higher dimensions we need to Wick order our interaction, see \cite[Section 1.5]{RS22} for a more in depth discussion. 
In the classical setting, our remainder term is also analysed using the support properties of the cut-off function. 

To complete the analysis in the bounded case, we prove convergence of the explicit terms as $\tau \to \infty$ by using the complex analytic techniques in \cite[Section 3.6]{RS22}, which is the only place we use the technical assumption that $f$ is smooth. We also need sufficient bounds on the corresponding untruncated terms, see Proposition \ref{q_expansion_nugeq0} below. To prove the required bounds, we use a diagrammatic representation of the terms, similar to \cite[Sections 2 and 4]{FKSS17}, but now adapted to the quintic and normal ordered case. To complete the time-independent analysis for unbounded potentials, we use a diagonalisation argument similarly as in \cite[Section 4]{RS22}; see also \cite[Section 5]{FKSS18}.

To treat the time-dependent case for bounded interaction potentials, we use a Schwinger-Dyson expansion, first used in \cite{FKSS18} and also used in \cite{RS22}. This allows us to deduce Theorem \ref{q_bounded_time_thm} from Theorem \ref{q_bounded_state_convergence_thm}. Finally, to prove Theorems \ref{q_L1_time_thm} and \ref{q_delta_time_thm}, we need to prove an approximation result analogous to \cite[Lemma 5.4]{RS22}.  A challenge that presents itself is that is that there is no local well-posedness result for \eqref{q_Hartree_defn} in $\mfh$.
To overcome this, we need to explicitly use the almost sure global well-posedness of \eqref{q_Hartree_defn} and a suitable approximation result; see Lemma \ref{q_Hartree_approximation_lemma} for precise details. The corresponding approximation result for \eqref{focusing_quintic_NLS} requires slightly different arguments, and is shown in Lemma \ref{q_Hartree_approximation_lemma}.

\subsection{Structure of the paper}
In Section \ref{Cauchy problem}, we prove local well-posedness and the existence of the Gibbs measure for the quintic Hartree equation. Section \ref{q_Power_series_expansions_of_the_quantum_and_classical_states} contains the analysis of the Duhamel expansions of the classical and quantum states. This includes proving bounds for both the explicit and remainder terms in both the classical and quantum cases. In Section \ref{q_untruncated}, we prove convergence for the untruncated explicit terms, which completes the proof of Theorem \ref{q_bounded_state_convergence_thm}. In Section \ref{q_Unbounded_cases_section}, we study the time-independent problem for interaction potentials as in Assumption \ref{w_assumption}. Section \ref{q_Time_dependent_case} is devoted to the analysis of the time-dependent problem, both for bounded and unbounded interaction potentials. 
In Appendix \ref{Appendix A}, we summarise our results on a slightly different model for the quintic Hartree equation that is studied in detail in \cite[Chapter 4]{Rout_Thesis_2023} and \cite{RS23}.

%Finally, Appendix \ref{Proof of the approximation lemma} is dedicated to the proof of an approximation result needed in Section \ref{Cauchy problem}.

%{\color{blue} (Vedran) Update this later.}

\section{Analysis of the quintic Hartree equation \eqref{q_Hartree_defn}}
\label{Cauchy problem}

%\textbf{COMPLETE.}

In this section, we consider $w$ as in Assumption \ref{w_assumption} and study the Cauchy problem for the quintic Hartree equation \eqref{q_Hartree_defn}, written as
\begin{align}
\label{q_Hartree_Cauchy_problem}
\begin{cases}
\mathrm{i}\partial_t u + (\Delta - \kappa)u =-\frac{1}{3}\,\int \dd y \, \dd z \, w(x-y)\,w(x-z)\,|u(y)|^2\,|u(z)|^2\,u(x)&
\\
-\frac{2}{3}\,\int \dd y \, \dd z \, w(x-y)\,w(y-z)\,|u(y)|^2\,|u(z)|^2\,u(x)&
\\
u|_{t=0}=u_0 \in H^s(\La)\,.
\end{cases}
\end{align}

\subsection{Deterministic local well-posedness and invariance of the Gibbs measure}

We first prove the following deterministic local well-posedness result.
\begin{proposition}[Deterministic local existence for \eqref{q_Hartree_Cauchy_problem}]
\label{Cauchy_problem_1}
The Cauchy problem \eqref{q_Hartree_Cauchy_problem} is locally well-posed in $H^s(\La)$ for $s>0$.
\end{proposition}
The above result was shown for local interactions $w$ (as in Assumption \ref{w_assumption} (ii)) in \cite[Theorem 1]{Bou93}. Our aim is to show it for nonlocal interactions $w$ as in Assumption \ref{w_assumption} (i).

We then prove the following probabilistic result, which was shown for local interactions in \cite[Section 4]{Bou94}.

\begin{proposition}[Invariance of truncated Gibbs measure and almost sure global existence for \eqref{q_Hartree_Cauchy_problem}]
\label{Cauchy_problem_2}
The following claims hold.
\begin{itemize}
\item[(i)]  Recall the probability space $(\mathbb{C}, \mc{G}, \mu)$ defined in \eqref{q_rigorous_gibbs_defn}, $\vph \equiv \vph^\om$ defined in \eqref{q_class_free_field}, and $\mathcal{N}$ defined in \eqref{q_mass}. For $B>0$ sufficiently small, we have 
\begin{equation*}
\e^{\frac{1}{3} \int \dd x \, \dd y \, \dd z \, w(x-y)\,w(x-z)\,|\vph(x)|^2|\vph(y)|^2|\vph(z)|^2}  \chi_{(\mcN \leq B)} \in L^1(\dd\mu)\,.
\end{equation*}
In particular, taking $K=B$ in \eqref{q_cut-off_support}, we get that the probability measure $\mbP^f_{\mathrm{Gibbs}}$ in \eqref{q_gibbs_cutoff_defn_rigorous} is well-defined.
\item[(ii)] Consider $s \in (0,\frac{1}{2})$. The measure $\mbP^f_{\mathrm{Gibbs}}$ is invariant under the flow of \eqref{q_Hartree_Cauchy_problem}. Furthermore, \eqref{q_Hartree_Cauchy_problem} admits global solutions for $\mbP^f_{\mathrm{Gibbs}}$-almost every $u_0 \in H^s(\Lambda)$.
\end{itemize}
\end{proposition}

Before proving Propositions \ref{Cauchy_problem_1} and \ref{Cauchy_problem_2}, we note several multilinear estimates. %in $X^{s,b}$ spaces (\textbf{DEFINE}).

\begin{lemma}
\label{Multlinear_estimates_1}
Let $w_1,w_2 \in L^1(\La)$ be given. Then, for $q \in H^{\frac{1}{3}}(\La)$, we have
\begin{equation*}
\biggl|\int \dd x\,\dd y\,\dd z\, w_1(x-y)\,w_2(y-z)\,|q(x)|^2\,|q(y)|^2\,|q(z)|^2\biggr| 
\leq \|w_1\|_{L^1}\,\|w_2\|_{L^1}\,\|q\|_{L^6}^6\,.
\end{equation*}
\end{lemma}

\begin{proof}
We write 
\begin{multline*}
\int \dd x\,\dd y\,\dd z\, w_1(x-y)\,w_2(y-z)\,|q(x)|^2\,|q(y)|^2\,|q(z)|^2
\\
=\int \dd x\, \bigl(w_1*|q|^2\bigr)(x)\,\bigl(w_2*|q|^2\bigr)(x)\,|q(x)|^2\,.
\end{multline*}
By H\"{o}lder's and Young's inequality, the above expression is in absolute value
\begin{equation*}
\leq \|w_1*|q|^2\|_{L^3}\, \|w_2*|q|^2\|_{L^3}\,\|q\|_{L^6}^2 \leq \|w_1\|_{L^1}\,\|w_2\|_{L^1}\,\|q\|_{L^6}^6\,.
\end{equation*}
\end{proof}

We recall the definition of the $X^{s,b}$ spaces (or the dispersive Sobolev spaces). For a function $v: \T \times \R \to \C$, we denote its spacetime Fourier transform by
\begin{equation}
\label{spacetime_Fourier_transform}
\tilde{v}(k,\eta) := \int_\R \dd t\, \int_\T\dd x\,  v(x,t) \, \e^{-2\pi \mathrm{i} k x - 2 \pi \mathrm{i} \eta t}.
\end{equation}
\begin{definition}
\label{q_Xsb_defn}
Given $v:\T \times \R \to \C$ and $s,b \in \R$, we define
\begin{equation}
\label{X^{s,b}_norm}
\|v\|_{X^{s,b}} := \left\| \left(1+|2\pi k|\right)^s\left(1 + |\eta + 2\pi k^2|\right)^b\tilde{v}(k,\eta)\right\|_{\ell^2_k L^2_\eta}=\|\e^{-\mathrm{i} t \Delta} \,v \|_{H^s_x H^b_t}\,.
\end{equation}
For an interval $I \subset \R$, we define the local $X^{s,b}$ space norm as
\begin{equation}
\label{local_X^{s,b}_norm}
\|v\|_{X^{s,b}_I} := \inf_V \|V\|_{X^{s,b}}\,,
\end{equation}
where the infimum is taken over all $V:\T \times \R \to \C$ satisfying $V|_{\T \times I} = v|_{\T \times I}$.
\end{definition}

\begin{lemma}
\label{Multlinear_estimates_2}
Consider $w_1,w_2 \in L^1(\La)$.
Given $s,\varepsilon>0$ and $v_j \in X^{s,1/2+\varepsilon}$, for $j=1,\ldots,5$, we let
\begin{multline}
\label{N_1_definition}
\mathcal{N}_1(v_1,v_2,v_3,v_4,v_5)(x,t):=
\\
\int \dd y\,\dd z\,w_1(x-y)\,w_2(x-z)\,v_1(y,t)\,\overline{v_2(y,t)}\,v_3(z,t)\,\overline{v_4(z,t)}\,v_5(x,t)
\end{multline}
and 
\begin{multline}
\label{N_2_definition}
\mathcal{N}_2(v_1,v_2,v_3,v_4,v_5)(x,t):=
\\
\int \dd y\,\dd z\,w_1(x-y)\,w_2(y-z)\,v_1(y,t)\,\overline{v_2(y,t)}\,v_3(z,t)\,\overline{v_4(z,t)}\,v_5(x,t)\,.
\end{multline}
For $\varepsilon>0$ sufficiently small, we have that for all $t_0 \in \R$ and $\delta>0$ small
\begin{multline}
\label{Multlinear_estimates_2_claim}
\|\mathcal{N}_1(v_1,v_2,v_3,v_4,v_5)\|_{X^{s,-\frac{1}{2}+\varepsilon}_{[t_0,t_0+\delta]}} +
\|\mathcal{N}_2(v_1,v_2,v_3,v_4,v_5)\|_{X^{s,-\frac{1}{2}+\varepsilon}_{[t_0,t_0+\delta]}} 
\\
\lesssim_{s,\varepsilon} \delta^{\varepsilon}\, \|w_1\|_{L^1}\,\|w_2\|_{L^1}\,\prod_{j=1}^{5}\|v_j\|_{X^{s,\frac{1}{2}+\varepsilon}_{[t_0,t_0+\delta]}}\,.
\end{multline}

\end{lemma}

\begin{proof}
It suffices to prove the result for smooth $v_j$. The general claim follows by density.
We first prove
\begin{equation}
\label{Multlinear_estimates_2_claim_2_1}
\|\mathcal{N}_1(v_1,v_2,v_3,v_4,v_5)\|_{X^{s,-\frac{1}{2}+\varepsilon}_{[t_0,t_0+\delta]}}
\lesssim_{s,\varepsilon} \delta^{\varepsilon}\, \|w_1\|_{L^1}\,\|w_2\|_{L^1}\,\prod_{j=1}^{5}\|v_j\|_{X^{s,\frac{1}{2}+\varepsilon}_{[t_0,t_0+\delta]}}\,.
\end{equation}
Before proving \eqref{Multlinear_estimates_2_claim_2_1}, we show the following global in time estimate.
\begin{equation}
\label{Multlinear_estimates_2_claim_1}
\|\mathcal{N}_1(v_1,v_2,v_3,v_4,v_5)\|_{X^{s,-\frac{1}{2}+\varepsilon}} \lesssim_{s,\varepsilon} \|w_1\|_{L^1}\,\|w_2\|_{L^1}\,\prod_{j=1}^{5}\|v_j\|_{X^{s,\frac{1}{2}-\varepsilon}}\,.
\end{equation}
%By a density argument, it suffices to show \eqref{Multlinear_estimates_2_claim_1} when the $w_i,\, i=1,2$ and $v_j,\, j=1,\ldots,5$ are smooth. Thus, all of the calculations that follow are rigorous.
Note that we can rewrite \eqref{N_1_definition} as
\begin{equation}
\label{N_1_convolution}
\mathcal{N}_1(v_1,v_2,v_3,v_4,v_5)(x,t)=\bigl[ w_1*(v_1 \overline{v}_2) \bigr]\,\bigl[ w_2*(v_3 \overline{v}_4) \bigr]\,v_5(x,t)\,,
\end{equation}
where $*$ denotes convolution in $x$. Using \eqref{N_1_convolution}, we obtain that for all $k \in \Z$ and $\eta \in \R$
\begin{multline}
\label{N_1_claim}
\bigl(\mathcal{N}_1(v_1,v_2,v_3,v_4,v_5)\bigr)\,\widetilde{\,}\,(k,\eta)
\\
= \sum_{k_1,\ldots,k_5} \int \dd \eta_1 \cdots \dd \eta_5\, \delta(k_1-k_2+k_3-k_4+k_5-k)\,\delta(\eta_1-\eta_2+\eta_3-\eta_4+\eta_5-\eta)\,
\\
\times \widehat{w}_1(k_1-k_2)\,\widehat{w}_2(k_3-k_4)\,\widetilde{v}_1(k_1,\eta_1)\,\overline{\widetilde{v}_2(k_2,\eta_2)}\,\widetilde{v}_3(k_3,\eta_3)\,\overline{\widetilde{v}_4(k_4,\eta_4)}\,\widetilde{v}_5(k_5,\eta_5)\,.
\end{multline}
In particular, from \eqref{N_1_claim} and the Hausdorff-Young inequality, we deduce that
\begin{multline}
\label{N_1_claim_A}
\big|\bigl(\mathcal{N}_1(v_1,v_2,v_3,v_4,v_5)\bigr)\,\widetilde{\,}\,(k,\eta)\big|
\leq \|w_1\|_{L^1}\,\|w_2\|_{L^1}\,
\\
\times \sum_{k_1,\ldots,k_5} \int \dd \eta_1 \cdots \dd \eta_5\, \delta(k_1-k_2+k_3-k_4+k_5-k)\,\delta(\eta_1-\eta_2+\eta_3-\eta_4+\eta_5-\eta)\,
\\
\times |\widetilde{v}_1(k_1,\eta_1)|\,|\widetilde{v}_2(k_2,\eta_2)|\,|\widetilde{v}_3(k_3,\eta_3)|\,|\widetilde{v}_4(k_4,\eta_4)|\,|\widetilde{v}_5(k_5,\eta_5)|\,.
\end{multline}
By Parseval's theorem for the spacetime Fourier transform, we note that the right-hand side of \eqref{N_1_claim_A} can be written as 
\begin{equation}
\label{N_1_claim_B}
\|w_1\|_{L^1}\,\|w_2\|_{L^1}\, (F_1F_2F_3F_4F_5)\,\widetilde{\,}\,(k,\eta)\,,
\end{equation}
where for $j=1,\ldots,5$, the function $F_j$ is chosen such that 
\begin{equation}
\label{F_j_choice}
\widetilde{F}_j=|\widetilde{v}_j|\,.
\end{equation}
Recalling Definition \ref{q_Xsb_defn} and using \eqref{N_1_claim_A}--\eqref{N_1_claim_B}, we obtain that
\begin{equation}
\label{N_2_claim_corollary}
\|\mathcal{N}_1(v_1,v_2,v_3,v_4,v_5)\|_{X^{s,-\frac{1}{2}+\varepsilon}} \leq \|w_1\|_{L^1}\,\|w_2\|_{L^1}\,\,\|F_1F_2 F_3 F_4 F_5\|_{X^{s,-\frac{1}{2}+\varepsilon}}\,.
\end{equation}
We now use the known quintic estimate 
\begin{equation}
\label{quintic_estimate_known}
\|F_1F_2F_3F_4F_5\|_{X^{s,-\frac{1}{2}+\varepsilon}} \lesssim_{s,\varepsilon} \prod_{j=1}^{5}\|F_j\|_{X^{s,\frac{1}{2}-\varepsilon}}\,,
\end{equation}
for $\varepsilon>0$ sufficiently small.
For a proof, see \cite[Proof of (3.56)]{Erdogan_Tzirakis}. Strictly speaking, the claim \cite[(3.56)]{Erdogan_Tzirakis} is stated for all $F_j$ being equal, but the proof based on the trilinear estimates \cite[Lemma 3.29, Corollary 3.30]{Erdogan_Tzirakis} extends to the general case verbatim. We omit the details. By \eqref{F_j_choice} and Definition \ref{q_Xsb_defn}, we have that for $j=1,\ldots,5$
\begin{equation}
\label{F_j_choice_norm}
\|F_j\|_{X^{s,\frac{1}{2}-\varepsilon}}=\|v_j\|_{X^{s,\frac{1}{2}-\varepsilon}}\,.
\end{equation}
The estimate \eqref{Multlinear_estimates_2_claim_1} then follows from \eqref{N_2_claim_corollary}--\eqref{F_j_choice_norm}.

We now use \eqref{Multlinear_estimates_2_claim_1} to show \eqref{Multlinear_estimates_2_claim_2_1} when $t_0=0$. The claim \eqref{Multlinear_estimates_2_claim_2_1} for general $t_0$ follows by a suitable translation in time. We recall \eqref{local_X^{s,b}_norm}, use \eqref{Multlinear_estimates_2_claim_1} for $V_j \in  X^{s,b}$ such that $V_j|_{\T \times [0,\delta]}=v_j|_{\T \times [0,\delta]}$ for $j=1,\ldots,5$, take infima over $V_j$, and deduce the claim \eqref{Multlinear_estimates_2_claim_2_1} for $t_0=0$ by using the localisation property\footnote{In fact, this argument shows \eqref{Multlinear_estimates_2_claim_2} with $\delta^{\varepsilon}$ replaced by $\delta^{5\varepsilon-}$, but we will not need this estimate in the sequel.}
\begin{equation}
\label{X^{s,b}_localisation}
\|v\|_{X^{s,b_1}_{[0,\delta]}} \lesssim_{s} \delta^{b_2-b_1}\,\|v\|_{X^{s,b_2}_{[0,\delta]}}\,,
\end{equation}
for $-\frac{1}{2}<b_1 \leq b_2<\frac{1}{2}$. For a self-contained proof of \eqref{X^{s,b}_localisation}, we refer the reader to 
\cite[Lemma 3.11]{Erdogan_Tzirakis}; see also \cite[Lemma 2.11]{Tao}. The estimate \eqref{Multlinear_estimates_2_claim_2_1} now follows.

We now show 
\begin{equation}
\label{Multlinear_estimates_2_claim_2_2}
\|\mathcal{N}_2(v_1,v_2,v_3,v_4,v_5)\|_{X^{s,-\frac{1}{2}+\varepsilon}_{[t_0,t_0+\delta]}}
\lesssim_{s,\varepsilon} \delta^{\varepsilon}\, \|w_1\|_{L^1}\,\|w_2\|_{L^1}\,\prod_{j=1}^{5}\|v_j\|_{X^{s,\frac{1}{2}+\varepsilon}_{[t_0,t_0+\delta]}}\,.
\end{equation}
By arguing as for \eqref{Multlinear_estimates_2_claim_2_1}, \eqref{Multlinear_estimates_2_claim_2_2} follows if we show
\begin{equation}
\label{Multlinear_estimates_2_claim_2}
\|\mathcal{N}_2(v_1,v_2,v_3,v_4,v_5)\|_{X^{s,-\frac{1}{2}+\varepsilon}} \lesssim_{s,\varepsilon} \|w_1\|_{L^1}\,\|w_2\|_{L^1}\,\prod_{j=1}^{5}\|v_j\|_{X^{s,\frac{1}{2}-\varepsilon}}\,.
\end{equation}
Note that we can rewrite \eqref{N_2_definition} as
\begin{equation}
\label{N_2_convolution}
\mathcal{N}_2(v_1,v_2,v_3,v_4,v_5)(x,t)=
\biggl(w_1* \Bigl\{v_1 \overline{v}_2 \bigl[w_2*\bigl(v_3 \overline{v}_4\bigr)\bigr] \Bigr\}\biggr)v_5
(x,t)\,.
\end{equation}
Using \eqref{N_2_convolution}, we obtain that for all $k \in \Z$ and $\eta \in \R$
\begin{multline}
\label{N_2_claim}
\bigl(\mathcal{N}_2(v_1,v_2,v_3,v_4,v_5)\bigr)\,\widetilde{\,}\,(k,\eta)
\\
= \sum_{k_1,\ldots,k_5} \int \dd \eta_1 \cdots \dd \eta_5\, \delta(k_1-k_2+k_3-k_4+k_5-k)\,\delta(\eta_1-\eta_2+\eta_3-\eta_4+\eta_5-\eta)\,
\\
\times \widehat{w}_1(k_1-k_2+k_3-k_4)\,\widehat{w}_2(k_3-k_4)
\\
\times
\widetilde{v}_1(k_1,\eta_1)\,\overline{\widetilde{v}_2(k_2,\eta_2)}\,\widetilde{v}_3(k_3,\eta_3)\,\overline{\widetilde{v}_4(k_4,\eta_4)}\,\widetilde{v}_5(k_5,\eta_5)\,.
\end{multline}
Using \eqref{N_2_claim}, and arguing as for \eqref{N_1_claim_A}, we obtain 
\begin{multline}
\label{N_2_claim_A}
\big|\bigl(\mathcal{N}_2(v_1,v_2,v_3,v_4,v_5)\bigr)\,\widetilde{\,}\,(k,\eta)\big|
\leq \|w_1\|_{L^1}\,\|w_2\|_{L^1}\,
\\
\times \sum_{k_1,\ldots,k_5} \int \dd \eta_1 \cdots \dd \eta_5\, \delta(k_1-k_2+k_3-k_4+k_5-k)\,\delta(\eta_1-\eta_2+\eta_3-\eta_4+\eta_5-\eta)\,
\\
\times |\widetilde{v}_1(k_1,\eta_1)|\,|\widetilde{v}_2(k_2,\eta_2)|\,|\widetilde{v}_3(k_3,\eta_3)|\,|\widetilde{v}_4(k_4,\eta_4)|\,|\widetilde{v}_5(k_5,\eta_5)|\,.
\end{multline}
We deduce \eqref{Multlinear_estimates_2_claim_2} by using \eqref{N_2_claim_A} and arguing analogously as in the proof of \eqref{Multlinear_estimates_2_claim_1}.
The claim \eqref{Multlinear_estimates_2_claim} follows from  \eqref{Multlinear_estimates_2_claim_2_1} and \eqref{Multlinear_estimates_2_claim_2_2}.
\end{proof}

Let us recall several general properties of $X^{s,b}$ spaces.  

\begin{lemma}[Properties of $X^{s,b}$ spaces]
\label{X^{s,b}_space_properties}
Let us fix $b=\frac{1}{2}+\varepsilon$ for $\varepsilon>0$ small and $s \in \R$. The following estimates hold for all $t_0 \in \R$ and $\delta>0$.
\begin{itemize}
\item[(i)] For all $u \in X^{s,b}_{[t_0,t_0+\delta]}$, we have $\|u\|_{L^{\infty}_{t \in [t_0,t_0+\delta]} H^s_x} \lesssim_b \|u\|_{X^{s,b}_{[t_0,t_0+\delta]}}$.
\item[(ii)] For all $\Phi \in H^s$, we have 
\begin{equation}
\label{(ii)_claim}
\|\mathrm{e}^{\mathrm{i}(t-t_0)\Delta} \Phi\|_{X^{s,b}_{[t_0,t_0+\delta]}} \lesssim_{b} \|\Phi\|_{H^s}\,.
\end{equation}
\item[(iii)] For all $F \in X^{s,b-1}_{[t_0,t_0+\delta]}$, we have
\begin{equation}
\label{(iii)_claim}
\biggl\|\int_{t_0}^t \dd t'\, \mathrm{e}^{\mathrm{i}(t-t')\Delta}\,F(\cdot,t')\biggr\|_{X^{s,b}_{[t_0,t_0+\delta]}} \lesssim_{b} \|F\|_{X^{s,b-1}_{[t_0,t_0+\delta]}}\,.
\end{equation}
\end{itemize}

\end{lemma}

\begin{proof}[Proof of Lemma \ref{X^{s,b}_space_properties}]
The estimates above are standard. Claim (i) follows by Sobolev embedding in the time variable, using the assumption that $b>\frac{1}{2}$. The particular form of the estimates given by (ii) and (iii) is proved in a self-contained way in \cite[Lemma 3.10]{Erdogan_Tzirakis} and \cite[Lemma 3.12]{Erdogan_Tzirakis} respectively\footnote{The results in \cite{Erdogan_Tzirakis} are proved for the Airy semigroup $W_t=\e^{-t \delta_x^3}$ on the real line. The arguments for the Schr\"{o}dinger semigroup $\e^{\mathrm{i} t \Delta}$ on the torus follow analogously.}.
In particular, the arguments from the proof of \cite[Lemma 3.10]{Erdogan_Tzirakis} imply that for all $\Phi \in H^s$ and $\psi \in C_c^{\infty}(\R)$, we have 
\begin{equation*}
\|\psi(t)\, \mathrm{e}^{\mathrm{i}t\Delta} \Phi\|_{X^{s,b}_{[t_0,t_0+\delta]}} \lesssim_{b,\psi} \|\Phi\|_{H^s}\,.
\end{equation*}
We now consider $\psi$ which is equal to $1$ on $[t_0,t_0+\delta]$ to deduce that
\begin{equation}
\label{(ii)_claim_auxiliary}
\|\mathrm{e}^{\mathrm{i}t\Delta} \Phi\|_{X^{s,b}_{[t_0,t_0+\delta]}} \lesssim_{b} \|\Phi\|_{H^s}\,.
\end{equation}
The estimate \eqref{(ii)_claim_auxiliary} follows from \eqref{(ii)_claim} and the unitarity of $\e^{\mathrm{i}t\Delta}$.

The argument for (iii) is similar. When $t_0=0$, it follows from the proof of \cite[Lemma 3.12]{Erdogan_Tzirakis}. More precisely, it suffices to show that for all nonnegative $\psi \in C_c^{\infty}(\R)$, we have
\begin{equation}
\label{X^{s,b}_(iii)}
\biggl\|\psi(t)\,\int_0^t \dd t'\, \mathrm{e}^{\mathrm{i}(t-t')\Delta}\,F(\cdot,t')\biggr\|_{X^{s,b}} \lesssim_{b,\psi} \|F\|_{X^{s,b-1}}\,.
\end{equation}
The claim of (iii) then follows from \eqref{X^{s,b}_(iii)} by taking infima over $F$ as in \eqref{local_X^{s,b}_norm}.
Using the fact that $\psi(t)$ and $\e^{\mathrm{i} t \Delta}$ commute, as well as \eqref{X^{s,b}_norm}, we have that the expression on the left-hand side of \eqref{X^{s,b}_(iii)} equals
\begin{equation}
\label{X^{s,b}_(iii)_2}
\biggl\|\psi(t)\,\int_0^t \dd t'\, \mathrm{e}^{-\mathrm{i} t' \Delta}\,F(\cdot,t')\biggr\|_{H^s_x H^b_t}
\leq \biggl\|\psi(t)\,\int_0^t \dd t'\, \biggl\|\mathrm{e}^{-\mathrm{i} t' \Delta}\,F(\cdot,t')\biggr\|_{H^s_x}\biggr\|_{H^b_t}
\,.
\end{equation}
Using \eqref{X^{s,b}_norm} and \eqref{X^{s,b}_(iii)_2}, we note that \eqref{X^{s,b}_(iii)} follows from the following inequality.
\begin{equation}
\label{X^{s,b}_(iii)_3}
\biggl\|\psi(t)\, \int_{0}^{t} \dd t'\, g(t')\, \biggr\|_{H^b_t} \lesssim_{b,\psi} \|g\|_{H^{b'}_t}\,.
\end{equation}
The estimate \eqref{X^{s,b}_(iii)_3} corresponds to \cite[(3.18)]{Erdogan_Tzirakis}, which is shown in the proof of \cite[Lemma 3.12]{Erdogan_Tzirakis}. This shows \eqref{(iii)_claim} when $t_0=0$.

The proof of \eqref{(iii)_claim} for general $t_0$ follows by translation.
In order to explain the last step in more detail, let us consider $F \in X^{s,b-1}_{[t_0,t_0+\delta]}$. We then consider $G(x,t):=F(x,t+t_0)$ and observe that 
\begin{equation}
\label{Duhamel_X^{s,b}_1}
\|G\|_{X^{s,b-1}_{[0,\delta]}}=\|F\|_{X^{s,b-1}_{[t_0,t_0+\delta]}}\,.
\end{equation} 
A direct calculation shows that
\begin{equation}
\label{Duhamel_X^{s,b}_2}
\int_{t_0}^{t} \dd t'\, \e^{\mathrm{i} (t-t') \Delta}\, F(t')=\int_{0}^{t-t_0} \dd t'\, \e^{\mathrm{i} (t-t_0-t') \Delta}\, G(t')\,,
\end{equation}
from where we deduce
\begin{equation}
\label{Duhamel_X^{s,b}_3}
\biggl\|\int_{t_0}^t \dd t'\, \mathrm{e}^{\mathrm{i}(t-t')\Delta}\,F(\cdot,t')\biggr\|_{X^{s,b}_{[t_0,t_0+\delta]}}=\biggl\|\int_{0}^{t} \dd t'\, \e^{\mathrm{i} (t-t') \Delta}\, G(t')\biggr\|_{X^{s,b}_{[0,\delta]}}\,.
\end{equation}
The result of claim (iii) then follows from \eqref{Duhamel_X^{s,b}_1}--\eqref{Duhamel_X^{s,b}_3} and \eqref{(iii)_claim} when $t_0=0$.
\end{proof}

\begin{remark}
We refer the reader to \cite[Lemma 5.3]{FKSS18}, \cite[Appendix A]{FKSS18} for a self-contained summary of similar estimates in $X^{s,b}$ spaces, based on \cite{KPV1,KPV2}. We also refer the reader to %\cite[Section 3.3.1]{Erdogan_Tzirakis} and 
\cite[Section 2.6]{Tao}.
Note that, in contrast to \cite[Lemma 5.3]{FKSS18}, in (ii)--(iii), we are working in local $X^{s,b}$ spaces instead of with a cut-off function $\psi$ in the $t$ variable.  
\end{remark}

\begin{proof}[Proof of Proposition \ref{Cauchy_problem_1}]
By replacing $u$ with $\e^{\mathrm{i}t\kappa}u$, we can reduce to considering the case\footnote{Note that this transformation does not change the $H^s$ norm.} when $\kappa=0$.
We consider solutions for non-negative times. The argument for negative times is analogous.
For the remainder of this section, we consider $w_1=w_2=w$ in \eqref{N_1_definition}--\eqref{N_2_definition}, and so the nonlinearity in \eqref{q_Hartree_Cauchy_problem} is equal to
\begin{equation}
\label{nonlinearity_N}
\mathcal{N}_{1,2}(u,u,u,u,u):=-\frac{1}{3}\mathcal{N}_1(u,u,u,u,u)-\frac{2}{3}\mathcal{N}_2(u,u,u,u,u)\,.
\end{equation}

Let us fix $b=\frac{1}{2}+\varepsilon$ for $\varepsilon>0$ sufficiently small. 
%Let $\psi \in C_c^{\infty}(\R)$ be a function that equals to $1$ on $[-1,1]$. 
With $\mathcal{N}_{1,2}(u,u,u,u,u)$ as in \eqref{nonlinearity_N}, we consider the map
\begin{equation}
\label{map_L}
(Lv)(\cdot,t):=\e^{\mathrm{i} t \Delta} \,u_0 -\mathrm{i} \,\int_{0}^{t}\, \dd t'\, \e^{\mathrm{i}(t-t')\Delta}\,\mathcal{N}_{1,2}(v,v,v,v,v)(t')\,.
\end{equation}
In order to show the local existence of a solution, by the Banach fixed point theorem, it suffices to show that for suitable $\alpha>0$ and $\delta \sim \|u_0\|_{H^s}^{-\alpha}$ sufficiently small, the map $L$ is a contraction on the ball
\begin{equation}
\label{ball_B}
\mathcal{B}:=\Bigl\{v \in X^{s,b}_{[0,\delta]}\,,\quad \|v\|_{X^{s,b}_{[0,\delta]}} \leq \mathcal{M} \|u_0\|_{H^s}\Bigr\}
\end{equation}
in the Banach space $X^{s,b}_{[0,\delta]}$ for suitable $\mathcal{M}>0$. By construction, each such fixed point $u$ will be a mild solution of \eqref{q_Hartree_Cauchy_problem} on the time interval $[0,\delta]$, meaning that for all $t \in [0,\delta]$, we have
\begin{equation}
\label{mild_solution_NLS}
u(\cdot,t)=\e^{\mathrm{i} t \Delta} \,u_0 -\mathrm{i} \,\int_{0}^{t}\, \dd t'\, \e^{\mathrm{i}(t-t')\Delta}\,\mathcal{N}_{1,2}(u,u,u,u,u)(t')\,.
\end{equation}
Furthermore, by Lemma \ref{X^{s,b}_space_properties} (i), it follows that $\|u\|_{L^{\infty}_{t \in [0,\delta]}H^s_x} \lesssim_{\mathcal{M},b} \|u_0\|_{H^s}$.

We now show that \eqref{map_L} is a contraction on \eqref{ball_B}. By using Lemma \ref{X^{s,b}_space_properties} (ii), (iii) and Lemma \ref{Multlinear_estimates_2} with $v_1=v_2=v_3=v_4=v_5=v$, we deduce that
\begin{equation}
\label{Lv_estimate_1}
\|Lv\|_{X^{s,b}_{[0,\delta]}} \leq C_1 \|u_0\|_{H^s}+C_2 \,\delta^{\varepsilon} \|v\|_{X^{s,b}_{[0,\delta]}}^5\,,
\end{equation}
for suitable constants $C_1,C_2>0$. Note that $C_1$ is the implied constant in Lemma \ref{X^{s,b}_space_properties} (ii).
Similarly, we have
\begin{equation}
\label{Lv_estimate_2}
\|Lv^{(1)}-Lv^{(2)}\|_{X^{s,b}_{[0,\delta]}} \leq C_3 \,\delta^{\varepsilon}  (\|v^{(1)}\|_{X^{s,b}_{[0,\delta]}}^4+\|v^{(2)}\|_{X^{s,b}_{[0,\delta]}}^4)\|v^{(1)}-v^{(2)}\|_{X^{s,b}_{[0,\delta]}}\,,
\end{equation}
for a suitable constant $C_3>0$.
In order to deduce \eqref{Lv_estimate_2}, we used Lemma \ref{X^{s,b}_space_properties} (iii), as well as the precise multilinear form of $\mathcal{N}_{1,2}$ given by \eqref{nonlinearity_N}, and Lemma  \ref{Multlinear_estimates_2} with $v_j$ taking values $v^{(1)}, v^{(2)}$, or $v^{(1)}-v^{(2)}$. From \eqref{Lv_estimate_1}--\eqref{Lv_estimate_2}, it follows that \eqref{map_L} is a contraction on \eqref{ball_B} if we take $\mathcal{M}=2C_1$ and $\delta \sim \|u_0\|_{H^s}^{-\frac{4}{\varepsilon}}$ sufficiently small. 
More precisely, we choose $\delta>0$ such that
\begin{equation}
\label{delta_choice}
C_2 \,\delta^{\varepsilon} \mathcal{M}^5 \|u_0\|_{H^s}^4\leq C_1\,,\qquad
2 C_3 \,\delta^{\varepsilon} \mathcal{M}^4 \|u_0\|_{H^s}^4 \leq \frac{1}{2}\,.
\end{equation}

The above argument also shows the conditional uniqueness of mild solutions of \eqref{q_Hartree_Cauchy_problem} in \eqref{ball_B}. Namely, suppose that $u^{(1)}, u^{(2)} \in \mathcal{B}$ both satisfy \eqref{mild_solution_NLS} for $t \in [0,\delta]$. %By multiplying both sides of \eqref{mild_solution_NLS} by $\psi(t)$ and 
By repeating the earlier arguments, we deduce that 
\begin{multline}
\label{u_difference_estimate}
\|u^{(1)}-u^{(2)}\|_{X^{s,b}_{[0,\delta]}} \leq C_3 \,\delta^{\varepsilon}  (\|u^{(1)}\|_{X^{s,b}_{[0,\delta]}}^4+\|u^{(2)}\|_{X^{s,b}_{[0,\delta]}}^4)\|u^{(1)}-u^{(2)}\|_{X^{s,b}_{[0,\delta]}}
\\
\leq \frac{1}{2}\,\|u^{(1)}-u^{(2)}\|_{X^{s,b}_{[0,\delta]}}\,.
\end{multline}
For the second inequality in \eqref{u_difference_estimate}, we used the second condition in \eqref{delta_choice}. From \eqref{u_difference_estimate}, it indeed follows that $u^{(1)}=u^{(2)}$. This concludes the proof.
\end{proof}

\begin{proof}[Proof of Proposition \ref{Cauchy_problem_2}]
The result when $w$ is as in Assumption \ref{w_assumption} (ii) is shown in \cite{Bou94}. Therefore, we need to show that it holds for $w$ as in Assumption \ref{w_assumption} (i).
Claim (i) follows by combining Lemma \ref{Multlinear_estimates_1} with $w_1=w_2=w$ and \cite[Lemma 3.10]{Bou94}. (For a summary of the proof of the latter claim, we refer the reader to \cite[Appendix A]{RS22}).

We now prove claim (ii). Once we have the local well-posedness given by Proposition \ref{Cauchy_problem_1} above and the tools used in the proof (most notably the multilinear estimate given by Lemma \ref{Multlinear_estimates_2}), the argument follows that of \cite[Section 4]{Bou94}, which corresponds to formally taking $w=\delta$ in \eqref{q_Hartree_Cauchy_problem}. We outline only the main differences needed to consider the nonlocal problem \eqref{q_Hartree_Cauchy_problem}. The main idea is to approximate \eqref{q_Hartree_Cauchy_problem} by a finite-dimensional system and to prove a suitable approximation result.

\textbf{Step 1. Introducing the finite-dimensional system.} 

Given $N \in \N^*$, we denote by $P_N$ the operator 
\begin{equation*}
P_N g(x):=\sum_{|k| \leq N} \widehat{g}(k) \,\e^{2\pi \mathrm{i} k x}\,,
\end{equation*}
i.e.\ the projection onto frequencies $|k| \leq N$. We then compare \eqref{q_Hartree_Cauchy_problem_N}
with its finite-dimensional truncation given by the following.
\begin{equation}
\label{q_Hartree_Cauchy_problem_N}
\begin{cases}
\mathrm{i} \partial_t u^N + (\Delta - \kappa)u^N=
\\
-\frac{1}{3} P_N\bigl[\int \dd y \, \dd z \, w(x-y)\,w(x-z)\,|u^N(y)|^2\,|u^N(z)|^2\,u^N(x)\bigr]
\\ -\frac{2}{3} P_N\bigl[\int \dd y \, \dd z \, w(x-y)\,w(y-z)\,|u^N(y)|^2\,|u^N(z)|^2\,u^N(x)\bigr]
\\u^N|_{t=0}=P_N u_0\,.
\end{cases}
\end{equation}
Let us note that, in \eqref{q_Hartree_Cauchy_problem_N}, we are applying $P_N(\cdot)$ in the $x$ variable.
We write the solution $u^N$ of the finite-dimensional system \eqref{q_Hartree_Cauchy_problem_N} as 
\begin{equation}
\label{u^N_Fourier_series}
u^N(x,t)=\sum_{|k| \leq N} a_k(t)\,\e^{2\pi \mathrm{i} k x}\,.
\end{equation}
In light of \eqref{u^N_Fourier_series}, we identity $u^N$ with $a=(a_k)_{|k| \leq N}$. With this identification, we can write \eqref{q_Hartree_Cauchy_problem_N} as a Hamiltonian system 
\begin{equation}
\label{Hamiltonian_system_N}
\frac{\dd a_k}{\dd t}=-\mathrm{i} \,\frac{\partial H_N(a)}{\partial \bar{a}_k}\,, \quad |k| \leq N\,,
\end{equation}
where the Hamiltonian is given by 
\begin{multline}
\label{H^N}
H_N(a)=\sum_{|k| \leq N} (4\pi^2 |k|^2+\kappa) |a_k|^2
-\frac{1}{3} \int \dd x\,\dd y\,\dd z\, w(x-y)\,w(x-z)\, 
\\
\times
\Biggl|\sum_{|k| \leq N} a_k\,\e^{2\pi \mathrm{i}k x}\Biggr|^2
\, \Biggl|\sum_{|k| \leq N} a_k\,\e^{2\pi \mathrm{i}k y}\Biggr|^2\, \Biggl|\sum_{|k| \leq N} a_k\,\e^{2\pi \mathrm{i}k z}\Biggr|^2
\end{multline}
and 
\begin{equation}
\label{del_bar}
\frac{\partial}{\partial \bar{a}_k}=\frac{1}{2}\biggl(\frac{\partial}{\partial\mathrm{Re}\,a_k}+\mathrm{i} \,\frac{\partial}{\partial\mathrm{Im}\,a_k}\biggr)\,.
\end{equation}
Let us note that \eqref{H^N} and \eqref{del_bar} differ from the corresponding quantities used in \cite{Bou94}. This is because the convention of the Poisson structure in \cite{Bou94} amounts to taking $\{u(x),\bar{u}(y)\}=2\mathrm{i} \delta(x-y)$, which differs from \eqref{Poisson_structure} by a factor of $2$. See Remark \ref{Hamiltonian_structure_remark} below for more details.

Let us show \eqref{Hamiltonian_system_N} in detail. By \eqref{q_Hartree_Cauchy_problem_N} and \eqref{H^N}, it suffices to show that for $|k| \leq N$, the $k$-th Fourier coefficient of 
\begin{multline}
\label{Q_N}
Q_N(a):=-\frac{1}{3} P_N\Biggl[ \int \dd y \,\dd z\, w(x-y)\,w(x-z)\, 
\\
\times
\Biggl|\sum_{|k_2| \leq N} a_{k_2}\e^{2\pi \mathrm{i}k_2 y}\Biggr|^2\,\Biggl|\sum_{|k_3| \leq N} a_{k_3}\e^{2\pi \mathrm{i}k_3 z}\Biggr|^2\,
\sum_{|k_1| \leq N} a_{k_1}\e^{2\pi \mathrm{i}k_1 x}\Biggr]
\\
-\frac{2}{3}P_N\Biggl[ \int \dd y \,\dd z\, w(x-y)\,w(y-z)\, 
\\
\times
\Biggl|\sum_{|k_2| \leq N} a_{k_2}\e^{2\pi \mathrm{i}k_2 y}\Biggr|^2\,\Biggl|\sum_{|k_3| \leq N} a_{k_3}\e^{2\pi \mathrm{i}k_3 z}\Biggr|^2\,
\sum_{|k_1| \leq N} a_{k_1}\e^{2\pi \mathrm{i}k_1 x}\Biggr]
\end{multline}
equals $\frac{\partial W_N(a)}{\partial \bar{a}_k}$, where
\begin{multline}
\label{W_N}
W_N(a):=-\frac{1}{3} \int \dd x\,\dd y\,\dd z\, w(x-y)\,w(x-z)\,
\\
\times \Biggl|\sum_{|k_1| \leq N} a_{k_1} \e^{2\pi \mathrm{i} k_1 x}\Biggr|^2\,\Biggl|\sum_{|k_2| \leq N} a_{k_2} \e^{2\pi \mathrm{i} k_2 y}\Biggr|^2\,\Biggl|\sum_{|k_3| \leq N} a_{k_3} \e^{2\pi \mathrm{i} k_3 z}\Biggr|^2\,.
\end{multline}
(The terms corresponding to the kinetic energy are easily seen to be the same and we omit the proof).

Expanding the factors of $w$ in \eqref{W_N} as a Fourier series, using $\frac{\partial}{\partial \bar{a}_k} \bar{a}_k=1, \frac{\partial}{\partial \bar{a}_k} a_k=0$, and the assumption that $w$ is even, we compute for $|k| \leq N$
\begin{multline}
\label{W_N_derivative}
\frac{\partial W_N(a)}{\partial \bar{a}_k}=-\frac{1}{3} \mathop{\sum_{k_1,\ldots,k_5}}_{|k_j| \leq N} \sum_{\zeta_1,\zeta_2} \int \dd x\,\dd y\,\dd z\,\widehat{w}(\zeta_1)\,\widehat{w}(\zeta_2)\,a_{k_1}\,a_{k_2}\,\bar{a}_{k_3}\,a_{k_4}\,\bar{a}_{k_5}\,
\\
\times
\e^{2\pi \mathrm{i}(\zeta_1+\zeta_2+k_1-k)x}\,\e^{2\pi \mathrm{i}(-\zeta_1+k_2-k_3)y}\,\e^{2\pi \mathrm{i}(-\zeta_2+k_4-k_5)z}
\,\chi_{|k| \leq N}
\\
-\frac{2}{3} \mathop{\sum_{k_1,\ldots,k_5}}_{|k_j| \leq N} \sum_{\zeta_1,\zeta_2} \int \dd x\,\dd y\,\dd z\,\widehat{w}(\zeta_1)\,\widehat{w}(\zeta_2)\,a_{k_1}\,a_{k_2}\,\bar{a}_{k_3}\,a_{k_4}\,\bar{a}_{k_5}\,
\\
\times
\e^{2\pi \mathrm{i}(\zeta_1+k_1-k)x}\,\e^{2\pi \mathrm{i}(-\zeta_1+\zeta_2+k_2-k_3)y}\,\e^{2\pi \mathrm{i}(-\zeta_2+k_4-k_5)z}
\,\chi_{|k| \leq N}
\\
=-\frac{1}{3}\mathop{\sum_{k_1,\ldots,k_5}}_{|k_j| \leq N} \sum_{\zeta_1,\zeta_2,\zeta_3}\,\widehat{w}(\zeta_1)\,\widehat{w}(\zeta_2)\,a_{k_1}a_{k_2}\bar{a}_{k_3}a_{k_4}\,\bar{a}_{k_5}\,
\\
\times
\delta(\zeta_1+\zeta_2+k_1-k)\,\delta(-\zeta_1+k_2-k_3)\,\delta(-\zeta_2+k_4-k_5)\,
\chi_{|k| \leq N}
\\
-\frac{2}{3}\mathop{\sum_{k_1,\ldots,k_5}}_{|k_j| \leq N} \sum_{\zeta_1,\zeta_2}\,\widehat{w}(\zeta_1)\,\widehat{w}(\zeta_2)\,a_{k_1}a_{k_2}\bar{a}_{k_3}a_{k_4}\,\bar{a}_{k_5}\,
\\
\times
\delta(\zeta_1+k_1-k)\,\delta(-\zeta_1+\zeta_2+k_2-k_3)\,\delta(-\zeta_2+k_4-k_5)\,
\chi_{|k| \leq N}
\,.
\end{multline}

By similar arguments, we can rewrite \eqref{Q_N} as

\begin{multline}
\label{Q_N_2}
Q_N(a)=-\frac{1}{3}P_N\Biggl[\mathop{\sum_{k_1,\ldots,k_5}}_{|k_j| \leq N} \sum_{\zeta_1,\zeta_2}\,\widehat{w}(\zeta_1)\,\widehat{w}(\zeta_2)\,\widehat{w}(\zeta_3)\,a_{k_1}a_{k_2}\bar{a}_{k_3}a_{k_4}\bar{a}_{k_5}
\\
\times
\delta(-\zeta_1+k_2-k_3)\,\delta(-\zeta_2+k_4-k_5)\,
\e^{2\pi \mathrm{i}(\zeta_1+\zeta_2+k_1)x}\Biggr]
\\
-\frac{2}{3}P_N\Biggl[\mathop{\sum_{k_1,\ldots,k_5}}_{|k_j| \leq N} \sum_{\zeta_1,\zeta_2}\,\widehat{w}(\zeta_1)\,\widehat{w}(\zeta_2)\,a_{k_1}a_{k_2}\bar{a}_{k_3}a_{k_4}\bar{a}_{k_5}
\\
\times
\delta(-\zeta_1+\zeta_2+k_2-k_3)\,\delta(-\zeta_2+k_4-k_5)\,
\e^{2\pi \mathrm{i}(\zeta_1+k_1)x}\Biggr]
\,.
\end{multline}
We hence deduce \eqref{Hamiltonian_system_N} by noting that \eqref{W_N_derivative} is indeed equal to the $k$-th Fourier coefficient of \eqref{Q_N_2}.

By construction, the \emph{truncated Gibbs measure}
\begin{equation}
\label{q_gibbs_cutoff_defn_N}
\dd \mbP^f_{\mathrm{Gibbs},N}(a) := \frac{1}{z_{\mathrm{Gibbs},N}^f} \,\e^{-H_N(a)}\, f\Biggl(\sum_{|k| \leq N} |a_k|^2\Biggr)\,\prod_{|k| \leq N} \, \dd a_k\,,
\end{equation}
where 
\begin{equation}
\label{z_gibbs_cutoff_defn_N}
z_{\mathrm{Gibbs},N}^f:=\int \,\prod_{|k| \leq N} \, \dd a_k\,\e^{-H_N(a)}\, f\Biggl(\sum_{|k| \leq N} |a_k|^2\Biggr)
\end{equation}
is invariant under the finite-dimensional Hamiltonian flow \eqref{q_Hartree_Cauchy_problem_N}.
Here, we recall Assumption \ref{f_assumption} and \eqref{H^N}. The normalisation factor \eqref{z_gibbs_cutoff_defn_N} is chosen in such that \eqref{q_gibbs_cutoff_defn_N} is a probability measure.
In order to deduce the invariance stated above, we used the fact that the flow \eqref{q_Hartree_Cauchy_problem_N} conserves mass. This is true because $w$ is real-valued by Assumption \ref{w_assumption} (i).

\textbf{Step 2. Approximation by the finite-dimensional system.} 

Having defined the finite-dimensional approximation \eqref{q_Hartree_Cauchy_problem_N} of \eqref{q_Hartree_Cauchy_problem}, we want to compare the flow of the two. We prove an approximation result, which is an analogue of \cite[Lemma 2.27]{Bou94} proved for the local quintic NLS.
In order to state the claim, we need to introduce some notation. Suppose that $\psi \in C_c^{\infty}(\R)$ is a function such that 
\begin{equation}
\label{function_psi}
\psi(\xi)=
\begin{cases}
1 \text{\, if $|\xi|\leq \frac{1}{2}$}, \\
0 \text{\, if $|\xi|>1$}\,.
\end{cases}
\end{equation}
With $\psi$ as in \eqref{function_psi}, we define the following Fourier multiplier operators.
\begin{equation}
\label{R_N}
(R_N^{-} g)\,\widehat{\,}\,(k):=\psi\biggl(\frac{3k}{N}\biggr)\,\widehat{g}(k)\,,\quad R_N^{+}:=\mathbf{1}-R_N^{-}\,.
\end{equation}
We note the following result, which corresponds to Bourgain's approximation lemma \cite[Lemma 2.27]{Bou94}.

\textbf{Claim (*) :} \emph{Let $A,T>0$ be given. Fix $u_0 \in H^s$ with 
\begin{equation}
\label{u_0_bound_assumption}
\|u_0\|_{H^s} \leq A\,.
\end{equation}
Consider for large $N$ a solution $u^N$ of \eqref{q_Hartree_Cauchy_problem_N} on $[0,T]$ that satisfies
\begin{equation}
\label{approximation_lemma_statement_1}
\sup_{t \in [0,T]} \|u^N(t)\|_{H^s} \leq \mathcal{M} A\,,
\end{equation}
for some constant $\mathcal{M}>0$ independent of $N$.
Then the initial value problem \eqref{q_Hartree_Cauchy_problem} is well-posed on $[0,T]$ and the following approximation bound holds for all $s_1 \in (0,s)$
\begin{equation}
\label{approximation_lemma_statement_2}
\sup_{t \in [0,T]} \|u(t)-u^N(t)\|_{H^{s_1}} \leq C(s,s_1,A,T,w,\mathcal{M})\,\Bigl(\|R_N^{+}w\|_{L^1}+N^{s_1-s}\Bigr)\,,
\end{equation}
provided that the expression on the right-hand side of \eqref{approximation_lemma_statement_2} is strictly less than $1$.
 }

%Claim (*) above is proved by using the multilinear estimate in Lemma \ref{Multlinear_estimates_2} followed by similar arguments as in the proof \cite[Lemma 2.27]{Bou94}. For completeness, we present the details of the proof in Appendix \ref{Proof of the approximation lemma}.

The well-posedness of \eqref{q_Hartree_Cauchy_problem} stated above is interpreted in $H^{s_1}$ for $s_1 \in (0,s)$. This claim follows immediately from \eqref{approximation_lemma_statement_2} and the local well-posedness in $H^{s_1}$ which we obtain from Proposition \ref{Cauchy_problem_1}. We present the details of the proof of Claim (*) in Section \ref{Proof of the approximation lemma} below.

Let us note that \eqref{approximation_lemma_statement_2} is an appropriate bound. It suffices to show that
\begin{equation}
\label{R_N_limit}
\lim_{N \rightarrow \infty} \|R_N^{+}w\|_{L^1}=0\,,
\end{equation}
In order to prove \eqref{R_N_limit}, we first recall \eqref{R_N} and use \cite[VII, Theorem 3.8]{Stein_Weiss} combined with Young's inequality to deduce that 
\begin{equation}
\label{R_N_bound}
\|R_N^{+}\|_{L^1 \rightarrow L^1} \lesssim 1\,,
\end{equation}
uniformly in $N$.
Now, we note that \eqref{R_N_limit} holds if $\widehat{w}(k)=0$ for finitely many $k \in \Z$ (i.e.\ when $\widehat{w}$ is compactly supported). Namely, in this case, we have that $R_N^{+}w=0$ for large enough $N$. By density of such functions in $L^1$ and by \eqref{R_N_bound}, we deduce that \eqref{R_N_limit} for general $w$ as in Assumption \ref{w_assumption} (i).

\textbf{Step 3. Conclusion of the proof.} 
Once we have the approximation result from Step 2, the proof follows identically to that given in \cite[Sections 3--4]{Bou94}. We refer the reader to \cite{Bou94} for details.
\end{proof}

\begin{remark}
\label{Hamiltonian_structure_remark}
In our convention, we write the fields as
\begin{equation*}
u (x)=\sum_{k} (p_k + \mathrm{i}q_k)\,\e^{2\pi \mathrm{i} k x}
\end{equation*}
where $p_k:=\mathrm{Re} \,\widehat{u}(k)\,, q_k:=\mathrm{Im}\,\widehat{u}(k)$. We work with $p_k$ and $q_k$ as the canonical coordinates. Our convention for the Hamiltonian system is 
\begin{equation}
\label{Hamiltonian_system_convention}
\begin{cases}
\dot{p}_k=\frac{1}{2}\,\frac{\partial H}{\partial q_k}\,\\
\dot{q}_k=-\frac{1}{2} \,\frac{\partial H}{\partial p_k}\,,
\end{cases}
\end{equation}
and our convention for the Poisson bracket is 
\begin{equation}
\label{Poisson_bracket_convention}
\{g_1,g_2\}:=\frac{1}{2}\, 
\sum_k \biggl(\frac{\partial g_1}{\partial q_k} \,\frac{\partial g_2}{\partial p_k}-\frac{\partial g_1}{\partial p_k} \,\frac{\partial g_2}{\partial q_k}\biggr)\,.
\end{equation}
We note that, with the Poisson bracket as in \eqref{Poisson_bracket_convention}, we have 
\begin{multline*}
\{u(x),\bar{u}(y)\}=
\\
\frac{1}{2}\, \sum_k \bigl(\mathrm{i} \e^{2\pi \mathrm{i} k x}\,\e^{-2\pi \mathrm{i} k y} 
-\e^{2\pi \mathrm{i} k x} (-\mathrm{i}) \e^{-2\pi \mathrm{i} ky}\bigr)=\mathrm{i} \sum_k \e^{2\pi \mathrm{i} k(x-y)}=\mathrm{i} \delta (x-y)\,.
\end{multline*}
Similarly, we have $\{u(x),u(y)\}=\{\bar{u}(x),\bar{u}(y)\}=0$. Therefore, the Poisson bracket convention in \eqref{Poisson_bracket_convention} coincides with \eqref{Poisson_structure}. 
%For $F$ a smooth function, and for $(p,q)$ solving \eqref{Hamiltonian_system_convention}, we compute using the chain rule
%\begin{equation*}
%\label{dF_dt}
%\frac{\mathrm d}{\mathrm d t} \,F(p,q)=\frac{1}{2}\,\sum_k \biggl(\frac{\partial H}{\partial q_k}\,\frac{\partial F}{\partial p_k}-\frac{\partial H}{\partial p_k}\frac{\partial F}{\partial q_k}\biggr)=\{H,F\}\,.
%\end{equation*}
\end{remark}

\subsection{Proof of the approximation lemma from Step 2}
\label{Proof of the approximation lemma}

In this section, we give the details of the proof of Claim (*) given in Step 2 of the proof of Proposition \ref{Cauchy_problem_2} (ii). In particular, we prove the estimate \eqref{approximation_lemma_statement_2}. Before proceeding with the proof, we prove a multilinear estimate.

\begin{lemma} 
\label{Difference_N}
Let us fix $b=\frac{1}{2}+\varepsilon$ for $\varepsilon>0$ small and $s_1>0$. For $t_0 \in \R$, $\delta>0$, and $v \in X^{s_1,b}_{[0,\delta]}$, and $\mathcal{N}_{1,2}(v,v,v,v,v)$ as in \eqref{nonlinearity_N} with $w_1=w_2=w$, we have
\begin{multline}
\label{Lemma_2.9_bound}
\|\mathcal{N}_{1,2}(v,v,v,v,v)-P_N \mathcal{N}_{1,2}(v,v,v,v,v)\|_{X^{s_1,b-1}_{[t_0,t_0+\delta]}} 
\\
\lesssim \delta^{\varepsilon}\,\|w\|_{L^1}\,\|v\|_{X^{s_1,b}_{[t_0,t_0+\delta]}}^4\, \Bigl(\|R_N^{+}w\|_{L^1} \|v\|_{X^{s_1,b}_{[t_0,t_0+\delta]}}+\|w\|_{L^1} \|R_N^{+} v\|_{X^{s_1,b}_{[t_0,t_0+\delta]}}\Bigr)\,.
\end{multline}
Here, we recall \eqref{R_N}.
\end{lemma}

\begin{proof}
%\textbf{COMPLETE THE PROOF OF LEMMA \ref{Difference_N}.}
Recalling \eqref{N_1_definition} and \eqref{N_2_definition}, we write
\begin{multline}
\label{Difference_N_1}
\mathcal{N}_1(v,v,v,v,v)-P_N \mathcal{N}_1(v,v,v,v,v)=
\\
-\frac{1}{3} \int \dd y\,\dd z\, \Bigl\{ w(x-y)\,w(x-z)\,v(x,t)-P_N \big[w(x-y)\,w(x-z)\,v(x,t)\bigr]\Bigr\}\,
\\
\times
|v(y,t)|^2\,|v(z,t)|^2\,.
\end{multline}
\begin{multline}
\label{Difference_N_1_second_term}
\mathcal{N}_2(v,v,v,v,v)-P_N \mathcal{N}_2(v,v,v,v,v)=
\\
-\frac{2}{3} \int \dd y\,\dd z\, \Bigl\{ w(x-y)\,v(x,t)-P_N \big[w(x-y)\,v(x,t)\bigr]\Bigr\}\,
\\
\times
w(y-z)\,|v(y,t)|^2\,|v(z,t)|^2\,.
\end{multline}
Let us first estimate \eqref{Difference_N_1}. To this end, we write 
\begin{equation}
\label{Difference_N_2}
w=R_N^{+}w+R_N^{-}w\,,\quad v(\cdot,t)=R_N^{+}v(\cdot,t)+R_N^{-}v(\cdot,t)
\end{equation}
for all the terms appearing in the curly brackets in \eqref{Difference_N_1}.
By \eqref{function_psi}--\eqref{R_N}, it follows that 
\begin{equation}
\label{Difference_N_3}
R_N^{-} w(x-y)\, R_N^{-} w(y-z)\,R_N^{-} v(x,t)-P_N \bigl[R_N^{-} w(x-y)\,R_N^{-} w(y-z)\,R_N^{-} v(x,t)\bigr]=0\,.
\end{equation}
Let us note that $R_N^{+}$ and $R_N^{-}$ bounded operators on $L^1$ and $X^{s_1,b}_{[t_0,t_0+\delta]}$, uniformly in $N$. The former follows from \eqref{R_N_bound} and the latter follows directly from \eqref{R_N}. Combining this observation with \eqref{Difference_N_1}, \eqref{Difference_N_2}--\eqref{Difference_N_3}, and Lemma \ref{Multlinear_estimates_2}, we deduce that
\begin{multline}
\label{Lemma_2.9_bound_1}
\|\mathcal{N}_{1}(v,v,v,v,v)-P_N \mathcal{N}_{1}(v,v,v,v,v)\|_{X^{s_1,b-1}_{[t_0,t_0+\delta]}} 
\\
\lesssim \delta^{\varepsilon}\,\|w\|_{L^1}\,\|v\|_{X^{s_1,b}_{[t_0,t_0+\delta]}}^4\, \Bigl(\|R_N^{+}w\|_{L^1} \|v\|_{X^{s_1,b}_{[t_0,t_0+\delta]}}+\|w\|_{L^1} \|R_N^{+} v\|_{X^{s_1,b}_{[t_0,t_0+\delta]}}\Bigr)\,.
\end{multline}
In order to estimate \eqref{Difference_N_1_second_term}, we argue analogously. When considering the term in curly brackets in \eqref{Difference_N_1_second_term}, instead of \eqref{Difference_N_3}, we write
\begin{equation}
\label{Difference_N_3_second_term}
R_N^{-} w(x-y)\,R_N^{-} v(x,t)-P_N \bigl[R_N^{-} w(x-y)\,R_N^{-} v(x,t)\bigr]=0\,.
\end{equation}
Using \eqref{Difference_N_3_second_term}, and arguing as for \eqref{Lemma_2.9_bound_1}, we obtain
\begin{multline}
\label{Lemma_2.9_bound_2}
\|\mathcal{N}_{2}(v,v,v,v,v)-P_N \mathcal{N}_{2}(v,v,v,v,v)\|_{X^{s_1,b-1}_{[t_0,t_0+\delta]}} 
\\
\lesssim \delta^{\varepsilon}\,\|w\|_{L^1}\,\|v\|_{X^{s_1,b}_{[t_0,t_0+\delta]}}^4\, \Bigl(\|R_N^{+}w\|_{L^1} \|v\|_{X^{s_1,b}_{[t_0,t_0+\delta]}}+\|w\|_{L^1} \|R_N^{+} v\|_{X^{s_1,b}_{[t_0,t_0+\delta]}}\Bigr)\,.
\end{multline}
The estimate \eqref{Lemma_2.9_bound} now follows from \eqref{Lemma_2.9_bound_1} and  \eqref{Lemma_2.9_bound_2} by recalling \eqref{nonlinearity_N}.
\end{proof}

\begin{proof}[Proof of Claim (*) from the proof of Proposition \ref{Cauchy_problem_2} (ii)]
By arguing analogously as in the proof of Proposition \ref{Cauchy_problem_1}, we can reduce to the case when $\kappa=0$. 
Let us note that the local well-posedness argument in the proof of Proposition \ref{Cauchy_problem_1} carries over immediately to the finite-dimensional approximation \eqref{q_Hartree_Cauchy_problem_N}. This is because the operator $P_N$ is a contraction on $X^{s,b}$ spaces (and hence on their local variants). In particular, there exists $\delta \sim_A 1$ such that both \eqref{q_Hartree_Cauchy_problem} and \eqref{q_Hartree_Cauchy_problem_N} are well-posed on $[0,\delta]$ in $H^{s_1}$ whenever the initial data is bounded in the $H^{s_1}$ norm by $\mathcal{M} A+1$. By construction, this is possible if we have
\begin{equation}
\label{delta_A_condition}
\delta^{\varepsilon} A^4 \ll 1\,.
\end{equation}

By recalling \eqref{q_Hartree_Cauchy_problem_N}, \eqref{u_0_bound_assumption}, and by using frequency localisation, we have that

\begin{equation}
\label{Claim*_proof_1}
\|u(0)-u^N(0)\|_{H^{s_1}} \lesssim N^{s_1-s}\,A\,.
\end{equation}
Define $U:= u - u^N$. Then, $U$ solves the following difference equation.
\begin{equation}
\label{q_difference_equation}
\begin{cases}
\textrm{i} \partial_t U+ (\Delta - \kappa)U= \mcN_{1,2}(u,u,u,u,u) - P_N\mcN_{1,2}(u^N,u^N,u^N,u^N,u^N), \\
U|_{t=0} = u(0) - u^N(0)\,.
\end{cases}
\end{equation}
Here, we recall \eqref{nonlinearity_N}.
Arguing analogously as in the proof of Proposition \ref{Cauchy_problem_1}, but now in the context of \eqref{q_difference_equation} (instead of \eqref{q_Hartree_Cauchy_problem}), 
and using \eqref{approximation_lemma_statement_1}, it follows that there exists
$T_0 \lesssim \delta$ depending only on $A$ such that for all $N$, we have
\begin{equation}
\label{Claim*_proof_2}
\sup_{t \in [0,T_0]} \|u(t)-u^N(t)\|_{H^{s_1}} \lesssim \|u(0)-u^N(0)\|_{H^{s_1}}\,.
\end{equation}
We note that the implied constant in \eqref{Claim*_proof_2} is independent of $N$.
Let us fix $\sigma_0 \in (0,1)$ arbitrarily small.
From  \eqref{Claim*_proof_1} and \eqref{Claim*_proof_2},  we have that
\begin{equation}
\label{sigma_0_bound}
\sup_{t \in [0,T_0]} \|u(t)-u^N(t)\|_{H^{s_1}} \leq \sigma_0
\end{equation}
for all $N$ large enough (depending on $\sigma_0,A ,s,s_1$). For the remainder of the proof, we consider such $N$.

Combining \eqref{approximation_lemma_statement_1} and \eqref{sigma_0_bound}, it follows that 
\begin{equation}
\label{Claim*_proof_3}
\|u(T_0)\|_{H^{s_1}} \leq \|u^N(T_0)\|_{H^{s_1}}+\sigma_0 \leq \mathcal{M} A+1\,.
\end{equation}
We now introduce the following Cauchy problem.
\begin{equation}
\label{q_Hartree_Cauchy_problem_v^N}
\begin{cases}
\mathrm{i} \partial_t v^N + (\Delta - \kappa)v^N = \mathcal{N}_{1,2}(v^N,v^N,v^N,v^N,v^N)\\
%\\
%-\frac{1}{3}\int \dd y \, \dd z \, w(x-y)\,x(x-z)\,|v^N(y)|^2\,|v^N(z)|^2\,v^N(x)\\
%-\frac{2}{3} \int \dd y \, \dd z \, w(x-y)\,w(y-z)\,|v^N(y)|^2\,|v^N(z)|^2\,v^N(x)
%\\
v^N|_{t=T_0}=u^N(T_0)\,.
\end{cases}
\end{equation}
Let us compare \eqref{q_Hartree_Cauchy_problem_v^N} with the flow of \eqref{q_Hartree_Cauchy_problem} started at $t=T_0$. Again arguing analogously as in the proof of Proposition \ref{Cauchy_problem_1}, but now in the context of \eqref{q_Hartree_Cauchy_problem_v^N}, it follows that we can take $\delta \sim_A 1$ possibly smaller satisfying \eqref{delta_A_condition} and obtain that \eqref{q_Hartree_Cauchy_problem_v^N} has a solution on the time interval $[T_0,T_0+\delta]$ which satisfies
\begin{multline}
\label{Claim*_proof_4}
\sup_{t \in [T_0,T_0+\delta]} \|u(t)-v^N(t)\|_{H^{s_1}} \leq \mathcal{K} \|u(T_0)-v^N(T_0)\|_{H^{s_1}}
\\
=\mathcal{K}\|u(T_0)-u^N(T_0)\|_{H^{s_1}}\leq \mathcal{K} \sigma_0\,,
\end{multline}
for some constant $\mathcal{K}>0$. 
In \eqref{Claim*_proof_4}, we recalled the initial condition in \eqref{q_Hartree_Cauchy_problem_v^N}, as well as \eqref{sigma_0_bound}. 

Let us now compare $v^N(t)$ and $u^N(t)$ for $t \in [T_0,T_0+\delta]$. 
%For the remainder of the proof, we always set $w_1=w_2$ when working with the quantity $\mathcal{N}_{1,2}$ in \eqref{N_1_definition}.
By \eqref{q_Hartree_Cauchy_problem_N}, \eqref{q_Hartree_Cauchy_problem_v^N}, and Duhamel's principle, we have that for all $t \in [T_0,T_0+\delta]$
\begin{equation}
\label{Claim*_proof_5}
v^N(t)-u^N(t)=-\mathrm{i}\int_{T_0}^{t}\dd t'\, \e^{\mathrm{i}(t-t') \Delta}\,\Gamma(\cdot,t')\,,
\end{equation}
where 
\begin{equation}
\label{Claim*_proof_6}
\Gamma:=\mathcal{N}_{1,2}(v^N,v^N,v^N,v^N,v^N)-P_N \mathcal{N}_{1,2}(u^N,u^N,u^N,u^N,u^N)\,.
\end{equation}
Using \eqref{Claim*_proof_5}--\eqref{Claim*_proof_6} and Lemma \ref{X^{s,b}_space_properties} (iii), we deduce that
\begin{multline}
\label{Claim*_proof_7}
\|v^N-u^N\|_{X^{s_1,b}_{[T_0,T_0+\delta]}} \lesssim \|\Gamma\|_{X^{s_1,b-1}_{[T_0,T_0+\delta]}} 
\\
\leq
\bigl\|\mathcal{N}_{1,2}(v^N,v^N,v^N,v^N,v^N)-P_N \mathcal{N}_{1,2}(v^N,v^N,v^N,v^N,v^N)\bigr\|_{X^{s_1,b-1}_{[T_0,T_0+\delta]}} 
\\
+\bigl\|P_N\mathcal{N}_{1,2}(v^N,v^N,v^N,v^N,v^N)-P_N \mathcal{N}_{1,2}(u^N,u^N,u^N,u^N,u^N)\bigr\|_{X^{s_1,b-1}_{[T_0,T_0+\delta]}}
\\
=:I+II\,.
\end{multline}
We estimate the terms $I$ and $II$ in \eqref{Claim*_proof_7} separately.

Let us first estimate $I$. By Lemma \ref{Difference_N}, we have that
\begin{equation}
\label{Claim*_proof_8}
I \lesssim \delta^{\varepsilon}\,\|w\|_{L^1}\,\|v^N\|_{X^{s_1,b}_{[T_0,T_0+\delta]}}^4\, \Bigl(\|R_N^{+}w\|_{L^1} \|v^N\|_{X^{s_1,b}_{[T_0,T_0+\delta]}}+\|w\|_{L^1} \|R_N^{+} v^N\|_{X^{s_1,b}_{[T_0,T_0+\delta]}}\Bigr)\,.
\end{equation}
Let us note that
\begin{equation}
\label{Claim*_proof_9}
\|v^N\|_{X^{s,b}_{[T_0,T_0+\delta]}} \lesssim A\,.
\end{equation}
In order to obtain \eqref{Claim*_proof_9}, we note that for $\delta$ as in \eqref{delta_A_condition}, we can argue as in the proof of Proposition \ref{Cauchy_problem_1} and obtain well-posedness of \eqref{q_Hartree_Cauchy_problem_v^N} on the time interval $[T_0,T_0+\delta]$ since $\|v^N(T_0)\|_{H^s}= \|u^N(T_0)\|_{H^s} \leq \mathcal{M}A$ by \eqref{approximation_lemma_statement_1}.
From \eqref{Claim*_proof_9}, we obtain 
\begin{equation}
\label{Claim*_proof_10}
\|v^N\|_{X^{s_1,b}_{[T_0,T_0+\delta]}} \lesssim A
\end{equation}
and 
\begin{equation}
\label{Claim*_proof_11}
\|R_N^{+} v^N\|_{X^{s_1,b}_{[T_0,T_0+\delta]}} \lesssim N^{s_1-s}\, A\,.
\end{equation}
In order to deduce \eqref{Claim*_proof_11} from \eqref{Claim*_proof_9}, we recall \eqref{R_N} and use frequency localisation.
Combining \eqref{delta_A_condition}, \eqref{Claim*_proof_8}, and \eqref{Claim*_proof_10}--\eqref{Claim*_proof_11}, we obtain that
\begin{equation}
\label{Claim*_proof_12}
I \lesssim \|w\|_{L^1}\, A \Bigl(\|R_N^{+}w\|_{L^1}+N^{s_1-s}\, \|w\|_{L^1}\Bigr)\,.
\end{equation}
Let us now estimate $II$. Since $P_N$ is a contraction on $X^{s_1,b-1}_{[T_0,T_0+\delta]}$, we have that 
\begin{equation}
\label{Claim*_proof_14}
II \leq \bigl\|\mathcal{N}_{1,2}(v^N,v^N,v^N,v^N,v^N)-\mathcal{N}_{1,2}(u^N,u^N,u^N,u^N,u^N)\bigr\|_{X^{s_1,b-1}_{[T_0,T_0+\delta]}}\,.
\end{equation}
By recalling \eqref{N_1_definition}, using multilinearity, and using Lemma \ref{Multlinear_estimates_2}, it follows that the right-hand side of \eqref{Claim*_proof_14} is
\begin{equation}
\label{Claim*_proof_15}
\lesssim \delta^{\varepsilon}\,\|w\|_{L^1}^2\,\|u^N-v^N\|_{X^{s_1,b}_{[T_0,T_0+\delta]}}
\,\Bigl(\|u^N\|_{X^{s_1,b}_{[T_0,T_0+\delta]}}^4+\|v^N\|_{X^{s_1,b}_{[T_0,T_0+\delta]}}^4\Bigr)\,.
\end{equation}
By arguing analogously as for \eqref{Claim*_proof_10}, we have
\begin{equation}
\label{Claim*_proof_16}
\|u^N\|_{X^{s_1,b}_{[T_0,T_0+\delta]}} \lesssim A\,.
\end{equation}
Using \eqref{Claim*_proof_10}, and \eqref{Claim*_proof_16} in \eqref{Claim*_proof_15}, it follows that
\begin{equation}
\label{Claim*_proof_17}
II \lesssim \delta^{\varepsilon} A^4 \,\|w\|_{L^1}^2\,\|u^N-v^N\|_{X^{s_1,b}_{[T_0,T_0+\delta]}}\,.
\end{equation}
We combine \eqref{Claim*_proof_7}, \eqref{Claim*_proof_12}, \eqref{Claim*_proof_17}, and choose $\delta \sim_A 1$ possibly smaller satisfying \eqref{delta_A_condition} to deduce that 
\begin{equation}
\label{Claim*_proof_18}
\|v^N-u^N\|_{X^{s_1,b}_{[T_0,T_0+\delta]}} \lesssim
\|w\|_{L^1}\, A \bigl(\|R_N^{+}w\|_{L^1}+N^{s_1-s}\, \|w\|_{L^1}\bigr)\,.
\end{equation}
Combining \eqref{Claim*_proof_18} and Lemma \ref{X^{s,b}_space_properties} (i), it follows that 
\begin{equation}
\label{Claim*_proof_19}
\sup_{t \in [T_0,T_0+\delta]} \|v^N(t)-u^N(t)\|_{H^s} \lesssim \|w\|_{L^1}\, A \bigl(\|R_N^{+}w\|_{L^1}+N^{s_1-s}\, \|w\|_{L^1}\bigr)\,.
\end{equation}
Using \eqref{Claim*_proof_4} and \eqref{Claim*_proof_19}, it follows that 
\begin{equation}
\label{Claim*_proof_20}
\sup_{t \in [T_0,T_0+\delta]} \|u(t)-u^N(t)\|_{H^s} \leq  C_0\,\|w\|_{L^1}\, A \bigl(\|R_N^{+}w\|_{L^1}+N^{s_1-s}\, \|w\|_{L^1}\bigr) + \mathcal{K} \sigma_0=:\sigma_1\,.
\end{equation}
for some constant $C_0$ (which depends on $s$ and $s_1$, but we suppress this dependence here).

We now iterate this construction. Namely, we start from the time interval $[T_0+(j-1)\delta,T_0+j\delta]$ on which we have 
\begin{equation*}
\sup_{t \in [T_0+(j-1)\delta,T_0+j\delta]} \|u(t)-u^N(t)\|_{H^s} \leq \sigma_j
\end{equation*}
and use the above arguments to deduce that 
\begin{equation*}
\sup_{t \in [T_0+j\delta,T_0+(j+1)\delta]} \|u(t)-u^N(t)\|_{H^s} \leq \sigma_{j+1}\,,
\end{equation*}
where 
\begin{equation}
\label{Claim*_proof_21}
\sigma_{j+1}:=C_0\,\|w\|_{L^1}\, A \bigl(\|R_N^{+}w\|_{L^1}+N^{s_1-s}\, \|w\|_{L^1}\bigr) + \mathcal{K} \sigma_j\,.
\end{equation}
The iteration step is possible provided that 
\begin{equation}
\label{sigma_j_condition}
\sigma_j \leq 1\,.
\end{equation}
(Recall \eqref{sigma_0_bound} and \eqref{Claim*_proof_3} above).

From \eqref{Claim*_proof_21}, by recalling \eqref{R_N_limit}, it follows that \eqref{sigma_j_condition} holds for all $j \leq \lceil T/\delta \rceil$ provided that $\sigma_0$ is chosen sufficiently small and provided that $N$ is chosen large enough. Note that
\begin{equation*}
[0,T] \subset [0,T_0] \cup  \bigcup_{j=1}^{\lceil T/\delta \rceil} [T_0+(j-1)\delta,T_0+j\delta]\,.
\end{equation*}
The claim then follows.

\end{proof}

\section{The time-independent problem with bounded interaction potentials. Proof of Theorem \ref{q_bounded_state_convergence_thm}} %Proof of Theorem \ref{q_bounded_state_convergence_thm}
\label{q_Power_series_expansions_of_the_quantum_and_classical_states}
In this section, we consider $w$ as in Assumption \ref{w_assumption_bounded} above. Our goal is to prove Theorem \ref{q_bounded_state_convergence_thm}. As a preliminary step in the analysis, we argue and expand the quantum and classical states into a power series. As in \cite{RS22}, we note that, due to the presence of the cut-off, the resulting series are analytic in the complex plane. The precise series in the quantum and classical setting are respectively given in \eqref{q_quantum_power_series_defn} and \eqref{q_classical_power_series_defn} below. In Section \ref{q_subsection_basic_estimates}, we state several estimates that will be used in the analysis. In Section \ref{q_subsection_power_series}, we explicitly compute the expansion for the quantum and classical states mentioned above. In Sections \ref{q_subsection_quantum_power_series_bounds} and \ref{q_subsection_classical_power_series_bounds}, we prove bounds on the explicit and remainder terms of the resulting series. In Section \ref{q_subsection_bounded_proof}, we comment on how to use this to prove Theorem \ref{q_bounded_state_convergence_thm}, provided that we have convergence of the untruncated explicit terms given by Proposition \ref{q_expansion_nugeq0}. 
We prove Proposition \ref{q_expansion_nugeq0} by using graphical methods in Section \ref{q_untruncated}.

\subsection{Basic Estimates}
\label{q_subsection_basic_estimates}
Throughout this section, we fix $p \in \mbb{N}^*$ and take $\xi \in \mc{L}(\mfhp)$ unless stated otherwise. We have the following estimates on the quantities $\Theta_{\tau}(\xi)$ and $\Theta(\xi)$ defined in \eqref{q_quantum_lift} and \eqref{q_random_var_defn} respectively.
\begin{lemma}
\label{q_quantum_lift_estimate}
For any $n \in \mbb{N}^*$, we have
\begin{equation*}
\left\|\Th_\tau(\xi)\big|_{\mf{h}^{(n)}}\right\| \leq \left(\frac{n}{\tau}\right)^p \|\xi\|.
\end{equation*}
\end{lemma}
\begin{lemma}
\label{q_random_var_estimate}
We have
\begin{equation*}
\left|\Th(\xi)\right| \leq \|\vph\|^{2p}_{\mf{h}}\|\xi\|.
\end{equation*}
\end{lemma}
Lemma \ref{q_quantum_lift_estimate} follows from \cite[(3.88)]{Kno09}, and Lemma \ref{q_random_var_estimate} is a consequence of \eqref{q_random_var_defn}. We also have the following estimates on the classical interaction, which follows immediately from H\"{o}lder's inequality.
\begin{lemma}
\label{q_classical_interaction_estimate_prop}
Suppose that the classical interaction $\mc{W}$ is defined as in \eqref{q_classical_interaction}. Then for $w \in L^\infty(\La)$, we have
\begin{equation*}
	%\label{q_classical_interaction_estimate_bounded}
|\mc{W}| \leq \frac{1}{3} \|w\|_{L^\infty}^2 \|\vph\|_{L^2}^6.
\end{equation*}
\end{lemma}

We also collect some estimates about Schatten space operators. The following result follows from the spectral decomposition of $|\mc{A}| = \sqrt{\mc{A}^*\mc{A}}$.
\begin{lemma}
\label{q_schatten_embedding}
Let $\mc{H}$ be a separable Hilbert space. Suppose $1 \leq p_1 \leq p_2 \leq \infty$ and $\mc{A} \in \mf{S}^{p_1}(\mc{H}) \cap \mf{S}^{p_2}(\mc{H})$. Then
\begin{equation*}
\|\mc{A}\|_{\mfS^{p_2}(\mc{H})} \leq \|\mc{A}\|_{\mfS^{p_1}(\mc{H})}\,.
\end{equation*}
\end{lemma}
We also have the following form of H\"older's inequality for Schatten spaces \cite{Sim05}.
\begin{proposition}[H\"older's inequality for Schatten spaces]
\label{q_schatten_holder_inequality}
Let $\mc{H}$ be a separable Hilbert space, and let $p,q \in [1,\infty]$ be such that $\frac{1}{p} + \frac{1}{q} = \frac{1}{r}$. Suppose $\mc{A}_1 \in \mf{S}^p(\mc{H})$ and $\mc{A}_2 \in \mf{S}^q(\mc{H})$. Then
\begin{equation*}
\|\mc{A}_1\mc{A}_2\|_{\mf{S}^r(\mc{H})} \leq \|\mc{A}_1\|_{\mfS^{p}(\mc{H})}\, \|\mc{A}_2\|_{\mf{S}^q(\mc{H})}\,.
\end{equation*}
\end{proposition}

\subsection{Power series expansions of the classical and quantum states}
\label{q_subsection_power_series}
In this section, we compute the power series for the quantum and classical states. Let us recall \eqref{q_quantum_state_defn} and \eqref{q_quantum_lift}. We note the following identities.
\begin{equation}
	\label{q_quantum_state_rewrite_1}
\rt(\Th_\tau(\xi)) = \frac{\tr_{\tau,1}(\Th_\tau(\xi))}{\tr_{\tau,1}(\mathbf{1})}\,,
\end{equation}
where for $z \in \C$, we let
\begin{equation}
\label{quantum_state_series} 
\tr_{\tau,z}(\mc{A}) := \frac{1}{\Zf} \,\Tr \left(\mc{A} \,\e^{-\Hf - z \mc{W}_\tau}\, f(\mc{N}_\tau)\right)\,,
\end{equation}
where we recall \eqref{Quantum_interaction}--\eqref{Quantum_free_Hamiltonian}
and note that $\mathbf{1}$ is the identity operator on $\mc{F}$. Let us also define
\begin{equation}
	\label{q_quantum_power_series_defn}
F^\xi_{\tau} (z) := \tr_{\tau,z} (\Th_\tau(\xi))\,.
\end{equation}
By performing a Duhamel expansion up to order $M$, we obtain the following result.
\begin{lemma}
\label{q_quantum_duhamel_expansion_lem}
For $M \in \mbb{N}^*$, we have $F^\xi_\tau(z) = \sum_{m=0}^{M-1} a^\xi_{\tm} z^m + R^\xi_{\tM}(z)$. Here, for $m \in \{0,1,\ldots,M-1\}$, we let
\begin{multline}
	\label{q_quantum_explicit_defn}
a^\xi_{\tm} := \frac{1}{\Zf} \,\mathrm{Tr} \bigg( (-1)^m \int_0^1 \dd t_1 \int_0^{t_1} \dd t_2 \, \cdots \int_0^{t_{m-1}} \dd t_m \, \Th_\tau(\xi) \,\e^{-(1-t_1)\Hf} \,\mc{W}_\tau \\
\times \e^{-(t_1-t_2)\Hf} \cdots \mc{W}_\tau\, \e^{-(t_{m-1}-t_m)\Hf} \,\mc{W}_\tau\, \e^{-t_m\Hf}\, f(\mc{N}_\tau)\bigg)\,,
\end{multline}
and, moreover
\begin{multline}
	\label{q_quantum_remainder_defn}
R^\xi_{\tau,M}(z) := \frac{1}{\Zf} \,\mathrm{Tr} \bigg( (-z)^{M} \int_0^1 \dd t_1 \int_0^{t_1} \dd t_2 \, \cdots \int_0^{t_{m-1}} \dd t_M \, \Th_\tau(\xi)\, \e^{-(1-t_1)\Hf}\, \mc{W}_\tau \\
\times \e^{-(t_1-t_2)\Hf} \cdots \mc{W}_\tau \, \e^{-(t_{M-1}-t_M)\Hf} \,\mc{W}_\tau \,\e^{-t_m(\Hf + z \mc{W}_\tau)} \,f(\mc{N}_\tau)\bigg)\,.
\end{multline}
\end{lemma}
\begin{proof}
We use the identity
\begin{equation*}
\e^{X+Y}=\e^X+\int_{0}^{1}\dd t\, \e^{(1-t)X}\,Y\,\e^{X+tY}
\end{equation*}
$M$ times in \eqref{quantum_state_series}--\eqref{q_quantum_power_series_defn}.
\end{proof}
%Let us define the set
%\begin{equation}
%\label{q_domain_intgration_defn}
%\mfU := \{\mbt \in \mbb{R}^m\ : 0  < t_m < \cdots < t_1  < 1 \}\,.
%\end{equation}
Given $g:\mbb{C} \rightarrow \mbb{C}$, and any operator $\mc{A}: \mc{F} \to \mc{F}$, that commutes with $\mc{N}_{\tau}$, we note that $\mc{A}$ also commutes with $g(\mc{N}_{\tau})$. This is because the operator $g(\mc{N}_{\tau})$ acts on the $n-th$ sector of Fock space as multiplication by $\frac{n}{\tau}$.
In particular, %we have that \eqref{q_quantum_explicit_defn} and \eqref{q_quantum_remainder_defn} commute with operators of the form $g(\mc{N}_{\tau})$. We thus use without future mention 
it follows that that, for every $\alpha>0$, $f^{\alpha}(\mc{N}_\tau)$ commutes with any operators of the form $\mc{W}_\tau$, $\e^{-t\Hf}$, $\e^{-t(\Hf+z\mc{W}_\tau)}$, occurring as factors in the integrands in \eqref{q_quantum_explicit_defn}--\eqref{q_quantum_remainder_defn} above. We use this fact without further mention in the sequel.

Let us recall \eqref{q_classical_state_defn} and \eqref{q_random_var_defn}. By analogy with \eqref{q_quantum_state_rewrite_1}, we rewrite the classical state as
\begin{equation}
\label{q_classical_state_rewrite_1}
\rho(\Th(\xi)) = \frac{\tr_1(\Th(\xi))}{\tr_1(\mathbf{1})}\,,
\end{equation}
where for a random variable $X$ and $z \in \C$, $\tr_z$ is defined as
\begin{equation}
\label{classical_state_series} 
\tr_z(X) := \int \dd \mu \, X \,\e^{-z \mc{W}} \,f(\mc{N})\,.
\end{equation}
For $z \in \C$, we define 
\begin{equation}
	\label{q_classical_power_series_defn}
	F^\xi (z) := \tr_z(\Th(\xi))\,.
\end{equation}
Then, recalling \eqref{q_classical_interaction}, we have the following analogue of Lemma \ref{q_quantum_duhamel_expansion_lem}.
\begin{lemma}
\label{q_classical_duhamel_expansion_lem}
For $M \in \N$, we have $F^\xi(z) = \sum_{m=0}^{M-1} a^\xi_{m} z^m + R^\xi_{M}(z)$. Here
\begin{equation}
	\label{q_classical_explicit_defn}
a^\xi_m := \frac{(-1)^m}{m!} \int \dd \mu \, \Th(\xi)\, \mc{W}^m\, f(\mc{N})\,,
\end{equation}
and
\begin{equation}
\label{q_classical_remainder_defn}
R^\xi_{M}(z) := \frac{(-z)^M}{M!} \int \dd \mu \, \Th(\xi) \,\mc{W}^M\, \e^{-\tilde{z} \, \mc{W}} \,f(\mc{N})\,,
\end{equation}
for some $\tilde{z} \in [0,z]$.
\end{lemma}
\subsection{Analysis of the quantum series \eqref{q_quantum_power_series_defn}}
\label{q_subsection_quantum_power_series_bounds} 
In this section we prove that the explicit and remainder terms defined in \eqref{q_quantum_explicit_defn} and \eqref{q_quantum_remainder_defn} satisfy sufficient bounds for \eqref{q_quantum_power_series_defn} to be analytic.
\begin{lemma}
\label{q_quantum_explicit_bounds_lem}
For any $m \in \mbb{N}$ and $a^\xi_{\tm}$ defined as in \eqref{q_quantum_explicit_defn}, we have
\begin{equation}
\label{q_quantum_explicit_bounds_1}
|a_{\tm}^\xi| \leq \frac{(K^3\|w\|^2_{L^\infty})^mK^p \|\xi\|}{3^m \, m!}.
\end{equation}
\end{lemma}
\begin{proof}
We argue similarly to the proof of \cite[Lemma 3.5]{RS22}. %Recalling \eqref{q_domain_intgration_defn}, 
Proposition \ref{q_schatten_holder_inequality} implies
\begin{multline}
	\label{q_quantum_explicit_rewrite1}
|a_{\tm}^\xi| \leq \frac{1}{\Zf} \int_0^1 \dd t_1 \int^{t_1}_{0} \dd t_2 \cdots \int^{t_{m-1}}_0 \dd t_m  \left\| \Th_\tau(\xi) f^{\frac{1}{m+1}}(\mc{N}_\tau) \right\|_{\mfS^{\infty}(\mfh)} \\
\times \left\| \mc{W}_\tau f^{\frac{1}{m+1}}(\mc{N}_\tau) \right\|^m_{\mfS^{\infty}(\mfh)}
\prod_{j=0}^m \left\|\e^{-(t_j-t_{j+1})H_{\tau,0}}\right\|_{\mfS^{\frac{1}{t_j-t_{j+1}}}(\mfh)}\,,
\end{multline}
where we take the convention $t_0 := 1$ and $t_{m+1} := 0$. Noting that 
\begin{equation*}
\|e^{-s\Hf}\|_{\mfS^{1/s}(\mfh)} = (\Zf)^{s}\,, 
\end{equation*}
since $\e^{-s\Hf}$ is a positive operator, it follows from \eqref{q_quantum_explicit_rewrite1} that
\begin{equation}
\label{q_quantum_explicit_rewrite2}
|a_{\tm}^\xi| \leq \frac{1}{m!} \left\| \Th_\tau(\xi)  f^{\frac{1}{m+1}}(\mc{N}_\tau) \right\|_{\mfS^{\infty}(\mfh)} \left\| \mc{W}_\tau f^{\frac{1}{m+1}}(\mc{N}_\tau) \right\|^m_{\mfS^{\infty}(\mfh)}\,.
\end{equation}
Noting that, by \eqref{Quantum_interaction}, $\mc{W}_\tau$ acts on $\mf{h}^{(n)}$ as multiplication by
\begin{equation*}
-\frac{1}{3 \tau^3}\mathop{\sum_{i,j,k}^{n}}_{i \neq j \neq k \neq i} w(x_i-x_j)\,w(x_i-x_k)\,.
\end{equation*}
Recalling Lemma \ref{q_quantum_lift_estimate}, we have
\begin{equation}
\label{q_mcWf_bound}
\left\| \mc{W}_{\tau} f^\frac{1}{m+1} (\mcN_\tau) \right\|_{\mfS^\infty(\mfh^{(n)})} \leq \frac{1}{3}\left(\frac{n}{\tau}\right)^3 \|w\|_{L^\infty}^2 \left|f\left(\frac{n}{\tau}\right)\right| \leq \frac{1}{3} K^3 \|w\|_{L^\infty}^2\,.
\end{equation}
Here we have also used the support properties of $f$ as in Assumption \ref{f_assumption}, namely \eqref{q_cut-off_support}, and $\|f\|_{L^\infty} \leq 1$. Using Lemma \ref{q_quantum_lift_estimate} and Assumption \ref{f_assumption}, we have
\begin{equation}
\label{q_quantum_lift_cutoff_bound}
\left\| \Th_\tau(\xi) f^{\frac{1}{m+1}}(\mc{N}_\tau) \right\|_{\mfS^{\infty}(\mfh)} \leq K^p \|\xi\|\,.
\end{equation}
Then \eqref{q_quantum_explicit_bounds_1} follows from \eqref{q_quantum_explicit_rewrite2}, \eqref{q_mcWf_bound}, and \eqref{q_quantum_lift_cutoff_bound}.
\end{proof}
In order to estimate $R^\xi_{\tau,M}(z)$, we apply the Feynman-Kac formula, which we recall now. Let $T > 0$, and let $\Om^T$ be the space of continuous paths $\om : [0,T] \rightarrow \La$. 
Given $x,x' \in \La$, we let $\Omega^T_{x,x'}$ denote the set of all paths $\omega \in \Omega^T$ such that $\omega(0)=x'$ and $\omega(T)=x$.
For $x,x' \in \La$, we characterise the Wiener measure $\mbb{W}_{x,x'}^T$ on $\Omega^T_{x,x'}$ by its finite-dimensional distribution. More precisely, if $0 < t_1 < \cdots < t_n < T$ and $g : \La^n \rightarrow \R$ is continuous, we have
\begin{multline*}
\int \mbb{W}_{x,x'}^T(\dd \om) \, g(\om(t_1,),\cdots,\om(t_n)) \\
= \int \dd x_1 \cdots \dd x_n \, \e^{t_1(x_1-x')\Delta}\,\e^{(t_2-t_1)(x_2 - x_1)\Delta}\cdots \\ 
\times \e^{(t_n-t_{n-1})(x_n - x_{n-1})\Delta}\, \e^{(T-t_n)(x - x_n)\Delta}\,g(x_1,\ldots,x_n)\,.
\end{multline*}
Here for $t > 0$, we define the heat periodic heat kernel on $\La$ as
\begin{equation*}
\label{q_heat_kernel}
\e^{t\Delta}(y) = \sum_{k \in \Z} \frac{1}{(4 \pi t)^{1/2}} \e^{-|y-k|^2/4t}\,.
\end{equation*}
The following formula holds \cite[Theorem X.68]{RS75}.
\begin{proposition}
\label{q_FK_formula}
Suppose $V: \La \rightarrow \mbb{C}$ is continuous and bounded below. Then for $t > 0$
\begin{equation*}
\e^{-t(V-\Delta)}(x;x') = \int \mbb{W}^t_{x,x'} (\dd \om) \, \e^{-\int^t_0 \dd s \, V(\om(s))}\,.
\end{equation*}
\end{proposition}
We now use Proposition \ref{q_FK_formula} to bound the quantum remainder term \eqref{q_quantum_remainder_defn}.
\begin{lemma}
	\label{q_quantum_remainder_bounds_lem}
For any $M \in \mbb{N}$ and $R^\xi_{\tM}(z)$ defined as in \eqref{q_quantum_remainder_defn}, we have
\begin{equation}
		\label{q_quantum_remainder_bounds_1}
\left|R_{\tM}^\xi(z)\right| \leq \e^{\frac{1}{3}K^3 |\mathrm{Re}(z)|\|w\|^2_{L^\infty}} \frac{(K^3\|w\|^2_{L^\infty})^M K^p \|\xi\|}{3^M \, M!}\,|z|^M\,.
\end{equation}
\end{lemma}
\begin{proof}
By arguing similarly as the proof of \cite[Lemma 3.6 and (3.22)]{RS22}, it suffices to show that for $t \in [0,1]$, we have
\begin{equation}
\label{q_quantum_remainder_needed}
\left| \left(\e^{-t(\Hf + z\mc{W}_\tau)}f^{\frac{1}{2}}(\mc{N}_\tau)\right)^{(n)}(\mbx;\mby)\right| \leq \e^{\frac{1}{3} K^3|\rm{Re}(z)|\|w\|^2_{L^\infty}}\left(\e^{-t\Hf}\right)^{(n)}(\mbx;\mby)\,,
\end{equation}
where $\mc{A}^{(n)}$ denotes the kernel of $\mc{A}$ restricted to the $n^{th}$ sector of Fock space. Noting that by \eqref{Quantum_interaction}, we have
\begin{equation*}
\left(\mc{W}_\tau\right)^{(n)}(\mbx;\mby) = -\frac{1}{3\tau^3} \mathop{\sum_{i,j,k}^{n}}_{i \neq j \neq k \neq i} w(x_i-x_j)\,w(x_i-x_k)\, \prod_{l=1}^n \delta(x_l-y_l)\,,
\end{equation*}
we can rewrite $\left(\e^{-t(\Hf + z\mc{W}_\tau)}f^{\frac{1}{2}}(\mc{N}_\tau)\right)^{(n)}(\mbx;\mby)$ using Proposition \ref{q_FK_formula} as 
\begin{equation}
\label{q_exp_rewrite_1}
\int \mbb{W}_{\mbx,\mby}^t (\dd \widetilde{\om}) \, \exp\biggl\{-\frac{t \kappa n }{\tau} + \frac{z}{3 \tau^3}\,\int_0^t \dd s {\sum^n_{ i \neq j \neq k \neq i}} w_{i,j}(\widetilde{\om}(s))\,w_{i,k}(\widetilde{\om}(s))\biggr\}\, f^{\frac{1}{2}}\left( \frac{n}{\tau} \right)\,.
\end{equation}
In \eqref{q_exp_rewrite_1}, we have written 
\begin{equation}
\label{q_exp_rewrite_1_B}
\widetilde{\om} \equiv (\omega_1,\ldots,\omega_n) \in \prod_{l=1}^{n} \Omega^t_{x_l,y_l}\,\quad
\mbb{W}^t_{\mbx,\mby} (\dd \widetilde{\om}) \equiv \prod_{l=1}^n \mbb{W}^t_{x_l,y_l}(\dd \om_{l})\,,
\end{equation}
and where we have abbreviated $w_{i,j}(\mbx) \equiv w(x_i-x_j)$, as well as 
\begin{equation*}
\sum^n_{ i \neq j \neq k \neq i} (\cdots) \;\equiv \mathop{\sum_{i,j,k}^{n}}_{i \neq j \neq k \neq i}(\cdots)\,,
\end{equation*}
all of which we use in the sequel.

Using the triangle inequality, it follows that \eqref{q_exp_rewrite_1} is in absolute value
\begin{equation}
\label{q_exp_rewrite_2}
\leq \int \mbb{W}_{\mbx,\mby}^t (\dd \widetilde{\om}) \,
\, \biggl|\exp\biggl\{-\frac{t \kappa n }{\tau} + \frac{z}{3 \tau^3}\,\int_0^t \dd s {\sum^n_{ i \neq j \neq k \neq i}} w_{i,j}(\widetilde{\om}(s))\,w_{i,k}(\widetilde{\om}(s))\biggr\}\biggr|\, f^{\frac{1}{2}}\left(\frac{n}{\tau}\right)\,.
\end{equation}
Using Proposition \ref{q_FK_formula} once more,  we deduce that \eqref{q_exp_rewrite_2} is 
\begin{equation}
\label{q_exp_rewrite_3}
\leq \sup_{\widetilde{\om}}\,  \biggl|\exp\biggl\{\frac{z}{3 \tau^3}\,\int_0^t \dd s {\sum^n_{ i \neq j \neq k \neq i}} w_{i,j}(\widetilde{\om}(s))\,w_{i,k}(\widetilde{\om}(s))\biggr\}\biggr|\,f^{\frac{1}{2}}\left(\frac{n}{\tau}\right)\, \left(\e^{-t\Hf}\right)^{(n)}(\mbx;\mby)\,,
\end{equation}
where the supremum is taken over all $\tilde{\omega}$ as in \eqref{q_exp_rewrite_1_B} above.
Arguing as in \eqref{q_mcWf_bound}, we have
\begin{equation}
\label{q_exp_rewrite_4}
 \sup_{\widetilde{\om}}\,  \biggl|\exp\biggl\{\frac{z}{3 \tau^3}\,\int_0^t \dd s {\sum^n_{ i \neq j \neq k \neq i}} w_{i,j}(\widetilde{\om}(s))\,w_{i,k}(\widetilde{\om}(s))\biggr\}\biggr|\,f^{\frac{1}{2}}\left(\frac{n}{\tau}\right) \leq \e^{\frac{1}{3}K^3 |\mathrm{Re}(z)| \|w\|_{L^\infty}^2}\,.
\end{equation}
Combining \eqref{q_exp_rewrite_3} and \eqref{q_exp_rewrite_4}, we obtain \eqref{q_quantum_remainder_needed}, thus completing the proof.
\end{proof}
Combining the bounds proved in Lemmas \ref{q_quantum_explicit_bounds_lem} and \ref{q_quantum_remainder_bounds_lem} with Taylor's theorem yields the following corollary.
\begin{corollary}
\label{q_quantum_power_series_analytic}
$F^\xi_{\tau}(z) = \sum_{m=0}^\infty a_{\tm}^\xi z^m$ is analytic on the whole of $\mbb{C}$.
\end{corollary}
\subsection{Analysis of the classical series \eqref{q_classical_power_series_defn}}
\label{q_subsection_classical_power_series_bounds}
We now prove that the explicit and remainder terms defined in \eqref{q_classical_explicit_defn} and \eqref{q_classical_remainder_defn} satisfy sufficient bounds for the function defined in \eqref{q_classical_power_series_defn} to be analytic on $\mbb{C}$.

\begin{lemma}
\label{q_classical_explicit_bounds_lemma}
Let $m \in \N$ and $a^\xi_m$ be defined as in \eqref{q_classical_explicit_defn}. Then
\begin{equation}
\left| a^\xi_{m} \right| \leq \frac{(K^3\|w\|_{L^\infty}^2)^m K^p \|\xi\|}{3^m \, m!}\,.
\end{equation}
\end{lemma}
\begin{proof}
Using Lemma \ref{q_random_var_estimate} as in the proof of \cite[Lemma 3.8]{RS22}, it is sufficient to prove that
\begin{equation}
\label{q_classical_explicit_needed}
\left|\mc{W} f^{\frac{1}{m+1}}(\mc{N})\right| \leq \frac{1}{3} K^3 \|w\|_{L^\infty}^2\,.
\end{equation}
Using Lemma \ref{q_classical_interaction_estimate_prop} and \eqref{q_cut-off_support}, we have
\begin{equation*}
\left|\mc{W} f^{\frac{1}{m+1}}(\mc{N})\right| \leq \frac{1}{3} K^3\|w\|_{L^\infty}^2 \|f^{\frac{1}{m+1}}\|_{L^\infty}\,.
\end{equation*}
Here we recall that $\mc{N} = \|\vph\|_{\mfh}^2$. Noting that $\|f\|_{L^\infty} \leq 1$, \eqref{q_classical_explicit_needed} follows.
\end{proof}
\begin{lemma}
\label{q_classical_remainder_lemma}
Let $M \in \N$ and $R^\xi_m(z)$ be defined as in \eqref{q_classical_remainder_defn}. Then
\begin{equation}
	\label{q_classical_remainder_bound}
\left|R^\xi_M(z)\right| \leq \e^{\frac{1}{3}K^3 \left|\mathrm{Re}(z)\right|\|w\|_{L^\infty}^2} \frac{\left(K^3 \|w\|_{L^\infty}^2\right)^MK^p \|\xi\|}{M! \, 3^M}\,|z|^M\,.
\end{equation}
\end{lemma}
\begin{proof}
We note that Lemma \ref{q_classical_interaction_estimate_prop} and \eqref{q_cut-off_support} imply that for any $\tilde{z} \in [0,z]$, we have
\begin{equation*}
\left|\e^{-z \mc{W}}f^{\frac{1}{M+2}}(\mc{N})\right| \leq \e^{\frac{1}{3} K^3 \left|\mathrm{Re}(z)\right| \|w\|_{L^\infty}^2}.
\end{equation*}
Recalling \eqref{q_classical_remainder_defn} and using Lemma \ref{q_classical_explicit_bounds_lemma}, we obtain \eqref{q_classical_remainder_bound}. 
\end{proof}
Combining Lemmas \ref{q_classical_explicit_bounds_lemma} and \ref{q_classical_remainder_lemma}, we have the following corollary.
\begin{corollary}
\label{q_classical_power_series_analytic}
$F^\xi(z) = \sum_{m=0}^\infty a_{m}^\xi z^m$ is analytic on the whole of $\mbb{C}$.
\end{corollary}
\subsection{Proof of Theorem \ref{q_bounded_state_convergence_thm}}

We note the following result, whose proof we defer to Section \ref{q_untruncated} below.
\begin{proposition}
	\label{q_expansion_nugeq0}
Let $\nu > 0$ be fixed. Let $\mcCp$ be as in \eqref{q_mathcalC_defn}. Define
\begin{align}
\notag
\alpha^{\xi,\nu}_{\tau,m} := \frac{(-1)^m}{Z_{\tau,0}} &\mathrm{Tr}\, \bigg( \int^1_0 \dd t_1 \int_0^{t_1} \dd t_2 \cdots \int_0^{t_{m-1}} \dd t_m \, \Th_\tau(\xi)\, \e^{-(1-t_1)\HnuN} \,\mcW_\tau \\
\notag
&\times \e^{-(t_1-t_2)\HnuN} \,\mcW_\tau \, \e^{-(t_2-t_3)\HnuN} \cdots \\
\label{alpha_1}
&\times \e^{-(t_{m-1}-t_m)\HnuN}\, \mcW_\tau\, \e^{-t_M\HnuN}\bigg), \\
\label{alpha_2}
\alpha_{m}^{\xi,\nu} := \frac{(-1)^m}{m!} &\int \dd \mu \, \Th(\xi)\,\mathcal{W}^m\, \e^{-\nu\mcN}.
\end{align}
Then, the following results hold uniformly in $\xi \in \mathcal{C}_p$.
\begin{enumerate}
	\item[(i)] $\left|\alpha_{\#,m}^{\xi,\nu}\right| \leq C(m,p,\nu)$.
	\item[(ii)] $\lim_{\tau \to \infty} \alpha^{\xi,\nu}_{\tau,m} = \alpha_{m}^{\xi,\nu}$.
\end{enumerate}
\end{proposition}

\label{q_subsection_bounded_proof}
\begin{proof}[Proof of Theorem \ref{q_bounded_state_convergence_thm}]
By using Proposition \ref{q_expansion_nugeq0} and by arguing analogously as for \cite[Lemma 3.13]{RS22}, we deduce that for $m \in \N$, we have
\begin{equation}
\label{explicit_convergence_RS22}
\lim_{\tau \to \infty}a^{\xi}_{\tau,m} = a^{\xi}_{m}\,,
\end{equation}
uniformly in $\xi \in\mathcal{C}_p$.
We then combine \eqref{explicit_convergence_RS22} with Corollary \ref{q_quantum_power_series_analytic}, Corollary \ref{q_classical_power_series_analytic}, Lemma \ref{q_quantum_explicit_bounds_lem}, Lemma \ref{q_classical_explicit_bounds_lemma}, and the dominated convergence theorem to deduce that for all $z \in \C$, we have
\begin{equation}
\label{explicit_convergence_RS22_2}
\lim_{\tau \rightarrow \infty} \sup_{\xi \in \mathcal{C}_{p}}\bigl|F^{\xi}_{\tau}(z) - F^{\xi}(z)\bigr| \leq \lim_{\tau \rightarrow \infty} \sum_m \sup_{{\xi} \in \mathcal{C}_p}| a^{\xi}_{\tau,m} - a^{\xi}_m| |z|^{m} = 0\,.
\end{equation}

The claim \eqref{partition_function_convergence_bounded} follows from \eqref{explicit_convergence_RS22_2} by taking $p=0$ and recalling \eqref{q_classical_partn_functn}, \eqref{q_quantum_partition_functions}--\eqref{q_relative quantum_partition_fn}, \eqref{quantum_state_series}--\eqref{q_quantum_power_series_defn}, and \eqref{classical_state_series}--\eqref{q_classical_power_series_defn} above.
Note that by convention, when $p=0$, there is no observable $\xi$ in the analysis above. Moreover, we take $\mathcal{A}=\mathbf{1}$ in \eqref{quantum_state_series} and $X=1$ in \eqref{classical_state_series},  respectively.

The claim \eqref{correlation_function_convergence_bounded}  follows from \eqref{explicit_convergence_RS22_2} by a duality argument. %analogously as in \cite[Section 3.7]{RS22}. 
More precisely, by using \eqref{q_quantum_state_defn} and \eqref{q_quantum_correlation_fn_defn}--\eqref{q_quantum_lift}, we have that for all $\xi \in \mathcal{L}(\mfhp)$
\begin{equation}
\label{duality_1}
\rho_{\tau}(\Theta_{\tau}(\xi))=\mathrm{Tr}\,(\gamma_{\tau,p}\,\xi)\,.
\end{equation}
Analogously, by using \eqref{q_classical_state_defn}--\eqref{q_random_var_defn}, we have 
\begin{equation}
\label{duality_2}
\rho(\Theta(\xi))=\mathrm{Tr}\,(\gamma_{p}\,\xi)\,.
\end{equation}
Recalling \eqref{q_quantum_state_rewrite_1} and \eqref{q_classical_state_rewrite_1}, we see that \eqref{explicit_convergence_RS22_2} implies
\begin{equation}
\label{duality_3}
\lim_{\tau \rightarrow \infty} \sup_{\xi \in \mathcal{C}_p} \bigl|\rho_{\tau}(\Theta_{\tau}(\xi))-\rho(\Theta(\xi))\bigr|=0\,.
\end{equation}

From \eqref{duality_1}--\eqref{duality_2} (and taking suprema over $\xi \in \mathfrak{B}_p$), we immediately deduce the weaker analogue of \eqref{correlation_function_convergence_bounded}  given by 
\begin{equation}
\label{correlation_function_convergence_bounded_weaker}
\lim_{\tau \rightarrow \infty} \|\gamma_{\tau,p} - \gamma_p\|_{\mfS^2(\mfhp)} = 0\,.
\end{equation}
We upgrade \eqref{correlation_function_convergence_bounded_weaker} to 
\eqref{correlation_function_convergence_bounded} by noting that $\gamma_{\tau,p} \geq 0$ and $\gamma_{p} \geq 0$ in the sense of operators and that $\lim_{\tau \rightarrow \infty} \mathrm{Tr} \gamma_{\tau,p} = \mathrm{Tr} \gamma_p$. The former claim follows from \cite[Lemma 3.17]{RS22}, whose proof carries over directly to the quintic setting. The latter claim from \eqref{duality_1}--\eqref{duality_3} by taking $\xi=\mathbf{1}_p \in \mathcal{C}_p$.
For the details of the last step, we refer the reader to \cite[Lemma 4.10]{FKSS17} and \cite[Lemma 3.19]{RS22}.
\end{proof}

\subsection{Graphical analysis of the untruncated explicit terms. Proof of Proposition \ref{q_expansion_nugeq0}}
\label{q_untruncated}

In this section, we prove Proposition \ref{q_expansion_nugeq0} stated above.
For the proof, we use a graphical argument similar to that used in \cite[Sections 2.3-2.6 and 4.1]{FKSS17}. The graphs will be different, due to the three-body interaction. The essence of the argument is quite similar.
For completeness, we review the proof and refer the reader to \cite{FKSS17} for more details and motivation.
%, normal ordering, and the one-dimensional setting of the problem.  Following the proof in \cite[Lemma 3.11]{RS22}, it suffices to prove the Proposition \ref{q_expansion_nugeq0} in the case $\nu = 0$. 
In what follows, we denote
\begin{equation}
\label{q_bxi_defn_remk}
b^{\xi}_{\#,m}:=\alpha^{\xi,0}_{\#,m}\,,
\end{equation}
where we recall \eqref{alpha_1}--\eqref{alpha_2}.
Let us note that it suffices to show Proposition \ref{q_expansion_nugeq0}, when $\nu=0$. The general claim follows from this one (with possibly different constants depending on $\nu$) by replacing $\kappa$ in \eqref{q_1-body_Hamn_defn} with $\kappa+\nu$ (see the proof of \cite[Lemma 3.11]{RS23} for more details). In Subsection \ref{Subsection_3.6.1}, we show that 
\begin{equation}
\label{Proposition_3.15_1_1}
\left|b^{\xi}_{\tau,m}\right| \leq C(m,p)\,,
\end{equation}
uniformly in $\xi \in \mathfrak{B}_p$.
In Subsection \ref{Subsection_3.6.2}, we show that 
\begin{equation}
\label{Proposition_3.15_1_2}
\lim_{\tau \to \infty} b^{\xi}_{\tau,m} =b^{\xi}_m\,,
\end{equation}
uniformly in $\xi \in \mathfrak{B}_p$.
In Subsection \ref{Subsection_3.6.3}, we show that \eqref{Proposition_3.15_1_1}--\eqref{Proposition_3.15_1_2} hold for $\xi=\mathbf{1}_p$. The latter requires a slightly modified graphical structure.
Putting these steps together, we complete the proof of Proposition \ref{q_expansion_nugeq0}.
\subsubsection{Proof of \eqref{Proposition_3.15_1_1} uniformly in $\xi \in \mathfrak{B}_p$}
\label{Subsection_3.6.1}
%In this subsection, we show the Proposition \ref{q_expansion_nugeq0} (i) in the quantum setting when $\xi \in \mathfrak{B}_p$ (here we recall \eqref{q_mathfrakB_defn}).
We begin by recalling a number of definitions and results from \cite{FKSS17}. Throughout, we abbreviate $\vph_{\tau,k} := \vph_\tau(e_{k})$, where the $e_{k}$ are defined as in \eqref{q_eigenvalues_defn} above.
\begin{definition}
For $t \in \R$, we define the operator valued distributions $\left(\e^{th/\tau} \vph_\tau\right)(x)$ and $\left(\e^{th/\tau} \vph_\tau^*\right)(x)$ as
\begin{align*}
\left(\e^{th/\tau} \vph_\tau\right)(x) &:= \sum_{k \in \N} \e^{t\lambda_{k}/\tau}e_{k}(x)\vph_{\tau,k}\,, \\
\left(\e^{th/\tau} \vph^*_\tau\right)(x) &:= \sum_{k \in \N} \e^{t\lambda_{k}/\tau}\overline{e}_{k}(x)\vph^*_{\tau,k}\,.
\end{align*}
Here, we recall \eqref{lambda_k}.
\end{definition}
In what follows, we use the result of \cite[Lemma 2.3]{FKSS17}.
\begin{lemma}
	\label{q_conjugation_creatann}
For $t \in \R$ we have,
\begin{equation*}
\e^{tH_{\tau,0}} \vph^*_\tau(x)\e^{-tH_{\tau,0}} = \left(\e^{th/\tau}\vph^*_\tau\right)(x), \quad \e^{tH_{\tau,0}} \vph_\tau(x)\e^{-tH_{\tau,0}} = \left(\e^{-th/\tau}\vph_\tau\right)(x)\,.
\end{equation*}
Here, we recall the definition \eqref{Quantum_free_Hamiltonian} of $H_{\tau,0}$.
\end{lemma}
The following result follows from Lemma \ref{q_conjugation_creatann} and the definition \eqref{Quantum_interaction}
of $\mcW_{\tau}$; see also \cite[Corollary 2.4]{FKSS17}.
\begin{lemma}
	\label{q_conjugated_mcWt}
For $t \in \R$, we have
\begin{multline}
	\label{q_conjugated_mcWt_equation}
\e^{tH_{\tau,0}} \,\mcWt\, \e^{-tH_{\tau,0}} = -\frac{1}{3} \int \dd x \, \dd y \, \dd z \,  \left(\e^{th/\tau} \vph^*_\tau\right)(x)\left(\e^{th/\tau} \vphts\right)(y) \left(\e^{th/\tau} \vph^*_\tau\right)(z) \\
\times w(x-y)\, w(x-z)\, \left(\e^{-th/\tau} \vpht\right)(x)\left(\e^{-th/\tau} \vph_\tau\right)(y)\left(\e^{-th/\tau} \vph_\tau\right)(z)\,.
\end{multline}
\end{lemma}
Given a closed operator $\mathcal{A}$ on $\mathcal{F}$, we define
\begin{equation}
\label{rho_{tau,0}}
\rho_{\tau,0}(\mathcal{A}) :=\frac{\Tr(\mathcal{A}\e^{-H_{\tau,0}})}{\Tr(\e^{-H_{\tau,0}})}\,.
\end{equation}
Furthermore, let us recall \eqref{q_bxi_defn_remk} and \eqref{alpha_1}.
Using Lemma \ref{q_conjugated_mcWt} and the cyclicity of the trace, it follows that for all $m,p \in \N$ and $\xi \in \mathfrak{B}_p$, we have
\begin{equation}
	\label{q_explicit_b_rewrite_f1}
b^\xi_{\tau,m} = \frac{1}{3^m} \int^1_0 \dd t_1 \int^{t_1}_0 \dd t_2 \cdots \int^{t_{m-1}}_0 \dd t_m \, g_{\tau,m}^\xi(\mbt)\,, 
\end{equation}
where $\mbt := (t_1,\ldots,t_m)$ and
\begin{multline}
\label{q_gtau_defn}
g^\xi_{\tau,m}(\mbt) := \int \dd x_1 \cdots \dd x_{m+p} \, \dd y_1 \cdots \dd y_{m+p} \, \dd z_1 \cdots \dd z_m \\
\times \left(\prod_{i=1}^m w(x_i-y_i)\, w(x_i-z_i) \right)
\xi(x_{m+1},\ldots,x_{m+p};y_{m+1},\ldots,y_{m+p}) \\
\times \rho_{\tau,0}\Biggl( \prod_{i=1}^m \bigg[ \left(\e^{t_i h_\tau/\tau} \vphts\right)(x_i)\left(\e^{t_i h_\tau/\tau} \vphts\right)(y_i) \left(\e^{t_i h_\tau/\tau} \vphts\right)(z_i) \\
\times \left(\e^{-t_i h_\tau/\tau} \vpht\right)(x_i)\left(\e^{-t_i h_\tau/\tau} \vpht\right)(y_i) \left(\e^{-t_i h_\tau/\tau} \vpht\right)(z_i) \bigg] \\
\times \prod_{i=1}^p \vphts(x_{m+i}) \prod_{i=1}^p \vpht(y_{m+i}) \Biggr)\,.
\end{multline}

Here and throughout all of our products are taken in the order of increasing indices.
We now fix $m,p \in \N$ and define an abstract vertex set $\Sigma$ containing $(6m+2p)$ elements as follows.
\begin{definition}
	\label{q_vertex_set_defn}
Given $m,p \in \N$, we define $\Sigma \equiv \Sigma(m,p)$ to be the set of triples $(i,r,\delta)$ with $i \in \{1,\ldots,m+1\}$. If $i \in \{1,\ldots,m\}$, we consider $r \in \{1,2,3\}$ and if $i = m+1$, we take $r \in \{1,\ldots, p\}$. Finally, we take $\delta \in \{\pm 1\}$. We write\footnote{We emphasise that this is a different object than the $\alpha_{\#,m}^{\xi,\nu}$ in \eqref{alpha_1}--\eqref{alpha_2} above.}  $\alpha = (i,r,\delta)$ and write the components of $\alpha$ as $i_\alpha,r_\alpha,\delta_\alpha$.
We use the lexicographical ordering on $\Sigma$ to order the vertices, which we denote by $\leq$. If $\alpha \leq \beta$ and $\alpha \neq \beta$, we write $\alpha < \beta$.
To each vertex $\alpha = (i,r,\delta) \in \Sigma$, we assign a spatial integration variable $x_\alpha$. Moreover, to each $i \in \{1,\ldots,m\}$, we assign a time $t_i$, and take $t_{m+1} := 0$ by convention. Where convenient, we write $x_{i,r,\delta}$ or $t_{i,r,\delta}$ instead of $x_\alpha$ or $t_\alpha$ respectively. Let us write
\begin{equation*}
\mbx := (x_\alpha)_{\alpha \in \Sigma} \in \Lambda^{\Sigma}, \qquad \mbt := (t_{\alpha})_{\alpha \in \Sigma} \in \R^{\Sigma}\,.
\end{equation*} 
We only consider $(t_1,\ldots, t_m)$ to be in the support of the integral from \eqref{q_explicit_b_rewrite_f1}. In other words, we always take $\mbt \in \mathfrak{V} \equiv \mathfrak{V}(m)$, where
\begin{equation}
\label{V_simplex}
\mathfrak{V} := \bigl\{\mbt \in \R^{\Sigma} : t_{i,r,\delta} = t_i \textrm{ with } 0= t_{m+1} < t_{m} < \cdots < t_1 < 1\bigr\}\,.
\end{equation}
\end{definition}
In the discussion that follows, we fix $m,p \in \N$ as well as $\xi \in \mathfrak{B}_p$.

\begin{remark}
We interpret the integrand in \eqref{q_gtau_defn} in terms of the set $\Sigma$ in Definition \ref{q_vertex_set_defn} as follows. Each occurrence of $ \vphts(\cdot)$ or $\vphts(\cdot)$ corresponds to an element of $\Sigma$. For $i \in \{1, \ldots, m\}$, $i$ denotes that we are working with the $i^{th}$ interaction, and hence that we are considering a factor of the form $\left(\e^{t_i h_\tau/\tau} \vphts\right)(\cdot)$ or  $\left(\e^{-t_i h_\tau/\tau} \vpht\right)(\cdot)$. When $\delta=+1$, it is the former and when $\delta=-1$, it is the latter. The index $r=1,2,3$ refers to the integration variable $x_i,y_i,z_i$ respectively.
Furthermore, when $i=m+1$, we consider the factor $\vphts(x_{m+r})$ and when $\delta=+1$, and the factor $\vpht(y_{m+r})$ when $\delta=-1$; see Figure \ref{Fig:uncollapsed} for a graphical representation.
Let us also note that
\begin{equation}
	\label{q_time_variable_identity}
\alpha < \beta \quad \Rightarrow \quad 0 \leq t_{\alpha} - t_{\beta} < 1\,.
\end{equation}
\end{remark}

\begin{figure}[htbp]
\centering
\includegraphics[scale=0.75]{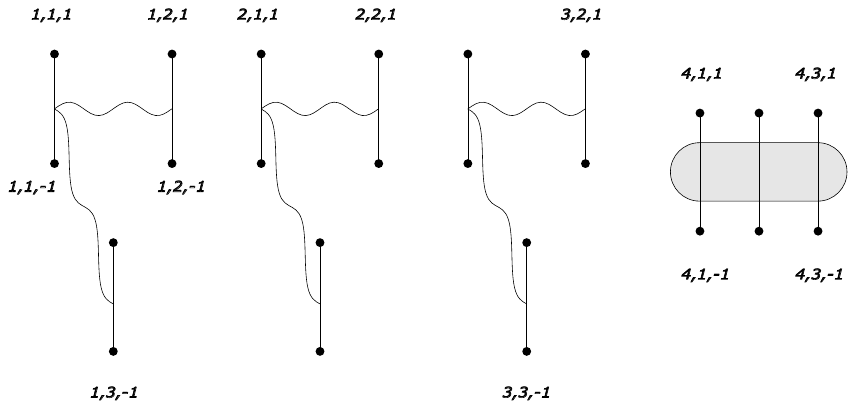}
\vspace*{-160mm}
\caption{An unpaired graph from Definition \ref{q_vertex_set_defn} with $m=p=3$. The black dots correspond to factors of $\vpht(\cdot)$ and $\vphts(\cdot)$. The wavy lines correspond to factors of the interaction potential $w$.}
\label{Fig:uncollapsed}
\end{figure}

Let us recall the \emph{Quantum Wick theorem}.

\begin{lemma}[Quantum Wick theorem]
\label{Quantum_Wick_theorem}
Let $\mathcal{A}_1, \ldots, \mathcal{A}_n$ be operators of the form $\mathcal{A}_i=b(f_i)$ or $\mathcal{A}_i=b^*(f_i)$, for $f_1, \dots, f_n \in \mathfrak{h}$. Here, we recall \eqref{Operator_b^*}--\eqref{Operator_b}. We then have
\begin{equation*}
\rho_{\tau,0} (\mathcal{A}_1 \cdots \mathcal{A}_n) = \sum_{\Pi \in M(n)} \prod_{(i,j) \in \Pi} \rho_{\tau,0}(\mathcal{A}_i \mathcal{A}_j)\,,
\end{equation*}
where, as in Proposition \ref{q_Wick_thm}, $M(n)$ denotes  the set of complete pairings of $\{1,\ldots,n\}$. The edges of $\Pi$ are now labelled using ordered pairs $(i,j)$ with $i < j$. Here, we also recall \eqref{rho_{tau,0}}. 
\end{lemma}
For a self-contained proof of Lemma \ref{Quantum_Wick_theorem}, we refer the reader to \cite[Lemma B.1]{FKSS17}.
As in \cite[Section 2]{FKSS17}, we use Lemma \ref{Quantum_Wick_theorem} to simplify the expression 
\eqref{q_gtau_defn}. Before proceeding, we define a few objects which we will use in the analysis.

\begin{definition}
\label{q_pairing_defn}
Given $\Pi$ a pairing of $\Sigma$, i.e. a one-regular graph on $\Sigma$, we regard its edges as ordered pairs $(\alpha,\beta)$ such that $\alpha<\beta$.
We then define $\mathfrak{P} \equiv \mathfrak{P}(m,p)$ to be the set of pairings $\Pi$ of $\Sigma$ satisfying $\delta_\alpha\delta_\beta = -1$ whenever $(\alpha,\beta) \in \Pi$; see Figure \ref{Fig:uncollapsed_paired}.
\end{definition}

\begin{figure}[htbp]
\centering
\includegraphics[scale=0.75]{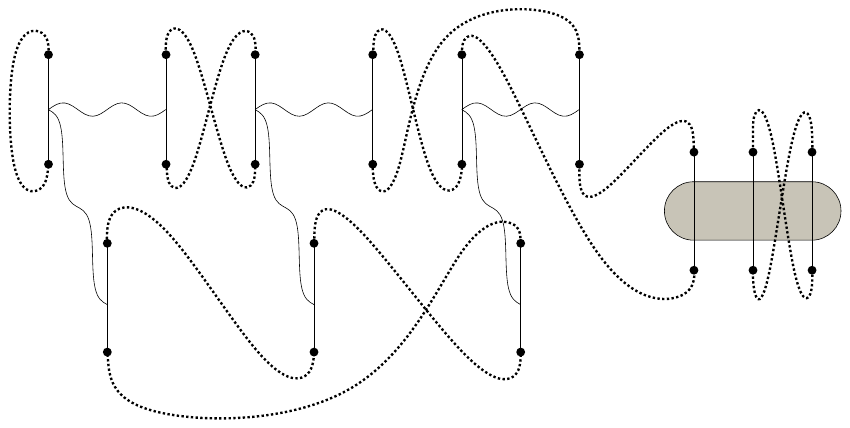}
\vspace*{-167mm}
\caption{A graph with $m=p=3$ with a valid pairing from Definition \ref{q_pairing_defn}.}
\label{Fig:uncollapsed_paired}
\end{figure}

\begin{definition}
For $\alpha \in \Sigma$ and $\mbt \in \mathfrak{V}$, define
\begin{align}
	\label{q_mcB_defn}
\mcB_{\alpha}(\mbx,\mbt) := 
\begin{cases}
\etavstar(x_\alpha) &\text{ if $\delta = 1$}\,, \\
\left(\e^{-t_\alpha h /\tau} \vpht \right) (x_\alpha) &\text{ if $\delta = -1$}\,.
\end{cases}
\end{align}
\end{definition}
\begin{definition}
For $\Pi \in \mfP$ and $\mbt \in \mathfrak{V}$, we define
\begin{multline}
\label{I_{tau,Pi}_definition}
I_{\tau,\Pi}^\xi(\mbt):= \\
\int_{\Lambda^{\Sigma}} \dd\mbx\,\Biggl( \prod^{m}_{i=1}  w(x_{i,1,1}-x_{i,2,1})\,w(x_{i,1,1}-x_{i,3,1}) \prod^{3}_{r=1} \delta(x_{i,r,1}-x_{i,r,-1})  \Biggr) \\
\times \xi(x_{m+1} , \ldots, x_{m+p};y_{m+1} , \ldots, y_{m+p})
\prod_{(\alpha, \beta) \in \Pi} \rho_{\tau,0}(\mcB_{\alpha}(\mbx,\mbt)\mcB_{\beta}(\mbx,\mbt))\,.
\end{multline} 
\end{definition}
We can now state the simplification of \eqref{q_gtau_defn} that we will use in the sequel.
\begin{lemma}
	\label{q_gtau_rewrite_lem}
	For $\mbt \in \mathfrak{V}$, we have $\g = \sum_{\Pi \in \mfP} I_{\tau,\Pi}(\mbt)$.
\end{lemma}
\begin{proof}
We follow the proof and use the notation used in \cite[Lemma 2.8]{FKSS17}. Namely for $\alpha \in \Sigma$, to each $x_\alpha$, we define a corresponding spectral label $k_\alpha \equiv k_{i,r,\delta}$, and write $\mbk := (k_\alpha)_{\alpha \in \Sigma}$. Let us recall \eqref{q_eigenvalues_defn}. For $l \in \N$, we let
\begin{equation*}
	u_{l}^{\alpha} := 
\begin{cases}
\overline{e}_{l} \text{ if $\delta = +1$}\,, \\
{e}_{l} \text{ if $\delta = -1$}\,.
\end{cases}
\end{equation*}
From \eqref{q_mcB_defn}, it follows that
\begin{equation}
	\label{q_mcB_rewrite}
\mcB_{\alpha}(\mbx, \mbt) = \sum_{k_{\alpha } \in \N} \e^{\delta_\alpha t_{\alpha } \lambda_{k_\alpha}/\tau} u_{k_{\alpha}}^{\alpha}(x_\alpha) \mcA_{\alpha}(\mbk)\,,
\end{equation}
where we define
\begin{equation*}
	\label{q_mcA_defn}
\mcA_{\alpha}(\mbk) := 
\begin{cases}
\vph_{\tau,k_\alpha}^* \text{ if $\delta = +1$}\,, \\
\vph_{\tau,k_\alpha} \text{ if $\delta = -1$}\,.
\end{cases}
\end{equation*}
Using \eqref{q_mcB_rewrite} in \eqref{q_gtau_defn}, we have
\begin{multline}
		\label{q_gtau_rewrite}
\g = \\
\int_{\Lambda^{\Sigma}} \dd\mbx \, \Biggl(\prod^{m}_{i=1}  w(x_{i,1,1}-x_{i,2,1})\,w(x_{i,1,1}-x_{i,3,1}) \prod^{3}_{r=1} \delta(x_{i,r,1}-x_{i,r,-1})  \Biggr) \\
\times \xi(x_{m+1,1,1} , \ldots, x_{m+1,p,1};y_{m+1,1,-1} , \ldots, y_{m+1,p,-1}) \\
\times \sum_{\mbk} \left[ \left( \prod_{\alpha \in \Sigma} \e^{\delta_\alpha t_{\alpha} \lambda_{\tau,k_\alpha}/\tau} u_{k_{\alpha}}^{\alpha}(x_\alpha)\right) \, \rho_{\tau,0}\left(\prod_{\alpha \in \Sigma}\mcA_{\alpha}(\mbk)\right) \right]\,.
\end{multline}
The result follows from applying Lemma \ref{Quantum_Wick_theorem} to \eqref{q_gtau_rewrite} and using \eqref{q_mcB_rewrite}.
\end{proof}
We define the following bounded operators.
\begin{align*}
	S_{\tau,t} &:= \e^{-th/\tau} \text{ if $t \geq 0$}\,, \\
	G_{\tau,t} &:= \frac{\e^{-th/\tau}}{\tau(\e^{h/\tau} - 1)} \text{ if $t \geq -1$}\,.
\end{align*}
Here we note that $G_{\tau,t}$ is the time-evolved Green function. 
\begin{lemma}
	\label{q_GS_positive_lem}
	For $t \geq 0$, both $G_{\tau,t}$ and $S_{\tau,t}$ have symmetric, non-negative kernels. 
\end{lemma}
\begin{proof}
	Since both $S_{\tau,t}$ and $G_{\tau,t}$ are self-adjoint, it follows from the Proposition \ref{q_FK_formula} that their kernels are non-negative and thus symmetric. See \cite[Lemma 2.9]{FKSS17} for more details.
\end{proof}
We thus have the following result for computing the free quantum states of products of pairs of $\mcB_\alpha$; see \cite[Lemma 2.10]{FKSS17}.
\begin{lemma}
	\label{q_graph_computation_lem}
Suppose $\alpha, \beta \in \Sigma$ with $\alpha < \beta$.
\begin{enumerate}
\item \label{q_graph_computation_1} If $\delta_{\alpha} = 1$ and $\delta_{\beta} = - 1$ with $t_\alpha-t_{\beta} < 1$, then
\begin{equation*}
\rho_{\tau,0}\left(\mcB_{\alpha}(\mbx,\mbt)\mcB_{\beta}(\mbx,\mbt)\right) = G_{\tau,-(t_{\alpha}-t_{\beta})}(x_{\alpha} ; x_\beta)\,.
\end{equation*}
\item \label{q_graph_computation_2} If $\delta_{\alpha} = -1$ and $\delta_{\beta} = 1$ with $t_\alpha-t_{\beta} \geq 0$, then
\begin{equation*}
\rho_{\tau,0}\left(\mcB_{\alpha}(\mbx,\mbt)\mcB_{\beta}(\mbx,\mbt)\right) = G_{\tau,t_{\alpha}-t_{\beta}}(x_{\alpha} ; x_\beta) + \frac{1}{\tau}S_{\tau,t_{\alpha}-t_{\beta}}(x_{\alpha} ; x_{\beta})\,.
\end{equation*}
\item \label{q_graph_computation_3} For both \eqref{q_graph_computation_1} and \eqref{q_graph_computation_2}, by Lemma \ref{q_GS_positive_lem}, we have $\rho_{\tau,0}\left(\mcB_{\alpha}(\mbx,\mbt)\mcB_{\beta}(\mbx,\mbt)\right) \geq 0$.
\end{enumerate}
\end{lemma}
We define a second, coloured graph, where we collapse the left-hand side vertices of the original graph. In other words, we identify $x_{i,r,1}$ and $x_{i,r,-1}$ for $i \leq m$. In order to make this precise, we use first define the following equivalence relation.
\begin{definition}
\label{Definition_3.28}
For $\alpha,\beta \in \Sigma$, we write $\alpha \sim \beta$ if and only if $i_{\alpha} = i_{\beta} \leq m$ and $r_{\alpha} =r_{\beta}$.
\end{definition}
Let us now implement Definition \ref{Definition_3.28} in the graphical structure.
\begin{definition}
\label{q_colored_graph_defn}
For each $\Pi \in \mfP$, we define the coloured graph $(V_{\Pi},E_{\Pi},\sigma_{\Pi}) = (V,E,\sigma)$ as follows. $V := \Sigma/\sim$ is the set of equivalence classes of $\Sigma$ under $\sim$. For $\alpha \in \Sigma$, let us denote by $[\alpha]$ its equivalence class under $\sim$. For each edge $(\alpha, \beta) \in \Pi$, we obtain an edge $e = \{[\alpha],[\beta]\}$, and we denote by $E$ the set of edges obtained in this way. Finally, we define the colour of an edge $e \in E$ as 
\begin{equation}
\label{sigma_e_definition}
\sigma(e) := \delta_{\beta}\,.
\end{equation}
This is well-defined by construction.
\end{definition}

\begin{figure}[htbp]
    \centering
    \includegraphics[scale=0.75]{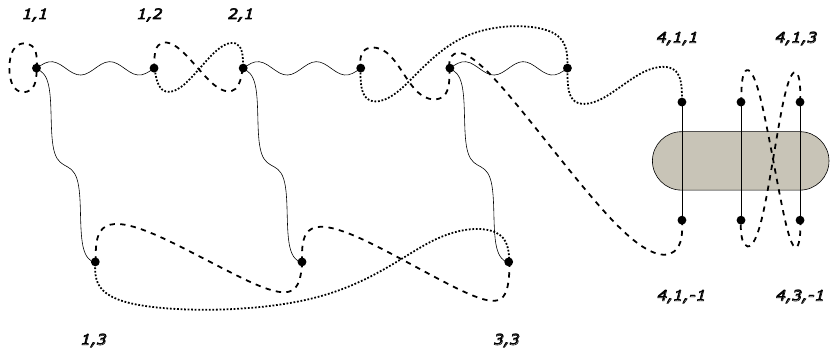}
    \vspace*{-167mm}
    \caption{An example of the coloured graph from Definition \ref{q_colored_graph_defn} with the same pairing as Figure \ref{Fig:uncollapsed_paired}. The wider dotted edges have colour $-1$, and the finer dotted lines have colour $1$.}
    \label{fig:collapsed_pairing}
\end{figure}

\begin{remark}
We make the following observations about $(V,E,\sigma)$.
\begin{enumerate}
\item The set $V$ inherits a well-defined total order from $\Sigma$ defined by $[\alpha] \leq [\beta]$ if $\alpha \leq \beta$. We also adopt the same convention as before to write $[\alpha] < [\beta]$ if $\alpha < \beta$.
\item We can write $V$ as the  disjoint union $V = V_2 \sqcup V_1$, where
\begin{equation*}
V_2 :=\{(i,r) \mid i \in \{1,\ldots m\}, r \in \{1,2,3\}\}, \quad V_1 := \{(m+1,r, \pm 1)\}.
\end{equation*}
Note that then each vertex in $V_j$ has degree $j$. %Moreover, no path that contains a vertex in $V_1$ can be a loop.
%Then each vertex of $V_j$ has degree $j$, and the only loops are in $V_2$.
\item We write $\mathrm{conn}(E)$ to denote the set of connected components of $E$. Thus $E = \bigsqcup_{\mcP \in \mathrm{conn}(E)} \mcP$. We call $\mcP \in \mathrm{conn}(E)$ a {\it path} of $E$.
\end{enumerate}
\end{remark}
We now fix $m,p \in \N$ and $\Pi \in \mfP$, and let $(V,E,\sigma)$ be the associated graph defined in Definition \ref{q_colored_graph_defn}. For each $\mbx \in \Lambda^\Sigma$ and $\mbt \in \R^{\Sigma}$, we associate integration labels $\mby := (y_a)_{a \in V} \in \Lambda^V$ and $\mbs := (s_a)_{a \in V} \in \R^V$ defined by
\begin{equation}
	\label{q_mbys_defn}
y_\alpha := x_{[\alpha]}, \qquad s_{\alpha}:= t_{[\alpha]}
\end{equation}
for any $\alpha \in \Sigma$. It follows from Definition \ref{q_vertex_set_defn} that the definition above does not depend on the choice of vertex $\alpha$. We note that \eqref{q_time_variable_identity} implies
\begin{equation}
	\label{q_s_variable_identity_1}
a<b \quad \Rightarrow \quad 0 \leq s_a - s_b < 1\,,
\end{equation}
and
\begin{equation}
	\label{q_s_variable_identity_2}
s_a = s_b \textrm{ if and only if } i_a = i_b\,,
\end{equation}
where we have used a slight abuse of notation to write $i_{\alpha} := i_{[\alpha]}$.

\begin{definition}
\label{Open_and_closed_paths}
We say that a path $\mcP \in \mathrm{conn}(E)$ is {\it closed} if all of its vertices are in $V_2$. Otherwise we call it {\it open}. We also denote by $V(\mcP) := \bigcup_{e \in \mcP} e$ and $V_i(\mcP) := V(\mcP) \cap V_i$.
\end{definition}
\begin{definition}
	\label{q_Jte_defn}
For $\mby \in \Lambda^{V}$ and $\mbs \in \R^V$ satisfying \eqref{q_s_variable_identity_1}, and $e = \{a,b\} \in E$, we define $\mby_e := (y_a;y_b) \in \Lambda^e$ and the integral kernels
\begin{align}
	\label{q_mcJ_defn}
\mathcal{J}_{\tau,e}(\mby_e,\mbs) &:= G_{\tau,\sigma(e)(s_a-s_b)}(\mby_e) +\frac{\chi(\sigma(e) = 1)\chi(s_a \neq s_b)}{\tau}S_{\tau,s_a-s_b}(\mby_e)\,, \\
	\label{q_mcJhat_defn}
\hat{\mathcal{J}}_{\tau,e}(\mby_e,\mbs) &:= G_{\tau,\sigma(e)(s_a-s_b)}(\mby_e)\,.
\end{align}
Here, we recall \eqref{sigma_e_definition}.
\end{definition}
Although both of the operators \eqref{q_mcJ_defn}--\eqref{q_mcJhat_defn}  always Hilbert-Schmidt, we never estimate them in the Hilbert-Schmidt norm as we want to prove estimates which are uniform in $\mbs$.
In the sequel, let us write 
\begin{equation}
\label{y_variable_splitting}
\mby = (\mby_1,\mby_2)\,,
\end{equation}
with $\mby_i := (\mby_a)_{a \in V_i}$. Slightly abusing notation, let us also write
\begin{equation*}
\xi(y_{m+1,1,1},\ldots,y_{m+1,p,1};y_{m+1,1,-1},\ldots,y_{m+1,p,-1}) =: \xi(\mby_1)\,.
\end{equation*}
\begin{lemma}
\label{q_graph_rewrite_1_lem}
For $\mbt \in \mathfrak{V}$ and $\mbs$ defined as in \eqref{q_mbys_defn}, we have
\begin{equation}
	\label{q_I_rewrite1}
I^\xi_{\tau,\Pi}(\mbt) = \int_{\Lambda^V} \dd\mby \Biggl(\prod_{i=1}^m w(y_{i,1}-y_{i,2})\, w(y_{i,1}-y_{i,3})  \Biggr) \xi(\mby_1) \prod_{e \in E} \Jte(\mby_e,\mbs)\,.
\end{equation}
\end{lemma}

\begin{proof}
We define a map $T: \Lambda^{V} \to \Lambda^{\Sigma}$ by $T\mby := (y_{[\alpha]})_{\alpha  \in \Sigma}$. Then $T$ is a bijection from $\Lambda^V$ onto the subset of $\Lambda^{\Sigma}$ defined by
\begin{equation*}
\{\mbx \in \Lambda^{\Sigma} \mid x_{i,r,1} = x_{i,r,-1} \text{ for all } i \in \{1,\ldots,m\}, r \in \{1,2,3\}\}\,.
\end{equation*}
From Lemma \ref{q_graph_computation_lem} and Definition \ref{q_Jte_defn}, it follows that
\begin{equation}
	\label{q_quantum_free_computn}
\rho_{\tau,0}\left(\mcB_{\alpha}(T\mby,\mbt)\mcB_{\alpha}(T\mby,\mbt)\right) = \Jte(\mby_e,\mbs)\,,
\end{equation}
where $(\alpha, \beta) \in \Pi$ is chosen such that $\left[\alpha\right] = a$ and $\left[\beta\right] = b$. By making the change of variables, we obtain $\mbx = T\mby$, \eqref{q_I_rewrite1}, as was claimed.
\end{proof}
Let us note the following corollary.
\begin{corollary} For $\mbt \in \mathfrak{V}$, we have
	\label{q_I_bound_1_corollary}
\begin{equation}
	\label{q_I_bound_1}
\left|I^\xi_{\tau,\Pi}(\mbt)\right| \leq  \|w\|_{L^\infty}^{2m} \int_{\Lambda^V} \dd\mby\, |\xi(\mby_1)| \prod_{e \in E} \Jte(\mby_e,\mbs)\,.
\end{equation}
\end{corollary}
\begin{proof}
This is a consequence of Lemmas \ref{q_graph_rewrite_1_lem} and \ref{q_graph_computation_lem} \eqref{q_graph_computation_3}.
\end{proof}

\begin{lemma}
	\label{q_closed_path_lem}
Suppose that $\mcP \in \mathrm{conn}(E)$ is a closed path as in Definition \ref{Open_and_closed_paths}. Then, we have
\begin{equation}
	\label{q_closed_path_bound}
\int_{\Lambda^{V(\mcP)}} \prod_{a \in V(\mcP)} \dd y_a \prod_{e \in \mcP} \Jte(\mby_e,\mbs) \leq C^{|V(\mcP)|}\,.
\end{equation}
Moreover
\begin{equation}
	\label{q_closed_path_convergence}
\left|\int_{\Lambda^{V(\mcP)}} \prod_{a \in V(\mcP)} \dd y_a \left[\prod_{e \in \mcP} \Jte(\mby_e,\mbs) - \prod_{e \in \mcP} \hat{\mathcal{J}}_{\tau,e}(\mby_e,\mbs)\right] \right| \to 0 \text{ as $\tau \to \infty\,$.}
\end{equation}
\end{lemma}

We note that $h$ defined as in \eqref{q_1-body_Hamn_defn} satisfies $h > 0$ and $h^{-1} \in \mfS^1(\mf{h}) \subset \mfS^2(\mf{h})$.

%{\color{red}
%Before proceeding with proof of Lemma \ref{q_closed_path_lem}, we first note some estimates on the quantum Green function, which are proved in \cite[Appendix C]{FKSS17}. 
%\begin{proposition}
%\label{q_estimates_quantum_Green}
%For $h$ defined as in \eqref{q_1-body_Hamn_defn}, we have
%\begin{equation}
%	\label{q_quantum_Green_estimate1}
%\lim_{\tau \to \infty} \left\|\frac{1}{\tau(\e^{h/\tau} - 1)} - h^{-1}\right\|_{\mf{S}^2(\mfh)} = 0,
%\end{equation}
%and
%\begin{equation}
%	\label{q_quantum_Green_estimate2}
%\lim_{\tau \to \infty} \frac{1}{\tau}\, {\rm Tr} \left(\frac{1}{\tau(\e^{h/\tau} - 1)}\right) = 0,
%\end{equation}
%and 
%\begin{equation}
%	\label{q_quantum_Green_estimate3}
%\lim_{\tau \to \infty} (1+t) \left\| \frac{\e^{-th/\tau}}{\tau(\e^{h/\tau} - 1)} - h^{-1} \right\|_{\mf{S}^2(\mfh)} \rightarrow 0,
%\end{equation}
%uniformly in $t \in (-1,1)$.
%\end{proposition}
%}

\begin{proof}[Proof of Lemma \ref{q_closed_path_lem}]
Since $\mcP$ is closed, we need to consider two cases depending on the length of $\mcP$. If $|\mcP|=1$, the path is a loop. Then the left-hand side of \eqref{q_closed_path_bound} equals 
\begin{multline}
	\label{q_norm_quantum_Green_fn}
\int_\Lambda \dd y_a \, G_\tau(y_a;y_a) = \|G_\tau\|_{\mfS^1(\mfh)} = \mathrm{Tr}(G_\tau) = \sum_{k \geq 0} \frac{1}{\tau(e^{\lambda_k/\tau} - 1)} \\
\leq \sum_{k \geq 0} \frac{1}{\lambda_k}  = \mathrm{Tr}(h^{-1})= \|G\|_{\mfS^1(\mfh)} < \infty\,. 
\end{multline}
Here we have used the positivity of $G_\tau$ as an operator as well as \eqref{lambda_k}.
	
We henceforth need to consider the case when $|\mcP| \geq 2$. Here, we argue as in \cite[Lemma 2.17]{FKSS17}. Let $|\mcP| = q$, and write $\mcP = \{e_1,\ldots,e_q\}$, where $e_j$ and $e_{j+1}$ are incident for $j=1,\ldots,q$. Throughout the proof we take the index $j$ to be modulo $q$. We denote by $a_j$ the unique vertex in $e_{j-1} \cap e_j$. Without loss of generality, we take $a_1 < a_2$.
	
An induction argument shows that the colour of $e_j$ is determined by the colour of $e_1$. Namely, recalling \eqref{sigma_e_definition}, we have
\begin{equation}
	\label{q_closed_colouring}
\sigma(e_j) = 
\begin{cases}
\sigma(e_1) \quad &\textrm{ if } a_j<a_{j+1}, \\
-\sigma(e_1) \quad &\textrm{\ if } a_j>a_{j+1}.
\end{cases}
\end{equation}
For $j=1,\ldots,q$, we define $a_{j,-} := \min\{a_j,a_{j+1}\}$ and $a_{j,+} := \max\{a_j,a_{j+1}\}$. From \eqref{q_closed_colouring}, we have
\begin{equation}
	\label{q_time_label_rewrite}
\sigma(e_j)(s_{a_{j,-}} - s_{a_{j,+}}) = \sigma(e_1)(s_{a_j}-s_{a_{j+1}})\,.
\end{equation}
Moreover, it is clear from \eqref{q_s_variable_identity_1} and \eqref{q_s_variable_identity_2} that $0 \leq \sigma(e_1)(s_{a_j}-s_{a_{j+1}}) < 1$. Thus, $G_{\tau, \sigma(e_1)(s_{a_j}-s_{a_{j+1}})}$ and $S_{\tau, \sigma(e_1)(s_{a_j}-s_{a_{j+1}})}$ are both well-defined. Substituting \eqref{q_time_label_rewrite} into \eqref{q_mcJ_defn}, we have
\begin{multline}
	\label{q_mcJ_rewrite}
\mcJ_{\tau,e_j}(\mby_{e_j};\mbs) = G_{\tau, \sigma(e_1)(s_{a_j}-s_{a_{j+1}})}(y_{a_j};y_{a_{j+1}}) \\ 
+ \frac{\chi(\sigma(e_j) = 1)\chi(s_{a_j} \neq s_{a_{j+1}})}{\tau}S_{\tau, \sigma(e_1)(s_{a_j}-s_{a_{j+1}})}(y_{a_j};y_{a_{j+1}})\,.
\end{multline}
Here we use that the kernels of $G_{\tau,t}$ and $S_{\tau,t}$ are symmetric. Rewriting the left-hand side of \eqref{q_closed_path_bound} using \eqref{q_mcJ_rewrite}, we get
\begin{equation}
	\label{q_closed_path_bound_rewrite}
\textrm{Tr} \left[\prod_{j=1}^q \left(G_{\tau, \sigma(e_1)(s_{a_j}-s_{a_{j+1}})}
+ \frac{\chi(\sigma(e_j) = 1)\chi(s_{a_j} \neq s_{a_{j+1}})}{\tau}S_{\tau, \sigma(e_1)(s_{a_j}-s_{a_{j+1}})}\right)\right].
\end{equation}
By definition, all of these operators commute, and hence the order of the above product does not matter. We define
\begin{equation*}
J_\mcP := \bigl\{j \in \{1,\ldots,q\} : \mcJ_{\tau,e_j} \neq \hat{\mathcal{J}}_{\tau,e_j}\bigr\}\,.
\end{equation*}
We note that $J_\mcP \bigl\{1,\ldots, q\}$. Namely, we recall \eqref{q_mcJ_defn}--\eqref{q_mcJhat_defn}
%Definition \ref{q_Jte_defn} 
and note that by construction, the smallest vertex in $\mcP$ (with respect to $\leq$) is incident to an edge $e$ with $\sigma(e)=-1$. Therefore, we can rewrite \eqref{q_closed_path_bound_rewrite} as
\begin{multline}
	\label{q_closed_path_bound_rewrite2}
\sum_{I \subset J_\mcP} \textrm{Tr} \left[ \left( \prod_{j \in \{1,\ldots,q\} \backslash I} G_{\tau, \sigma(e_1)(s_{a_j}-s_{a_{j+1}})} \right)
\left( \prod_{j \in I} \frac{1}{\tau}S_{\tau, \sigma(e_1)(s_{a_j}-s_{a_{j+1}})} \right) \right] \\
= \sum_{I \subset J_\mcP} \textrm{Tr} \left[\left(\prod_{j \in \{1,\ldots,q\} \backslash I} G_{\tau,0} \right) \left(\prod_{j \in \{1,\ldots,q\} \backslash I} \frac{1}{\tau} S_{\tau,0} \right) \right] = \sum_{J \subset J_\mcP} \frac{1}{\tau^{|I|}} \textrm{Tr} (G_\tau^{q-|I|})\,.
\end{multline}
The first equality holds since $\sum_{j=1}^q \left( s_{a_j} - s_{a_{j+1}} \right) = s_{a_1} - s_{a_{q+1}} = 0$ because $s_{a_1} = s_{a_{q+1}}$. By Lemma \ref{q_schatten_embedding}, we note that for $|I| \leq q-2$,
\begin{equation}
	\label{q_quantumgreen_trace_estimate}
\textrm{Tr}(G_{\tau}^{q-|I|}) = \|G_\tau\|_{\mfS^{q-|I|}(\mfh)}^{q-|I|} \leq \|G_\tau\|^{q-|I|}_{\mfS^2(\mfh)}\,.
\end{equation}
Using \eqref{q_quantumgreen_trace_estimate} we have that \eqref{q_closed_path_bound_rewrite2} is
\begin{align}
\nonumber
&\leq \sum_{\substack{I \subset J_\mcP \\ |I| \leq q-2}} \frac{1}{\tau^{|I|}} \|G_\tau\|^{q-|I|}_{\mfS^2(\mfh)} + \frac{q}{\tau^{q-1}} \|G_\tau\|_{\mfS^1(\mfh)} \\
	\label{q_closed_path_bound_rewrite3}
&\leq C^{|V(\mcP)|} \left(1 + \|G_\tau\|_{\mfS^2(\mfh)} + \|G_\tau\|_{\mfS^1(\mfh)}\right)^{|V(\mcP)|}\,.
\end{align}
In the first inequality above, we used $q \geq 2$. We now deduce \eqref{q_closed_path_bound} from \eqref{q_closed_path_bound_rewrite3}, by using 
\begin{equation}
\label{G_tau_bound}
\|G_{\tau}\|_{\mfS^2(\mfh)}  \leq \|G_{\tau}\|_{\mfS^1(\mfh)} \leq C \,,
\end{equation}
for some $C>0$ independent of $\tau$, which follows from \eqref{q_norm_quantum_Green_fn}, and Lemma \ref{q_schatten_embedding}.
	
To obtain \eqref{q_closed_path_convergence}, we split into the same two cases. If $|\mcP| = 1$, then $\mcP = \{e\}$, so the path is a loop, so $s_a=s_a$. So $\mcJ_e = \hat{\mcJ}_e$, and there is nothing to prove. If $|\mcP| \geq 2$, we apply the same argument as used in the proof of \eqref{q_closed_path_bound}. The only difference is that we now sum over {\it non-empty} subsets $I$ of $J_\mcP$ in \eqref{q_closed_path_bound_rewrite2}. This results in the an extra power of $\frac{1}{\tau}$, and one less power of $\left(1 + \|G_\tau\|_{\mfS^2(\mfh)} + \|G_\tau\|_{\mfS^1(\mfh)}\right)$ in \eqref{q_closed_path_bound_rewrite3}. We hence deduce \eqref{q_closed_path_convergence}.
\end{proof}
\begin{lemma}
	\label{q_open_path_lem}
Suppose that $\mcP \in \mathrm{conn}(E)$ is an open path with endpoints $b_1, b_2 \in V_2(\mcP)$. Then, we have
\begin{equation}
	\label{q_open_path_bound}
\left\|\int_{\Lambda^{V_2(\mcP)}} \prod_{a \in V_2(\mcP)} \dd y_a \prod_{e \in \mcP} \Jte(\mby_e,\mbs)\right\|_{L_{y_{b_1},y_{b_2}}^2} \leq C^{|V_2(\mcP)|}\,.
\end{equation}
Moreover
\begin{equation}
	\label{q_open_path_convergence}
\left\|\int_{\Lambda^{V(\mcP)}} \prod_{a \in V(\mcP)} \dd y_a \left[\prod_{e \in \mcP} \Jte(\mby_e;\mbs) - \prod_{e \in \mcP} \hat{\mathcal{J}}_{\tau,e}(\mby_e;\mbs)\right] \right\|_{L^2_{y_{b_1},y_{b_2}}} \to 0 \text{ as $\tau \to \infty$\,.}
\end{equation}
\end{lemma}

\begin{proof}
We argue similarly as for \cite[Lemma 2.18]{FKSS17}. We first prove \eqref{q_open_path_bound}. Let $b_1,b_2$ be as in the statement of the lemma. Without loss of generality, suppose $b_1 < b_2$, so $\delta_{b_1} = 1$ and $\delta_{b_2} = -1$. Let $q:= |V_2(\mathcal{P})|$. If $q=0$, then $\mathcal{J}_{\tau,e}(\mby;\mbs) = G_{\tau}(y_{b_1},y_{b_2})$, since $\sigma(e) = -1$. Hence, \eqref{q_open_path_bound} follows from \eqref{G_tau_bound}.
	
Suppose that $q \geq 1$. Write $\mcP := \{e_1,\ldots,e_{q+1}\}$, where ${b_1} \in e_1$, $b_2 \in e_{q+1}$ and $a_j$ is the unique vertex in $e_j \cap e_{j+1}$. An induction argument shows that the colour of $e_j$ is determined by the colour of $e_1$. Namely, we have
\begin{equation}
	\label{q_open_colour}
\sigma(e_j) = 
\begin{cases}
-1 \quad &\textrm{if } a_{j-1} < a_j, \\
1 \quad &\textrm{if } a_{j} < a_{j-1}\,.
\end{cases}
\end{equation}
Define $a_{j,-} := \min\{a_{j-1},a_{j}\}$ and $a_{j,+} := \max\{a_{j-1},a_{j}\}$. Then \eqref{q_open_colour} implies
\begin{equation}
\label{q_open_time}
\sigma(e_j)(s_{a_{j,-}} - s_{a_{j,+}}) = s_{a_j}-s_{a_{j-1}}\,.
\end{equation}
As in the proof of Lemma \ref{q_closed_path_lem}, we use \eqref{q_s_variable_identity_1} and \eqref{q_s_variable_identity_2} to deduce that $0 \leq \sigma(e_j)(s_{a_{j,-}} - s_{a_{j,+}}) < 1$. Substituting \eqref{q_open_colour} into \eqref{q_Jte_defn}, we have
\begin{multline}
	\label{q_open_Jte_rewrite}
\mathcal{J}_{\tau,e}(\mby_{e_j},\mbs) = G_{\tau,s_{a_{j}} - s_{a_{j-1}}}(y_{a_{j-1}};y_{a_j}) \\ 
+ \frac{\chi(\sigma(e_j)=1) \chi(s_{a_j} \neq s_{a_{j-1}})}{\tau} S_{\tau,s_{a_{j}} - s_{a_{j-1}}}(y_{a_{j-1}};y_{a_j}).
\end{multline}
Here, we have used the symmetry of the kernels of $G_{\tau,t}$ and $S_{\tau,t}$. Using \eqref{q_open_Jte_rewrite}, we have
\begin{multline}
\label{q_open_rewrite}
\int_{\Lambda^{V_2(\mcP)}} \prod_{a \in V_{2}(\mcP)} \dd y_a \prod_{e \in \mcP} \mcJ_{\tau,e}(\mby,\mbs) \\
= \left[\prod^{q+1}_{j=1}\left(G_{\tau,s_{a_j} - s_{a_{j-1}}}\right) + \frac{\chi(\sigma(e_j)=1) \chi(s_{a_j} \neq s_{a_{j-1}})}{\tau} S_{\tau,s_{a_{j}} - s_{a_{j-1}}}\right](y_{b_1};y_{b_2})\,.
\end{multline}
Define $J_{\mcP} := \{j \in \{1,\ldots,q+1\} : \mcJ_{\tau,e_j} \neq \hat{\mcJ}_{\tau,e_j}\}$. Since $\sigma(e_1) = 1$ and $s_{b_1} \neq s_{a_1}$ (because $q \geq$ 1), we have that $1 \in J_{\mcP}$. Moreover, $\sigma(e_{q+1}) = -1$, so $q+1 \notin J_\mcP$, so $1 \leq |J_{\mcP}| \leq q$. We rewrite \eqref{q_open_rewrite} as
\begin{multline}
	\label{q_open_rewrite2}
\sum_{I \subset J_{\mcP}} \left[\left( \prod_{j \in \{1,\ldots,q\} \backslash I} G_{\tau, s_{a_j}-s_{a_{j-1}}} \right) \left( \prod_{j \in I} \frac{1}{\tau} S_{\tau, s_{a_j}-s_{a_{j-1}}} \right) \right](y_{b_1};y_{b_2}) \\
= \sum_{I \subset J_{\mcP}} \left[\left( \prod_{j \in \{1,\ldots,q\} \backslash I} G_{\tau, 0} \right) \left( \prod_{j \in I} \frac{1}{\tau} S_{\tau, 0} \right) \right](y_{b_1};y_{b_2}) \\
= \sum_{I \subset J_{\mcP}} \frac{1}{\tau^{|I|}}\left(G_{\tau}^{q+1-|I|}\right)(y_{b_1};y_{b_2})\,.
\end{multline}
In the first inequality, we uses $\sum_{j=1}^{q+1}(s_{a_{j}} - s_{a_{j-1}}) = s_{a_{q+1}} - s_{a_0} = 0$, which is true since $s_{a_{q+1}} = s_{a_0}=0$.
	
For $k \geq 2$, applying the Cauchy-Schwarz inequality to the operator kernels implies that
\begin{multline}
	\label{q_open_path_CS}
G^k_\tau(y_{b_1};y_{b_2}) 
\\
:= \int_{\Lambda^{k-1}} \prod_{j=1}^{k-1} \dd x_j \, G_{\tau}(y_{b_1};x_1) G_{\tau}(x_1;x_2) \cdots G_{\tau}(x_{k-2};x_{k-1}) G_{\tau}(x_{k-1};y_{b_2}) \\
\leq \|G_{\tau}\|_{\mfS^2(\mfh)}^{k-2} \|G_{\tau}(y_{b_1},\cdot)\|_{\mfh} \|G_{\tau}(y_{b_2},\cdot)\|_{\mfh}\,.
\end{multline}
Applying \eqref{q_open_path_CS}, we deduce that \eqref{q_open_rewrite2} is
\begin{multline}
	\label{q_open_rewrite3}
\leq \sum_{\substack{I \subset J_{\mcP} \\ |I| \leq q-1}} \frac{1}{\tau^{|I|}}\|G_\tau\|_{\mfS^2(\mfh)}^{q+1-|I|} \|G_{\tau}(y_{b_1},\cdot)\|_{\mfh} \|G_{\tau}(y_{b_2},\cdot)\|_{\mfh} + \frac{q+1}{\tau^q} G_{\tau}(y_{b_1};y_{b_2}) \\
\leq C^{|V_2(\mcP)|}\left(1+\|G_\tau\|_{\mfS^2(\mfh)}\right)^{|V_2(\mcP)|} \left( \|G_{\tau}(y_{b_1},\cdot)\|_{\mfh} \|G_{\tau}(y_{b_2},\cdot)\|_{\mfh} + G_{\tau}(y_{b_1};y_{b_2})\right).
\end{multline}
Then \eqref{q_open_path_bound} follows from \eqref{q_open_rewrite3} and \eqref{G_tau_bound}.
	
To prove \eqref{q_open_path_convergence}, we note that if $q=0$, we have $\mcJ_{\tau,e} = \hat{\mathcal{J}}_{\tau,e}$, since $\delta_{b_2} = -1$. In this case, \eqref{q_open_path_convergence} automatically holds. If $q \geq 1$, we argue as for \eqref{q_open_path_bound}, except we sum over {\it non-empty} subsets of $J_\mcP$ in \eqref{q_open_rewrite2}. This results in an extra factor of $\frac{1}{\tau}$ and one less power of $(1+\|G_\tau\|)_{\mfS^2(\mfh)}$ in \eqref{q_open_rewrite3}. \eqref{q_open_path_convergence} then follows from \eqref{G_tau_bound}.
\end{proof}
We can now bound the quantity \eqref{I_{tau,Pi}_definition}.
\begin{lemma}
	\label{q_I_bound2_lem}
For $\Pi \in \mfP$ and $\mbt \in \mathfrak{V}$, we have
\begin{equation}
	\label{q_I_bound2_bound}
\left|I^\xi_{\tau,\Pi}(\mbt)\right| \leq C^{m+p}\|w\|_{L^\infty}^{2m}\,.
\end{equation}
\end{lemma}
\begin{proof}
We argue as in \cite[Lemma 2.19]{FKSS17}. We use the splitting \eqref{y_variable_splitting} and Corollary \ref{q_I_bound_1_corollary} to rewrite \eqref{q_I_bound_1} as
\begin{equation}
	\label{q_I_openclosed_rewrite1}
\left|I^\xi_{\tau,\Pi}(\mbt)\right| \leq \|w\|_{L^\infty}^{2m} \int_{\Lambda^{V_1}} \dd\mby_1 \, |\xi(\mby_1)| \int_{\Lambda^{V_2}} \dd\mby_2 \prod_{e \in E} \mcJ_{\tau,e}(\mby_e,\mbs)\,.
\end{equation}
Let us introduce the partition $\textrm{Conn}(E) = \textrm{Conn}_c(E) \sqcup \textrm{Conn}_o(E)$. In other words, we partition $\textrm{Conn}(E)$ into the closed connected paths $\textrm{Conn}_c(E)$ and the open connected paths $\textrm{Conn}_o(E)$. Then, we have
\begin{multline}
	\label{q_openclosed_decomposition}
\int_{\Lambda^{V_2}} \dd\mby_2 \prod_{e \in E} \mcJ_{\tau,e}(\mby_e,\mbs)
= \prod_{\mcP \in \textrm{Conn}_c(E)} \left(\int_{\Lambda^{V(\mcP)}} \prod_{a \in V(\mcP)} \dd y_a \prod_{e \in \mcP}\mcJ_{\tau,e}(\mby_e,\mbs)\right) \\
\times \prod_{\mcP \in \textrm{Conn}_o(E)} \left(\int_{\Lambda^{V_2(\mcP)}} \prod_{a \in V_2(\mcP)} \dd y_a \prod_{e \in \mcP}\mcJ_{\tau,e}(\mby_e,\mbs)\right).
\end{multline}
Substituting \eqref{q_openclosed_decomposition} into \eqref{q_I_openclosed_rewrite1} and using \eqref{q_closed_path_bound}, we have
\begin{multline}
\label{3.73}
|I^\xi_{\tau,\Pi}(\mbt)| \\
\leq C^m \|w\|_{L^\infty}^{2m} \int_{\Lambda^{V_1}} \dd \mby_1 \, |\xi(\mby_1)| \prod_{\mcP \in \textrm{Conn}_o(E)}\left(\int_{\Lambda^{V_2(\mcP)}} \prod_{a \in V_2(\mcP)} \dd y_a \prod_{e \in \mcP}\mcJ_{\tau,e}(\mby_e,\mbs)\right).
\end{multline}
We note that \eqref{q_I_bound2_bound} follows from \eqref{3.73} by applying the Cauchy-Schwarz inequality in the $\mby_1$ variables, followed by \eqref{q_open_path_bound}. We also use that $|\textrm{Conn}_o(E)| = p$ to get the factor of $C^p$ on the right-hand side of \eqref{q_I_bound2_bound}.
\end{proof}
Using Lemma \ref{q_I_bound2_lem}, we can now bound the quantity $g^{\xi}_{\tau,m}(\mbt)$ defined in \eqref{q_gtau_defn}.
\begin{lemma}
\label{Lemma_3.39}
For $\mbt \in \mathfrak{V}$, we have 
\begin{equation*}
\left|\g\right| \leq C^{m+p}\|w\|_{L^\infty}^{2m}(3m+p)!\,.
\end{equation*}
\end{lemma}

\begin{proof}
The claim follows from Lemma \ref{q_gtau_rewrite_lem}, \eqref{q_I_bound2_bound}, and the observation that $|\mfP| \leq (3m+p)!$.
\end{proof}
%By integrating over the $m$-dimensional simplex, have the following corollary, which proves Proposition \ref{q_expansion_nugeq0} (ii) in the quantum setting when $\xi \in \mathfrak{B}_p$.
We can now bound the quantity $b_{\tau,m}^{\xi}$ given by \eqref{q_bxi_defn_remk} and \eqref{alpha_1} above and obtain
\eqref{Proposition_3.15_1_1} uniformly in $\xi \in \mathfrak{B}_p$.
%Proposition \ref{q_expansion_nugeq0} (i).
\begin{corollary}
	\label{q_bxi_tau_bound_cor}
Uniformly in $\xi \in \mathfrak{B}_p$, we have
\begin{equation}
	\label{q_bxi_tau_bound_cor_bound}
|b^\xi_{\tau,m}| \leq (Cp)^p C^m (m!)^2 \|w\|_{L^\infty}^{2m} =: C(m,p)\,.
\end{equation}
\end{corollary}
\begin{proof}
The claim follows from \eqref{q_explicit_b_rewrite_f1} and Lemma \ref{Lemma_3.39}, after integrating over the simplex \eqref{V_simplex} (which gives a factor of $\frac{1}{m!}$) and using Stirling's formula.
\end{proof}
\subsubsection{Proof of \eqref{Proposition_3.15_1_2} uniformly in $\xi \in \mathfrak{B}_p$}
\label{Subsection_3.6.2}
Let us make the following definition.
\begin{definition}
\label{J_{e}_definition}
For $e =\{a,b\} \in E$, we define
\begin{equation*}
\mathcal{J}_e(\mby_e) := G(y_a;y_b)\,.
\end{equation*}
\end{definition}
Proposition \ref{q_FK_formula} implies the following lemma.
\begin{lemma}
\label{q_G_positive_lemma}
The kernel of $G$ is non-negative and symmetric.
\end{lemma}
\begin{definition}
For $\Pi \in \mfP$, we define
\begin{equation}
	\label{q_I_classical_defn}
I^\xi_{\Pi} := \int_{\Lambda^{V}} \dd\mby\,\Biggl( \prod_{i=1}^m w(y_{i,1}-y_{i,2})\,w(y_{i,1}-y_{i,3}) \Biggr) \xi(\mby_1) \prod_{e \in E} \mathcal{J}_e(\mby_e)\,.
\end{equation}
\end{definition}
\begin{lemma}
	\label{q_I_convergence_lemma}
For each $\Pi \in \mfP$ and $\mbt \in \mathfrak{V}$, we have
\begin{equation}
	\label{q_I_convergence}
I^\xi_{\tau,m}(\mbt) \to I_{m}^{\xi} \text{ uniformly in $\xi \in \mathfrak{B}_p$ as $\tau \to \infty$\,.}
\end{equation}
\end{lemma}

\begin{proof}
The proof is similar to the proof of \cite[Lemma 2.25]{FKSS17}. For $\mbt \in \mathfrak{V}$, we define
\begin{multline}
	\label{q_I_quantum_hat_defn}
\hat{I}_{\tau,\Pi}^\xi(\mbt) := 
\\
\int_{\Lambda^{V}} \dd \mby \,\Biggl(\prod_{i=1}^{n}w(y_{i,1}-y_{i,2}) \, w(y_{i,1}-y_{i,3}) \Biggr) \xi(\mby_1) \prod_{e \in E} \hat{\mathcal{J}}_{\tau,e}(\mby_e)\,.
\end{multline}
Let first show that 
\begin{equation}
	\label{q_I_convergence_subfact}
\hat{I}_{\tau,\Pi}^{\xi}(\mbt) \to I_{\Pi}^{\xi} \textrm{ uniformly in $\xi \in \mathfrak{B}_p$ as $\tau \to \infty$\,.}
\end{equation}
Namely, from \eqref{q_I_classical_defn} and \eqref{q_I_quantum_hat_defn}, we have
	\begin{multline}
	\label{q_I_difference}
\hat{I}_{\tau,\Pi}^{\xi}(\mbt) - I_{\Pi}^{\xi} = \int_{\Lambda^{V}} \dd \mby \Biggl( \prod_{i=1}^{m}  w(y_{i,1}-y_{i,2}) \,w(y_{i,1}-y_{i,3}) \Biggr) \xi(\mby_1) \\
\times \left[\prod_{e \in E} \hat{\mathcal{J}}_{\tau,e}(\mby_e) - \prod_{e \in E} \mathcal{J}_{e}(\mby_e)\right].
\end{multline}
By telescoping, we can write
\begin{multline}
	\label{q_I_telescope}
\prod_{\substack{e \in E}} \hat{\mcJ}_{\tau,e}(\mby_e,\mbs) - \prod_{e \in E} \mcJ_e(\mby_e) = \\
\sum_{e_0 \in E} \left[\prod_{\substack{e \in E \\ e<e_0}} \hat{\mcJ}_{\tau,e}(\mby_e,\mbs) \left(\hat{\mcJ}_{\tau,e_0}(\mby_{e_0},\mbs) - \mcJ_{e_0}(\mby_{e_0})\right) \prod_{\substack{e \in E \\ e>e_0}} \mcJ_e(\mby_e) \right],
\end{multline}
where we order the elements of $E$ arbitrarily. Substituting \eqref{q_I_telescope} into \eqref{q_I_difference}, we have
\begin{multline}
	\label{q_I_difference2}
\left| \hat{I}_{\tau,\Pi}^{\xi}(\mbt) - I_{\Pi}^{\xi} \right| \leq \sum_{e_0 \in E} \|w\|_{L^\infty}^{2m} \int_{\Lambda^{V}} \dd \mby \, |\xi(\mby_1)| \\
\times \left[\prod_{\substack{e \in E \\ e<e_0}} \hat{\mcJ}_{\tau,e}(\mby_e,\mbs) \left|\hat{\mcJ}_{\tau,e_0}(\mby_{e_0},\mbs) - \mcJ_{e_0}(\mby_{e_0})\right| \prod_{\substack{e \in E \\ e>e_0}} \mcJ_e(\mby_e)\right].
\end{multline}
Here, we have used that $\mcJ_{\tau,e}(\mby_e;\mbs)$ and $\mcJ_{e}(\mby_e) \geq 0$ by Lemmas \ref{q_GS_positive_lem} and \ref{q_G_positive_lemma} (recalling Definitions \ref{q_Jte_defn} and \ref{J_{e}_definition}).

Let us denote by $\sigma_{\tau,e_0}^\xi(\mbt)$ the summand in \eqref{q_I_difference2} corresponding to $e_0$. Since \eqref{q_I_difference2} is a finite sum, in order to obtain \eqref{q_I_convergence_lemma},  it suffices to show that for each $e_0 \in E$, we have
\begin{equation}
	\label{q_sigma_convergence}
\sigma^\xi_{\tau,e_0}(\mbt) \to 0 \textrm{ uniformly in $\xi \in \mathfrak{B}_p$ as $\tau \to \infty\,$.}
\end{equation}
We fix $e_0 \in E$. Let us define an integral kernel associated with an edge $e \in E$ by
\begin{equation}
	\label{q_Jtilde_defn}
\tilde{\mcJ}_{\tau,e}(\mby_e,\mbs) :=
\begin{cases}
\hat{\mathcal{J}}_{\tau,e}(\mby_e,\mbs) \quad &\textrm{if } e < e_0\,, \\
\left| \hat{\mathcal{J}}_{\tau,e}(\mby_e,\mbs) - \mcJ_e(\mby_e) \right| \quad &\textrm{if } e = e_0\,, \\
\mcJ_e(\mby_e) \quad &\textrm{if } e > e_0\,.
\end{cases}
\end{equation}
We have the following estimates for $\tilde{\mcJ}_{\tau,e}$.
\begin{enumerate}
\item If $e=\{a,a\}$ (i.e. a loop in the graph) and $e \neq e_0$, then
\begin{equation}
	\label{q_Jtilde_estimate_loop_not_e0}
\|\tilde{\mcJ}_{\tau,e}(\cdot,\mbs)\|_{\mfS^1(\mfh)} \leq \|G_\tau\|_{\mfS^1(\mfh)} + \|G\|_{\mfS^1(\mfh)} \leq C\,,
\end{equation}
which holds by \eqref{q_norm_quantum_Green_fn}. %since we are in one dimension -- see \cite[(4.7)]{FKSS17}.
\item If $e=\{a,a\}$ and $e = e_0$, then
\begin{multline}
	\label{q_Jtilde_estimate_loop_e0}
\|\tilde{\mcJ}_{\tau,e}(\cdot,\mbs)\|_{\mfS^1(\mfh)} = 
\int \dd y \, \left|\hat{\mcJ}_{\tau,e_0}(\mby_{e_0},\mbs) - \mcJ_{e_0}(\mby_{e_0}) \right|  
\\
= \int \dd y \,  \left| \frac{1}{\tau(\e^{h/\tau}-1)}(y;y) - \frac{1}{h}(y;y) \right| 
= \int \dd y \, \sum_{k \in \N} \left(\frac{1}{\lambda_k} - \frac{1}{\tau(e^{\lambda_k/\tau} - 1)} \right)
\\
= \|G_\tau - G\|_{\mfS^1(\mfh)} \to 0\,,
\end{multline}
as $\tau \to \infty$ by spectral decomposition and the dominated convergence theorem. The third equality above follows by comparing Taylor series.
\item If $e=\{a,b\}$ with $a < b$ and $e < e_0$, then $\tilde{\mcJ}_{\tau,e}(\mby_e,\mbs) = G_{\tau,\sigma(e)(s_a-s_b)}(y_a;y_b)$. 

Let us note that 
\begin{equation}
	\label{q_quantum_Green_estimate3}
\lim_{\tau \to \infty} (1+t) \left\| \frac{\e^{-th/\tau}}{\tau(\e^{h/\tau} - 1)} - h^{-1} \right\|_{\mf{S}^2(\mfh)} \rightarrow 0\,,
\end{equation}
uniformly in $t \in (-1,1)$. The claim \eqref{q_quantum_Green_estimate3} follows by a spectral argument; see \cite[Lemma C.2]{FKSS17} for the proof of a more general claim.
Then \eqref{q_quantum_Green_estimate3} implies
\begin{equation}
	\label{q_Jtilde_estimate_nonloop_e<e0}
\|\tilde{\mcJ}_{\tau,e}(\cdot,\mbs)\|_{\mfS^2(\mfh)} \leq C_{\mbs}\,,
\end{equation}
where the constant depends on $\mbs$.
\item If $e=\{a,b\}$ with $a < b$ and $e = e_0$. Then 
\begin{equation}
\label{q_quantum_Green_estimate3_1}
\tilde{\mcJ}_{\tau,e}(\mby_e,\mbs) = \\ \left(G_{\tau,\sigma(e)(s_a-s_b)} - G\right)(y_a;y_b)\,. 
\end{equation}
Therefore, \eqref{q_quantum_Green_estimate3_1}--\eqref{q_quantum_Green_estimate3} imply that
\begin{equation}
	\label{q_Jtilde_estimate_nonloop_e=e0}
\|\tilde{\mcJ}_{\tau,e}(\cdot,\mbs)\|_{\mfS^2(\mfh)} \to 0 \textrm{ as $\tau \to \infty\,$.}
\end{equation}

\item If $e=\{a,b\}$ with $a < b$ and $e > e_0$. Then $\tilde{\mcJ}_{\tau,e}(\mby_e,\mbs) = G(y_a;y_b)$. Then
\begin{equation}
	\label{q_Jtilde_estimate_nonloop_e>e0}
\|\tilde{\mcJ}_{\tau,e}(\cdot,\mbs)\|_{\mfS^2(\mfh)} = \|G\|_{\mfS^2(\mfh)} \leq C\,.
\end{equation}
\end{enumerate}
Applying the same decomposition from the proof of Lemma \ref{q_I_bound2_lem}, we have
\begin{multline}
	\label{q_sigma_rewrite}
\sigma^{\xi}_{\tau,e_0}(\mbt) = \|w\|_{L^\infty}^{2m} \int_{\Lambda^{V_1}} \dd\mby_1 \prod_{\mcP \in \textrm{Conn}_c(E)}\left(\int_{\Lambda^{V(\mcP)}} \prod_{a \in V(\mcP)} \dd y_a \prod_{e \in \mcP} \tilde{\mcJ}_{\tau,e}(\mby_e,\mbs)\right) \\
\times \prod_{\mcP \in \textrm{Conn}_o(E)}\left(\int_{\Lambda^{V_2(\mcP)}} \prod_{a \in V_2(\mcP)} \dd y_a \prod_{e \in \mcP} \tilde{\mcJ}_{\tau,e}(\mby_e,\mbs)\right)\,.
\end{multline}
For $\mcP \in \textrm{Conn}_c(E)$, apply the Cauchy-Schwarz inequality in $y_a$ for $a \in V(\mcP)$ to deduce that
\begin{equation}
	\label{q_Jtilde_closed_estimate}
\int_{\Lambda^{V(\mcP)}} \prod_{a \in V(\mcP)} \dd y_a \prod_{e \in \mcP} \tilde{\mcJ}_{\tau,e}(\mby_e,\mbs) \leq \prod_{e \in \mcP} \|\tilde{\mcJ}_{\tau,e}(\cdot,\mbs)\|_{\mfS^2(\mfh)}\,.
\end{equation}
For $\mcP \in \textrm{Conn}_o(E)$, we recall the notation from the proof of Lemma \ref{q_open_path_lem}. Namely
\begin{equation*}
\mcP = \{e_1,\ldots,e_{q+1}\}, \quad V_1(\mcP) = \{b_1,b_2\}, \quad \quad V_2(\mcP) = \{a_1,\ldots,a_q\}.
\end{equation*}
We apply the Cauchy-Schwarz inequality in the $y_a$ for $a \in V_2(\mcP)$ to obtain
\begin{multline}
	\label{q_Jtilde_open_estimate}
\int_{\Lambda^{V_2(\mcP)}} \prod_{a \in V_2(\mcP)} \dd y_a \prod_{e \in \mcP} \tilde{\mcJ}_{\tau,e}(\mby_e,\mbs) \leq \\
\prod_{j=2}^q \|\tilde{\mcJ}_{\tau,e_j}(\cdot,\mbs)\|_{\mfS^2(\mfh)} \|\tilde{\mcJ}_{\tau,e_1}(y_{b_1},\cdot)\|_{\mfh} \|\tilde{\mcJ}_{\tau,e_{q+1}}(y_{b_2},\cdot)\|_{\mfh}\,.
\end{multline}
Substituting \eqref{q_Jtilde_closed_estimate} and \eqref{q_Jtilde_open_estimate} into \eqref{q_sigma_rewrite} and applying the Cauchy-Schwarz inequality in the $\mby_1$ variables yields
\begin{equation}
	\label{q_sigma_rewrite2}
\sigma^\xi_{\tau,e_0}(\mbt) \leq \|w\|_{L^\infty}^{2m} \|\xi\|_{\mfS^2(\mfh)} \prod_{e \in E} \|\tilde{\mcJ}_{\tau,e}(\cdot,\mbs)\|_{\mfS^2(\mfh)}\,.
\end{equation}
Recalling \eqref{q_Jtilde_estimate_loop_not_e0}--\eqref{q_Jtilde_estimate_nonloop_e>e0} and Lemma \ref{q_schatten_embedding}, we obtain \eqref{q_sigma_convergence} from \eqref{q_sigma_rewrite2}. Hence, \eqref{q_I_convergence_subfact} follows.

Let us also note that 
\begin{equation}
	\label{q_I_hat_convergence}
I_{\tau,\Pi}^\xi(\mbt) \to \hat{I}_{\tau,\Pi}^{\xi}(\mbt) \textrm{ uniformly in $\xi \in \mathfrak{B}_p$ as $\tau \to \infty$\,.}
\end{equation}
In order to obtain \eqref{q_I_hat_convergence}, we use a telescoping argument analogous to \eqref{q_I_telescope} above, perform a decomposition into open and closed paths as in the proof of Lemma \ref{q_I_bound2_lem}, and use Lemmas \ref{q_closed_path_lem} and \ref{q_open_path_lem}. We omit the details.
We obtain the claim of Lemma \ref{q_I_convergence_lemma} from \eqref{q_I_convergence_subfact} and \eqref{q_I_hat_convergence}.
\end{proof}
For $m \in \N$, let us define
\begin{equation}
	\label{q_Wick_bm}
b^\xi_{\infty,m} := \frac{(-1)^m}{m!\,3^m}\sum_{\Pi \in \mathfrak{P}} I^\xi_{\Pi}\,.
\end{equation}
We now conclude the claimed convergence result.
\begin{lemma}
	\label{q_b^xi_convergence_lemma}
For $b^\xi_{\tau,m}, b^{\xi}_m$ defined as in \eqref{q_bxi_defn_remk}, we have
\begin{equation}
	\label{q_b^xi_convergence}
b^\xi_{\tau,m} \to b^\xi_{m} \textrm{ as $\tau \to \infty$ uniformly in $\xi \in \mathfrak{B}_p\,$}.
\end{equation}
\end{lemma}

\begin{proof}
By Lemma \ref{q_gtau_rewrite_lem} and Proposition \ref{q_I_bound2_lem}, %and Proposition \ref{q_estimates_quantum_Green}, 
we have that $g^\xi_{\tau,m}(\mbt)$ is bounded uniformly in $\tau > 0$, $\xi \in \mathfrak{B}_p$, and $\mbt \in \mathfrak{V}$. Moreover, Lemmas \ref{q_gtau_rewrite_lem} and \ref{q_I_convergence_lemma} imply that for any $\mbt \in \mathfrak{V}$
\begin{equation}
\label{q_b^xi_convergence_1}
g^\xi_{\tau,m}(\mbt) \to \sum_{\Pi \in \mfP} I^\xi_{\Pi} \textrm{ as $\tau \to \infty$ uniformly in $\xi \in \mathfrak{B}_p\,$}.
\end{equation}
Recalling \eqref{q_explicit_b_rewrite_f1} and applying \eqref{q_Wick_bm}, as well as \eqref{q_b^xi_convergence_1}, combined with the dominated convergence theorem, we obtain that $b^\xi_{\tau,m} \to b^\xi_{\infty,m}$ as $\tau \to \infty$ uniformly in $\xi \in \mathfrak{B}_p$. The claim of the lemma follows by noting that, by Wick's theorem (Proposition \ref{q_Wick_thm}), we have $b^\xi_{\infty,m} = b^\xi_m$.
\end{proof}

\subsubsection{Proof of \eqref{Proposition_3.15_1_1}--\eqref{Proposition_3.15_1_2} for $\xi=\mathbf{1}_p$}
\label{Subsection_3.6.3}
To conclude the proof of Proposition \ref{q_expansion_nugeq0}, it remains to prove \eqref{Proposition_3.15_1_1}--\eqref{Proposition_3.15_1_2} for $\xi=\mathbf{1}_p$.
We hence consider the operator with kernel
\begin{equation}
	\label{q_identity_kernel}
\xi(x_1,\ldots,x_p;y_1,\ldots,y_p) := \prod_{j=1}^p \delta(x_j-y_j)\,.
\end{equation}
To do this, we need to introduce a slightly modified version of the graphs defined in Definition \ref{q_vertex_set_defn} above. The modification is analogous to that used in \cite[Section 4.2]{FKSS17}.
\begin{definition}
	\label{q_delta_graph_defn}
Let $m,p \in \N$ be given. We consider the same abstract vertex set $\Sigma \equiv \Sigma(m,p)$ with $6m+2p$ elements and set of matchings of $\mfP \equiv \mfP(m,p)$ as in Definition \ref{q_vertex_set_defn}.
For each $\Pi \in \mfP$, we consider a coloured multigraph $(\tilde{V},\tilde{E},\tilde{\sigma}) \equiv (\tilde{V}_{\Pi},\tilde{E}_{\Pi},\tilde{\sigma}_{\Pi})$, with $\tilde{\sigma}: \tilde{E} \to \{\pm 1\}$, defined as follows.
\begin{enumerate}
\item We say $\alpha \sim \beta$ if and only if $i_{\alpha} = i_{\beta}$ and $r_{\alpha} = r_{\beta}$. We define the set $\tilde{V} := \{[\alpha]:\alpha \in \Sigma\}$ and write $\tilde{V} = \tilde{V}_1 \sqcup \tilde{V}_2$, where 
\begin{equation*}
\tilde{V}_1 := \{(m+1,r): r \in \{1,\ldots,p\}\}
\end{equation*}
and 
\begin{equation*}
\tilde{V}_2 := \{(i,r) : i \in \{1,\ldots,m\}, r \in \{1,2,3\}\}\,.
\end{equation*}
The set $\tilde{V}$ inherits an order from the lexicographical ordering (given in Definition \ref{q_vertex_set_defn}). Namely $[\alpha] \leq [\beta]$ if $\alpha \leq \beta$. 
\item For $\Pi \in \mfP$, we note that $(\alpha,\beta)$ induces an edge $e:=\{[\alpha],[\beta]\}$ in $\tilde{E}$. Let us $\tilde{\sigma}(e):= \delta_{\beta}$. This is well-defined by construction.
\item Let $\mathrm{conn}(\tilde{E})$ denote the set of connected components of $\tilde{E}$. Note we can write $\tilde{E} = \bigsqcup_{\mcP \in \mathrm{conn}(\tilde{E})} \mcP$. We call the connected components $\mcP$ of $\tilde{E}$ paths. 
\end{enumerate}
\end{definition}

\begin{remark}
In a slight abuse of notation, we denote the equivalence relation in Definition \ref{q_colored_graph_defn} and in Definition \ref{q_delta_graph_defn} by $\sim$. From context, it will always be clear to which equivalence relation we are referring. The same holds for the order $\leq$ induced by lexicographical order on $\Sigma$.
\end{remark}
\begin{remark}
The difference between the graph structure in Definition \ref{q_delta_graph_defn} and the one in Definition \ref{q_colored_graph_defn} is that, in the former, we identify the nodes corresponding to the observable $\xi$ with kernel \eqref{q_identity_kernel}; see Figure \ref{Fig:delta_graph}.
\end{remark}

\begin{figure}[htbp]
\centering
\includegraphics[scale=0.75]{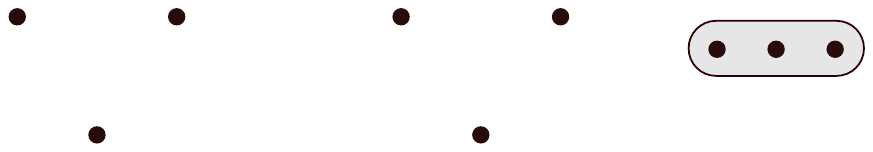}
\vspace*{-175mm}
\caption{An unpaired graph with $m=2$, $p=3$ corresponding to Definition 
 \ref{q_delta_graph_defn}.}
\label{Fig:delta_graph}
\end{figure}

By construction, every vertex in $\tilde{V}$ has degree 2. Hence all paths $\mcP \in \mathrm{conn}(\tilde{E})$ are closed.
For fixed $\Pi \in \mathfrak{P}$ and $\mbt \in \mathfrak{V}$, we define
\begin{equation}
I^\xi_{\tau,\Pi}(\mbt) 
:= \int_{\Lambda^{\tilde{V}}} \prod_{a \in \tilde{V}} \dd y_a \Biggl(\prod_{i=1}^m w(y_{i,1}-y_{i,2}) \,w(y_{i,1}-y_{i,3})\, \Biggr)\prod_{e \in \tilde{E}} \mcJ_{\tau,e}(\mby_e; \mbs)\,.
\end{equation}
As in Definition \ref{q_vertex_set_defn}, we consider the spatial labels $\mby = (y_a)_{a \in \tilde{V}}$ and time labels $\mbs = (s_a)_{a \in \tilde{V}}$, except that we now adopt $\tilde{V}$ as the vertex set and let the equivalence relation be the one from Definition \ref{q_delta_graph_defn} above. Moreover, we adopt the convention that $\mby_i := (y_a)_{a \in \tilde{V}_i}$. Given $\mcP \in \mathrm{conn}(\tilde{E})$, we denote the set of vertices of $\mcP$ by $\tilde{V}(\mcP)$, and write $\tilde{V}_i(\mcP) := \tilde{V}(\mcP) \cap \tilde{V}_i$ for $i =1,2$, analogously as in Definition \ref{Open_and_closed_paths}.
Let us also define $\mcJ_{\tau,e}$ analogously to \eqref{q_mcJ_defn}, replacing $(V,E)$ with $(\tilde{V},\tilde{E})$.
\begin{lemma}
	\label{q_delta_path_bound_lem}
Suppose that $\mcP \in \mathrm{conn}(\tilde{E})$. Then
\begin{equation}
	\label{q_delta_path_bound}
\int_{\Lambda^{\tilde{V}(\mcP)}} \prod_{a \in \tilde{V}(\mcP)} \dd y_a \prod_{e \in \mcP} \mcJ_{\tau,e}(\mby_e,\mbs) \leq C^{|\tilde{V}(\mcP)|}\,.
\end{equation}
\end{lemma}
\begin{proof}
We follow the approach of the proof given in \cite[Lemma 4.9]{FKSS17}. We have two cases, depending on whether
$\tilde{V}(\mcP) \subset \tilde{V}_2$ or $\tilde{V}_1(\mcP) \neq \emptyset$. In the first case, all the vertices lies in $\tilde{V}_2$, so we can argue as in the proof of Lemma \ref{q_closed_path_lem}.

Let us henceforth assume that $\tilde{V}_1(\mcP) \neq \emptyset$. %$|\tilde{V}(\mcP)| = 1$ or $|\tilde{V}(\mcP)| > 1$. 
If $|\tilde{V}(\mcP)| = 1$, then $\mcP$ is a loop, and the left-hand side of \eqref{q_delta_path_bound} is $\|G_{\tau}\|_{\mfS^1}$. Therefore, we can argue as in the proof of Lemma \ref{q_closed_path_lem} to get the required bound.
	
Let us now suppose that $|\tilde{V}(\mcP)| > 1$. Since $\mcP$ is a closed path, there exist $b_1\ldots,b_k \in \tilde{V}_1(\mcP)$ such that $\mcP = \sqcup_{j=1}^k \mcP_j$, where for each $j=1,\ldots,k$, $\mcP_j = \{e_1^j,\ldots,e^j_{q_j}\}$, with $b_j \in e^j_1$, $b_{j+1} \in e^j_{q_j}$, and $e^j_{k} \cap e^j_{k+1} \in \tilde{V}_2$ for $k\in \{1,\ldots,q_{j-1}\}$. Since $\mcP$ is closed, we set $b_{k+1} := b_1$. Let us note that if $b_j$ and $b_{j+1}$ are connected by a path of length one, then we have $q_j = 1$.
Therefore, the left-hand side of \eqref{q_delta_path_bound} can be written as
\begin{equation}
	\label{q_delta_path_rewritten}
\int_{\Lambda^{k}} \dd y_{b_1} \ldots \dd y_{b_1} \prod_{j=1}^k \left[\int_{\Lambda^{\tilde{V}_2(\mcP_j)}} \prod_{a \in \tilde{V}_2(\mcP_j)} \dd y_a \prod_{e \in \mcP_j} \mcJ_{\tau,e}(\mby_e,\mbs)\right]\,.
	\end{equation}
Arguing as in the proof of Lemma \ref{q_open_path_lem} (in particular as in the proof of \eqref{q_open_rewrite3}), we have that the $j^{th}$ factor in \eqref{q_delta_path_rewritten} is less than or equal to
\begin{multline}
	\label{q_delta_jth_factor}
C^{|\tilde{V}_2(\mcP_j)|}\left(1+\|G_\tau\|_{\mfS^2(\mfh)}\right)^{|\tilde{V}_2(\mcP_j)|} \left( \|G_{\tau}(y_{b_j},\cdot)\|_{\mfh} \|G_{\tau}(y_{b_{j+1}},\cdot)\|_{\mfh} + G_{\tau}(y_{b_j};y_{b_{j+1}})\right)\,.
\end{multline}
\eqref{q_delta_path_bound} follows from applying the Cauchy-Schwarz inequality in each of the $y_{b_1},\ldots,y_{b_k}$ variables, and using \eqref{G_tau_bound}. %and using \cite[Lemma C.1]{FKSS17}.
\end{proof}
Let us note that Lemma \ref{q_delta_path_bound_lem} implies \eqref{Proposition_3.15_1_1} for $\xi = \mathbf{1}_p$. 

For fixed $\Pi \in \mfP$, we define
\begin{equation}
\label{I^xi_Pi_delta}
I^\xi_\Pi := \int_{\Lambda^{\tilde{V}}} \dd \mby \Biggl(\prod_{i=1}^m  w(y_{i,1}-y_{i,2}) \,w(y_{i,1}-y_{i,3})\Biggr) \prod_{e \in \tilde{E}} \mcJ_e(\mby_e)\,.
\end{equation}
Let $b^\xi_{\infty,m}$ be defined as in \eqref{q_Wick_bm}, where $I^\xi_\Pi $ is now given by \eqref{I^xi_Pi_delta} instead of by \eqref{q_I_classical_defn}. The same telescoping argument used in the proof of Lemma \ref{q_I_convergence_lemma} (adapted to the framework of the proof of Lemma \ref{q_delta_path_bound_lem}) implies that \eqref{Proposition_3.15_1_2} holds for $\xi = \mathbf{1}_p$. We omit the details.
%combined with the proof of Lemma \ref{q_b^xi_convergence_lemma} implies that 
%\eqref{q_b^xi_convergence} also holds for  
This completes the proof of Proposition \ref{q_expansion_nugeq0}.

\section{The time-independent problem with unbounded interaction potentials. Proof of Theorems \ref{q_L1_state_convergence_thm} and \ref{delta_function_state_convergence_thm}}
\label{q_Unbounded_cases_section}

In this section, we consider $w$ as in Assumption \ref{w_assumption} above.  Our goal is to prove Theorems \ref{q_L1_state_convergence_thm} and \ref{delta_function_state_convergence_thm}, and thus complete the analysis of the time-independent problem outlined in Section \ref{Time-independent results}. 
Throughout this section, we use the notational conventions as given in the statements of Theorems \ref{q_L1_state_convergence_thm} and \ref{delta_function_state_convergence_thm} above.
%Note that a duality \textbf{REWRITE} argument (as in \cite[Section 3.3]{FKSS17}) and Theorem \ref{q_bounded_state_convergence_thm} above imply the following claim.
Let us first note the following claim, which follows from \eqref{duality_1}--\eqref{duality_2} and Theorem \ref{q_bounded_state_convergence_thm} by duality.

%We now complete the analysis for the unbounded time-independent problem, dealing with $L^1$ and focusing delta function interaction potentials. To do this, we need the quintic analogues of the results from \cite[Section 4]{RS22}.

\begin{lemma}
\label{q_bounded_state_convergence_lem}
Suppose that $w \in L^\infty$ is even and real-valued (as in Assumption \ref{w_assumption_bounded}). Then, for all $p \in \N^*$, we have $\rho_{\tau}(\Theta_{\tau}(\xi)) \to \rho(\Th(\xi))$ as $\tau \to \infty$ uniformly in $\xi \in \mcCp$. Here, we recall the definition  \eqref{q_mathcalC_defn} of $\mcCp$.
\end{lemma}

We first note an approximation result.

\begin{lemma}
\label{q_unbounded_time_ind_lemma}
The following claims hold.
\begin{itemize}
\item[(i)]
Let $w$ be as in Assumption \ref{w_assumption} (i) and let $w^\eps \in L^\infty$ be a sequence of interaction potentials as in Assumption \ref{w_assumption_bounded} satisfying $w^{\varepsilon} \rightarrow w$ in $L^1$ as $\varepsilon \to 0$. Then, there is a sequence $(\varepsilon_\tau)$ tending to $0$ as $\tau \rightarrow \infty$ such that for any $p \in \N^*$
\begin{equation}
\label{q_unbounded_time_ind_equation}
\lim_{\tau \rightarrow \infty} \rt^{\varepsilon_\tau}(\Th_\tau(\xi)) = \rho(\Th(\xi))
\end{equation}
uniformly in $\xi \in \mathcal{C}_p$ with $\mathcal{C}_p$ defined as in \eqref{q_mathcalC_defn}.
\item[(ii)] 
Let $w$ be as in Assumption \ref{w_assumption} (ii). Let $w^{\varepsilon}$ be defined as in \eqref{w_epsilon_local_problem_definition} above. Then, there is a sequence $(\varepsilon_\tau)$ tending to $0$ as $\tau \rightarrow \infty$ such that for any $p \in \N^*$, \eqref{q_unbounded_time_ind_equation} holds (with notational conventions as in Theorem \ref{delta_function_state_convergence_thm}).
\end{itemize}
\end{lemma}
\begin{proof} (i) Let us first consider $w$ as in Assumption \ref{w_assumption} (i) (and with the notational conventions as in Theorem \ref{q_L1_state_convergence_thm}). 
By a diagonal argument, it suffices to prove that for fixed $\varepsilon>0$, we have
\begin{equation}
\label{q_unbounded_state_conv_1}
\lim_{\tinfty} \rt^{\varepsilon}(\Th_\tau(\xi)) = \rho^\eps(\Th(\xi))
\end{equation}
uniformly in $\xi \in \mc{C}_p$, and that
\begin{equation}
\label{q_unbounded_state_conv_2}
\lim_{\varepsilon \to 0} \rho^{\varepsilon} (\Th(\xi)) = \rho(\Theta(\xi))\,,
\end{equation}
uniformly in $\xi \in \mc{C}_p$. The convergence \eqref{q_unbounded_state_conv_1} follows by noting that $w^{\varepsilon}$ satisfies Assumption \ref{w_assumption_bounded} for each fixed $\varepsilon>0$ and by using Lemma \ref{q_bounded_state_convergence_lem}. 

Let us now prove the convergence \eqref{q_unbounded_state_conv_2}. We recall \eqref{q_classical_interaction} (and the corresponding analogue for interaction $w^{\varepsilon}$). 
Let us first show that 
\begin{equation}
\label{(4.5)_claim}
\lim_{\varepsilon \rightarrow 0} \mcW^\eps = \mcW\,,\quad \text{$\mu$-almost surely\,.}
\end{equation}
We have
\begin{multline}
\label{q_mcWeps_to_mcW1}
3\left|\mcW - \mcW^\eps\right| \leq \int \dd x \, \dd y \, \dd z \, |\vph(x)|^2\,|\vph(y)|^2\,|\vph(z)|^2 \\
\times |w(x-y)\,w(x-z) - w^\eps(x-y)\,w^\eps(x-z)|\,.
\end{multline}
By adding and subtracting $w^\eps(x-y)\,w(x-z)$, and by applying the triangle inequality, we see that the term inside the absolute value on the right-hand side of \eqref{q_mcWeps_to_mcW1} is
\begin{equation}
\label{q_mcWeps_to_mcW2}\leq |w(x-y)-w^\eps(x-y)|\,|w(x-z)|
+ |w^\eps(x-y)|\,|w(x-z) -w^\eps(x-z)|\,.
\end{equation}
To bound \eqref{q_mcWeps_to_mcW1}, we apply Lemma \ref{Multlinear_estimates_1} separately to each of the terms obtained from \eqref{q_mcWeps_to_mcW2} and deduce that
\begin{equation}
\label{(4.5)_proof}
\left|\mcW - \mcW^\eps\right|  \lesssim \|w-w^\eps\|_{L^1} \,(\|w\|_{L^1}+\|w^{\eps}\|_{L^1}  \,)\|\vph\|_{L^6}^6\,,
\end{equation}
We hence obtain \eqref{(4.5)_claim} from \eqref{(4.5)_proof}, since $w^{\varepsilon} \rightarrow w$ in $L^1$ and $\varphi \in L^6$ $\mu$-almost surely.

By using Lemma \ref{Multlinear_estimates_1} and the fact that $w^{\varepsilon} \rightarrow w$ in $L^1$ as $\varepsilon \to 0$,
it follows that there exists $c_0>0$ (depending on $w$) such that for $\varepsilon>0$ sufficiently small, we have
\begin{equation}
\label{(4.6)_bound}
\bigl|\e^{-\mathcal{W}^{\varepsilon}}-\e^{-\mathcal{W}}\bigr| \leq 2 \e^{c_0 \|\varphi\|_{L^6}^6}\,.
\end{equation}
By Proposition \ref{Cauchy_problem_1} with local interaction (as in \cite{Bou94}) and Assumption \ref{f_assumption} (with $K$ sufficiently small), we know that
\begin{equation}
\label{(4.7)_bound}
\e^{c_0 \|\varphi\|_{L^6}^6}\, f^{\frac{1}{2}}(\mathcal{N}) \in L^1(\dd \mu)\,.
\end{equation}
By Lemma \ref{q_random_var_estimate} and Assumption \ref{f_assumption}, we have 
\begin{equation}
\label{(4.8)_bound}
\Theta(\xi)\,f^{\frac{1}{2}}(\mathcal{N}) \in L^{\infty}(\dd \mu)\,.
\end{equation}
Combining \eqref{(4.5)_claim} and \eqref{(4.6)_bound}--\eqref{(4.8)_bound}, it follows that 
\begin{equation}
\label{(4.9)_claim}
\lim_{\varepsilon \rightarrow 0} \int \dd \mu\, \bigl|\e^{-\mathcal{W}^{\varepsilon}}-\e^{-\mathcal{W}} \bigr|\,|\Theta(\xi)|\,f(\mathcal{N})=0\,.
\end{equation}
By analogous arguments, we obtain 
\begin{equation}
\label{(4.10)_claim}
\lim_{\varepsilon \rightarrow 0} z^{\varepsilon}=z\,.
\end{equation}
We write 
\begin{equation}
\label{(4.10)_claim_B}
\rho^{\varepsilon}(\Theta(\xi))-\rho(\Theta(\xi))=\frac{1}{z}\,\int \dd \mu \biggl(\frac{z}{z^{\varepsilon}}\,\e^{-\mathcal{W}^{\varepsilon}}-\e^{-\mathcal{W}}\biggr)\,\Theta(\xi)\,f(\mathcal{N})\,,
\end{equation}
and deduce the claim of the lemma from \eqref{(4.9)_claim}--\eqref{(4.10)_claim_B}. For the last step, we also used \eqref{(4.7)_bound} and  $\e^{-\mathcal{W}^{\varepsilon}} \leq  \e^{c_0 \|\varphi\|_{L^6}^6}$, which is proved as \eqref{(4.6)_bound} above. The claim for  $w$ as in Assumption \ref{w_assumption} (i) now follows

(ii) Let us now consider $w$ as in Assumption \ref{w_assumption} (ii) (and notational conventions as in Theorem \ref{delta_function_state_convergence_thm}). Once we show that \eqref{(4.5)_claim} holds, the proof follows as above. We start from \eqref{q_classical_interaction} (with $w$ replaced by $w^{\varepsilon}$) and use Parseval's identity to write

\begin{multline}
\label{W^epsilon_Fourier}
\mathcal{W}^{\varepsilon} =-\frac{1}{3}\, \sum_{k_1,\ldots,k_6} \widehat{w}^{\varepsilon}(k_1-k_2)\, \widehat{w}^{\varepsilon}(k_3-k_4)\,\widehat{\varphi}(k_1)\,\overline{\widehat{\varphi}(k_2)}\,\widehat{\varphi}(k_3)\,\overline{\widehat{\varphi}(k_4)}\,\widehat{\varphi}(k_5)\,\overline{\widehat{\varphi}(k_6)}
\\
\times
\delta(k_1-k_2+k_3-k_4+k_5-k_6)\,.
\end{multline}
Recalling Assumption \ref{w_assumption} (ii) and \eqref{w_epsilon_local_problem_definition}, it follows that
\begin{equation}
\label{w^epsilon_properties}
\lim_{\varepsilon \rightarrow 0} \widehat{w}^{\varepsilon}(k)=c=\widehat{w}(k)\,\,\, \text{for all } k \in \Z\,,\qquad  \sup_{\varepsilon>0} \|\widehat{w}^{\epsilon}\|_{\ell^{\infty}} \lesssim 1\,.
\end{equation}
Let us define $F:\T \rightarrow \C$ to be the inverse Fourier transform of $|\widehat{\varphi}|$. We then use Parseval's identity, and Sobolev embedding to deduce that
\begin{multline}
\label{J}
\mathcal{J} \equiv \mathcal{J}(\varphi):=\sum_{k_1,\ldots,k_6} |\widehat{\varphi}(k_1)|\,|\widehat{\varphi}(k_2)|\,|\widehat{\varphi}(k_3)|\,|\widehat{\varphi}(k_4)|\,|\widehat{\varphi}(k_5)|\,|\widehat{\varphi}(k_6)|
\\
\times
\delta(k_1-k_2+k_3-k_4+k_5-k_6)= \|F\|_{L^6}^6 \lesssim \|F\|_{H^{\frac{1}{3}}}^6=\|\varphi\|_{H^{\frac{1}{3}}}^6\,.
\end{multline}
For the last equality in \eqref{J}, we used $\|F\|_{H^{\frac{1}{3}}}=\|\varphi\|_{H^{\frac{1}{3}}}$ which follows from the fact that the $H^{\frac{1}{3}}$ norm is invariant under taking absolute values in the Fourier transform. Using \eqref{W^epsilon_Fourier}--\eqref{J} and the dominated convergence theorem, we deduce \eqref{(4.5)_claim} for $w$ as in Assumption \ref{w_assumption} (ii). The claim now follows.
\end{proof}

We now have all the tools necessary to prove Theorems \ref{q_L1_state_convergence_thm} and  \ref{delta_function_state_convergence_thm}.
\begin{proof}[Proof of Theorem \ref{q_L1_state_convergence_thm}]
We deduce \eqref{correlation_function_convergence_L1} (with the same subsequence $(\eps_\tau)$) from Lemma \ref{q_unbounded_time_ind_lemma} (i) by a duality argument based on \eqref{duality_1}--\eqref{duality_2} as in the proof of Theorem  \ref{q_bounded_state_convergence_thm}.
The proof of \eqref{partition_function_convergence_L1} is similar. Let us first note that by \eqref{partition_function_convergence_bounded}, we have that for fixed $\varepsilon>0$, 
\begin{equation}
\label{partition_function_convergence_L1_proof_1}
\lim_{\tau \rightarrow \infty} \mathcal{Z}^{\varepsilon}_{\tau}=z^{\varepsilon}\,.
\end{equation}
The proof of \eqref{partition_function_convergence_L1} proceeds analogously as the proof of \eqref{correlation_function_convergence_L1}. The only difference is that instead of \eqref{q_unbounded_state_conv_1} and \eqref{q_unbounded_state_conv_2}, we use \eqref{partition_function_convergence_L1_proof_1} and \eqref{(4.10)_claim} respectively.
\end{proof}

\begin{proof}[Proof of Theorem \ref{delta_function_state_convergence_thm}]
The proof follows from Lemma \ref{q_unbounded_time_ind_lemma} (ii) by arguing analogously as in the proof of Theorem \ref{q_L1_state_convergence_thm}. In particular, we obtain \eqref{correlation_function_convergence_delta} and \eqref{partition_function_convergence_delta} by arguing as for \eqref{correlation_function_convergence_L1} and \eqref{partition_function_convergence_L1}  respectively.
\end{proof}

\section{The time-dependent problem. Proofs of Theorems \ref{q_bounded_time_thm}--\ref{q_delta_time_thm}}
\label{q_Time_dependent_case}
\subsection{Bounded interaction potentials. Proof of Theorem \ref{q_bounded_time_thm}}
Let us consider $w$ as in Assumption \ref{w_assumption_bounded}.
To prove Theorem \ref{q_bounded_time_thm}, we note the following analogues of the Schwinger-Dyson expansion results from \cite[Sections 3.2 and 3.3]{FKSS18}.
\begin{lemma}[Quantum Schwinger-Dyson Expansion]
	\label{q_SD_quantum}
	Let $\xi \in \mcL(\mathfrak{h}^{(p)})$, $\mathcal{K} >0$, $\varepsilon > 0$, and $t \in \mathbb{R}$ be given. Then there exists $L(\mathcal{K},\varepsilon,t,\|\xi\|,p) \in \mathbb{N}$, $(e^l)_{l=0}^L$, where $e^l=e^l(x_i,t) \in \mcL(\mfh^{(l)})$, and $\tau_0 = \tau_0 (\mathcal{K},\varepsilon,t,\|\xi\|) > 0$ such that
	\begin{equation*}
		\left\| \left(\Psi_\tau^t\Th_\tau(\xi) - \sum^L_{l=0}\Th_\tau(e^l)\right)\bigg|_{\mfh^{(\leq \mathcal{K}\tau)}} \right\| < \varepsilon
	\end{equation*}
	for all $\tau > \tau_0$.
\end{lemma}
\begin{proof}
	We use the same proof and notation as \cite[Lemma 3.9]{FKSS18}. We begin by noting that all of the calculations up to \cite[(3.44)]{FKSS18} hold since they do not use the explicit form of $\mcW$. Instead of \cite[3.46]{FKSS17}, the same argument yields that the norm of $A^t_{\tau,\infty}(\xi)$ is bounded by
	\begin{equation*}
		\frac{|t|^j}{j!} (p+j)^j 2^j \left(\frac{n}{\tau}\right)^{p+j} \|w\|_{L^\infty}^{2j} \|\xi\| \leq \e^p \mathcal{K}^p\left( 2\e\mathcal{K}\|w\|_{L^\infty}^2|t|\right)^j \|\xi\|,
	\end{equation*}
as in \cite[(3.47)]{FKSS18}. An analogous argument to \cite[Lemma 3.9]{FKSS18} yields the result. We refer the reader to \cite[Lemma 3.9]{FKSS18} for the precise definitions and arguments.
\end{proof}
\begin{lemma}[Classical Schwinger-Dyson Expansion]
	\label{q_SD_classical}
Let $\xi \in \mcL(\mfh^{(p)})$, $\mathcal{K} > 0$, $\varepsilon > 0$, and $t \in \mathbb{R}$ be given. Then for $L(\mthc{K},\varepsilon, t,\|\xi\|,p) \in \mathbb{\N}$ and $\tau_0(\mthc{K},\varepsilon, t,\|\xi\|) > 0$ chosen possibly larger than in Lemma \ref{q_SD_quantum} and the same choice of $e^l \in \mcL(\mfh^{(l)})$ as in Lemma \ref{q_SD_quantum} we have
\begin{equation*}
\left| \left( \Psi^t \Th(\xi) - \sum^L_{l=0}\Th(e^l) \right) \chi_{\{\mthc{N} \leq \mthc{K}\}} \right| \leq \varepsilon
\end{equation*}
for all $\tau \geq \tau_0$.
\end{lemma}
\begin{proof}
The proof is analogous to the proof of Lemma \ref{q_SD_quantum}.
\end{proof}
\begin{proof}[Proof of Theorem \ref{q_bounded_time_thm}]
By the proof of \cite[Proposition 2.1]{FKSS18}, we note that the time-dependent claim follows from the corresponding time-independent claim, provided that we have the results of Lemmas \ref{q_SD_quantum} and \ref{q_SD_classical} above. The time-independent claim was shown in Theorem \ref{q_bounded_state_convergence_thm}. The claim now follows.
\end{proof}
\begin{remark}
The analogous reduction of the time-dependent result to the time-dependent result was also used in the proof of \cite[Theorem 1.10]{RS22}.
\end{remark}
\subsection{Unbounded interaction potentials. Proofs of Theorems \ref{q_L1_time_thm}--\ref{q_delta_time_thm}}

We now study the time-dependent problem with unbounded interaction potentials and prove Theorems \ref{q_L1_time_thm}--\ref{q_delta_time_thm}.

We first consider $w$ as in Assumption \ref{w_assumption} (i) and prove Theorem \ref{q_L1_time_thm}.
Throughout this subsection, we consider $w$ as in Assumption \ref{w_assumption}. Before proceeding with the proof of Theorem \ref{q_L1_time_thm}, we note an approximation result concerning the approximation of the flow of \eqref{q_Hartree_defn}.

Let $s \in (0,\frac{1}{2})$ be given. We denote by $\mathcal{G} \subset H^s(\Lambda)$ the set constructed in Proposition \ref{Cauchy_problem_2} (ii) above. We know $\mbP^f_{\mathrm{Gibbs}} (\mathcal{G})=1$ and initial data in $\mathcal{G}$ gives rise to global solutions of \eqref{q_Hartree_Cauchy_problem}. Here, we recall \eqref{q_gibbs_cutoff_defn_rigorous}.

\begin{lemma}
\label{q_Hartree_approximation_lemma}
Let $w^\eps \in L^\infty(\Lambda)$ be a sequence such that 
\begin{equation}
\label{w^{epsilon}_convergence_approximation}
\lim_{\varepsilon \rightarrow 0} \|w^{\varepsilon}-w\|_{L^1}=0\,.
\end{equation} Furthermore, fix $s \in (0,\frac{1}{2})$, $T>0$, and consider $\psi \in \mathcal{G}$, with $\mathcal{G} \subset H^s(\Lambda)$ as above.

\begin{enumerate}
\item[(i)]
For $\varepsilon>0$, the following Cauchy problem admits a solution on $[-T,T]$.
\begin{equation}
\label{q_Hartree_approximation_cauchy_problem_2}
\begin{cases}
\mathrm{i} \partial_t u^\eps+(\Delta - \kappa)u^\eps =
\\
-\frac{1}{3}  \int \dd y \, \dd z \, w^\eps(x-y)\,w^\eps(x-z)\, |u^\eps(y)|^2\,|u^\eps(z)|^2\,u^\eps(x) 
\\
-\frac{2}{3}  \int \dd y \, \dd z \, w^\eps(x-y)\,w^\eps(y-z)\, |u^\eps(y)|^2\,|u^\eps(z)|^2\,u^\eps(x) 
\\
u^\eps|_{t=0} = \psi\,.
\end{cases}
\end{equation}
\item[(ii)]
We have
\begin{equation}
\label{Lemma_5.3_convergence}
\lim_{\eps \to 0} \|u^{\varepsilon}- u\|_{L^\infty_{[-T,T]}H^s} = 0\,,
\end{equation}
where 
\begin{equation}
\label{q_Hartree_approximation_cauchy_problem_1}
\begin{cases}
\mathrm{i}\partial_t u +(\Delta-\kappa)u=  
\\
-\frac{1}{3} \int \dd y \, \dd z \, w(x-y)\,w(x-z)\,|u(y)|^2\,|u(z)|^2\,u(x) 
\\
-\frac{2}{3}\int \dd y \, \dd z \, w(x-y)\,w(y-z)\,|u(y)|^2\,|u(z)|^2\,u(x)
\\
u|_{t=0} = \psi\,.
\end{cases}
\end{equation}
\end{enumerate}

\end{lemma}

\begin{proof}
As in the proof of Proposition \ref{Cauchy_problem_1}, it suffices to consider the case $\kappa=0$. By symmetry, it suffices to consider only positive times.
Using \cite[Lemma 4.4]{Bou94} and \cite[(4.10)]{Bou94}, we have a more quantitative description of the set $\mathcal{G}$ in Proposition \ref{Cauchy_problem_2} (ii) above\footnote{Once we have the setup of Section \ref{Cauchy problem}, we can directly apply these arguments from \cite{Bou94}.}.
Namely, given $\nu>0$, we can write 
\begin{equation}
\label{G_union}
\mathcal{G} = \bigcup_{\vartheta,A>0} \mathcal{K}_{\vartheta,A}\,,
\end{equation}
where
\begin{equation}
\label{K_eta,A}
\mathcal{K}_{\vartheta,A}:=\Biggl\{\psi \in H^s(\Lambda):\, \|\psi\|_{L^2}^2 \leq K\,,\quad \|u(t)\|_{H^s} \leq A \log \biggl( \frac{1+|t|}{\vartheta}\biggr)^{s+\nu}\Biggr\}\,.
\end{equation}
In \eqref{K_eta,A}, the constant $K$ is as in Assumption \ref{f_assumption} and $u$ is the solution of \eqref{q_Hartree_approximation_cauchy_problem_1} with initial data $\psi$.

By \eqref{G_union}, we deduce that the claim follows if we show that it holds for $\psi \in \mathcal{K}_{\vartheta,A}$ with $\vartheta,A>0$ fixed. Throughout the proof, we consider $\varepsilon$ sufficiently small such that 
\begin{equation}
\label{w^epsilon_bound}
\|w^{\varepsilon}\|_{L^1} \leq 2 \|w\|_{L^1}\,.
\end{equation}
This is possible by \eqref{w_epsilon_local_problem_definition}.
For fixed $\vartheta,A$, we let 
\begin{equation}
\label{mathcal_A_choice}
\mathcal{A}:=A \, \log \biggl( \frac{1+|T|}{\vartheta}\biggr)^{s+\nu}\,.
\end{equation}
We fix $\theta>0$ small and consider
\begin{equation}
\label{b_choice_Section_5}
b=\frac{1}{2}+\theta\,.
\end{equation}
With parameters as above, we consider $\delta>0$ such that 
\begin{equation}
\label{delta_choice_Section_5}
\delta^{\theta} \,\|w\|_{L^1}^2 \, (\mathcal{A}+1)^4 \ll 1\,,
\end{equation}
where the smallness of the right-hand side of \eqref{delta_choice_Section_5} will be determined later (depending on the other parameters above).
For the remainder of the proof, we let 
\begin{equation}
\label{Section_5_choice_of_n}
n:=\lfloor T/\delta \rfloor\,.
\end{equation}
Let us now show that the following properties hold for $k=0,1,\ldots,n-1$.
\begin{itemize}
\item[(1)] For $\varepsilon>0$ small, we have $u^{\varepsilon} \in L^{\infty}_{[0,k\delta]}H^s_x$ and
\begin{equation}
\label{Section_5_Induction_1}
\lim_{\varepsilon \rightarrow 0} \|u^{\varepsilon}(k\delta)-u(k\delta)\|_{H^s}=0\,.
\end{equation}
\item[(2)] There exists $\varepsilon_k>0$ such that for all $\varepsilon \in (0,\varepsilon_k)$, and for a constant $C>0$ independent of $k$, we have $u^{\varepsilon} \in X^{s,b}_{[k\delta,(k+1)\delta]}$ and

\begin{multline}
\label{Section_5_Induction_2}
\|u^{\varepsilon}-u\|_{X^{s,b}_{[k\delta,(k+1)\delta]}} \leq C\|u^{\varepsilon}(k\delta)-u(k\delta)\|_{H^s}
\\
+C \delta^{\theta}\, \|w^{\varepsilon}-w\|_{L^1}\,\|w\|_{L^1}\,(\mathcal{A}+1)^5\,.
\end{multline}
\end{itemize}

We show \eqref{Section_5_Induction_1}--\eqref{Section_5_Induction_2} by induction on $k$.

\paragraph{\textbf{Base}} 
We consider $k=0$. Note that \eqref{Section_5_Induction_1} is automatically satisfied since $u^{\varepsilon}(0)=u(0)=\psi$. Using \eqref{w^epsilon_bound}--\eqref{delta_choice_Section_5} and arguing analogously as in the proof of Proposition \ref{Cauchy_problem_1}, we deduce that
\begin{equation}
\label{Section_5_Base_1}
u^{\varepsilon} \in X^{s,b}_{[0,\delta]}\,, \qquad \|u^{\varepsilon}\|_{X^{s,b}_{[0,\delta]}} \lesssim \mathcal{A}\,.
\end{equation}
Let us note that we also have
\begin{equation}
\label{Section_5_Base_2}
u \in X^{s,b}_{[0,\delta]}\,, \qquad \|u\|_{X^{s,b}_{[0,\delta]}} \lesssim \mathcal{A}\,,
\end{equation}
by the same argument. 

We use Duhamel's principle for the difference equation solved by $u^{\varepsilon}-u$ and write for $t \in [0,\delta]$
\begin{multline}
\label{Section_5_Base_3}
u^\eps(x,t) - u(x,t) =
\\
\frac{\mathrm{i}}{3} \,\int_0^t \dd t' \, \e^{\mathrm{i}(t-t')\Delta} \biggl[\int \dd y \, \dd z
 \, w^\eps(\cdot-y)\,w^\eps(\cdot-z) |u^\eps(y,t')|^2 |u^\eps(z,t')|^2 u^\eps(\cdot,t') \\
- \int \dd y \, \dd z\, w(\cdot-y)\, w(\cdot-z)\, |u(y,t')|^2\, |u(z,t')|^2 \,u(\cdot,t')\biggr](x)\\
+\frac{2\mathrm{i}}{3} \,\int_0^t \dd t' \, \e^{\mathrm{i}(t-t')\Delta} \biggl[\int \dd y \, \dd z
\, w^\eps(\cdot-y)\, w^\eps (y-z)\, |u^\eps(y)|^2 |u^\eps(z)|^2 u^\eps(\cdot) \\
- \int \dd y \, \dd z\, w(\cdot-y)\, w(y-z)\,  |u(y,t')|^2\, |u(z,t')|^2 \,u(\cdot,t')\biggr](x)\,.
\end{multline}
By using telescoping, Lemma \ref{X^{s,b}_space_properties} (iii), Lemma \ref{Multlinear_estimates_2}, and \eqref{w^epsilon_bound}, we deduce from \eqref{Section_5_Base_3} that
\begin{multline}
\label{Section_5_Base_4}
\|u^\eps - u\|_{X^{s,b}_{[0,\delta]}} 
\leq  C \delta^{\theta}\,\|w^{\varepsilon}-w\|_{L^1} \, \|w\|_{L^1}\, \bigl(\|u^{\varepsilon}\|_{X^{s,b}_{[0,\delta]}}^5 +\|u\|_{X^{s,b}_{[0,\delta]}}^5 \bigr)
\\
+C \delta^{\theta}\, \|w\|_{L^1}^2\, \|u^{\varepsilon} - u\|_{X^{s,b}_{[0,\delta]}}\, \bigl(\|u^{\varepsilon}\|_{X^{s,b}_{[0,\delta]}}^4 +\|u\|_{X^{s,b}_{[0,\delta]}}^4\bigr)
\,.
\end{multline}
Using \eqref{Section_5_Base_1}--\eqref{Section_5_Base_2}, followed by \eqref{delta_choice_Section_5} (with sufficiently small right-hand side), we obtain from \eqref{Section_5_Base_4} that
\begin{equation}
\label{Section_5_Base_5}
\|u^\eps - u\|_{X^{s,b}_{[0,\delta]}} \leq C \delta^{\theta}\, \|w^{\varepsilon}-w\|_{L^1}\,\|w\|_{L^1}\, \mathcal{A}^5\,,
\end{equation}
with a different choice of $C$. We obtain \eqref{Section_5_Induction_2} from \eqref{Section_5_Base_5}.
\paragraph{\textbf{Inductive Step}}
Suppose that \eqref{Section_5_Induction_1}--\eqref{Section_5_Induction_2} hold for some $0 \leq k \leq n-2$. 
Let us observe that by Lemma \ref{X^{s,b}_space_properties} (i) and \eqref{Section_5_Induction_2} for $k$, we have that for $\eps \in (0,\eps_k)$
\begin{multline}
\label{Section_5_Base_6}
\|u^{\varepsilon}((k+1)\delta)-u((k+1)\delta)\|_{H^s} \lesssim \|u^{\varepsilon}-u\|_{X^{s,b}_{[k\delta,(k+1)\delta]}}
\\
\leq C \|u^{\varepsilon}(k\delta)-u(k\delta)\|_{H^s}+C\delta^{\theta} \|w^{\varepsilon}-w\|_{L^1}\,\|w\|_{L^1}^2\,(\mathcal{A}+1)^5\,.
\end{multline}
We deduce
\begin{equation}
\label{Section_5_Step_1}
\lim_{\varepsilon \rightarrow 0} \|u^{\varepsilon}((k+1)\delta)-u((k+1)\delta)\|_{H^s}=0
\end{equation}
from \eqref{Section_5_Base_6}, combined with \eqref{Section_5_Induction_1} for $k$ and 
\eqref{w^{epsilon}_convergence_approximation}. This shows that \eqref{Section_5_Induction_1} holds for $k+1$.

By \eqref{Section_5_Step_1}, it follows that there exists $\varepsilon_{k+1} \in (0,\eps_k)$ small enough such that for all $\varepsilon \in (0,\varepsilon_{k+1})$, we have
\begin{equation}
\label{Section_5_Step_2}
\|u^{\varepsilon}((k+1)\delta)-u((k+1)\delta)\|_{H^s} \leq 1\,.
\end{equation}
From \eqref{Section_5_Step_2}, we deduce that for all $\varepsilon \in (0,\varepsilon_{k+1})$, we have
\begin{equation}
\label{Section_5_Step_3}
\|u^{\varepsilon}((k+1)\delta)\|_{H^s} \leq \|u((k+1)\delta)\|_{H^s}+1 \leq \mathcal{A}+1\,.
\end{equation}
For \eqref{Section_5_Step_3}, we recalled \eqref{K_eta,A} and \eqref{mathcal_A_choice} above.
Using \eqref{Section_5_Step_3}, recalling \eqref{w^epsilon_bound}, \eqref{b_choice_Section_5}--\eqref{delta_choice_Section_5} and arguing analogously as in the proof of Proposition \ref{Cauchy_problem_1}, we deduce that
\begin{equation}
\label{Section_5_Step_4}
u^{\varepsilon} \in X^{s,b}_{[(k+1)\delta,(k+2)\delta]}\,, \qquad \|u^{\varepsilon}\|_{X^{s,b}_{[(k+1)\delta,(k+2)\delta]}} \lesssim \mathcal{A}+1\,.
\end{equation}
As in \eqref{Section_5_Base_2} (with the same argument), we have
\begin{equation}
\label{Section_5_Step_5}
u \in X^{s,b}_{[(k+1)\delta,(k+2)\delta]}\,, \qquad \|u\|_{X^{s,b}_{[(k+1)\delta,(k+2)\delta]}} \lesssim \mathcal{A}\,.
\end{equation}
Similarly as for \eqref{Section_5_Base_3}, we use Duhamel's principle to write for $t \in [(k+1)\delta,(k+2)\delta]$
\begin{multline}
\label{Section_5_Step_6}
u^\eps(x,t) - u(x,t) =\e^{\mathrm{i}(t-(k+1)\delta) \Delta}\Bigl(u^{\varepsilon}(\cdot,(k+1)\delta)-u(\cdot,(k+1)\delta)\Bigr)(x)
\\
+\frac{\mathrm{i}}{3} \,\int_{(k+1)\delta}^t \dd t' \, \e^{\mathrm{i}(t-t')\Delta} \biggl[\int \dd y \, \dd z
        \, w^\eps(\cdot-y)\, w^\eps(\cdot-z) |u^\eps(y,t')|^2 \,|u^\eps(z,t')|^2\, u^\eps(\cdot,t') \\
        - \int \dd y \, \dd z\, w(\cdot-y)\,w(\cdot-z)\, |u(y,t')|^2\, |u(z,t')|^2 \,u(\cdot,t')\biggr](x)\\
+\frac{2\mathrm{i}}{3} \,\int_{(k+1)\delta}^t \dd t' \, \e^{\mathrm{i}(t-t')\Delta} \biggl[\int \dd y \, \dd z
        \, w^\eps(\cdot-y)\, w^\eps (y-z)\,|u^\eps(y,t')|^2 \,|u^\eps(z,t')|^2 \,u^\eps(\cdot,t') \\
        - \int \dd y \, \dd z\, w(\cdot-y)\, w(y-z)\, |u(y,t')|^2\, |u(z,t')|^2 \,u(\cdot,t')\biggr](x)
	\end{multline}
Using Lemma \ref{X^{s,b}_space_properties} (ii) to estimate the first term on the right-hand side of \eqref{Section_5_Step_6}, and using telescoping, Lemma \ref{X^{s,b}_space_properties} (iii), Lemma \ref{Multlinear_estimates_2}, and \eqref{w^epsilon_bound} (as for \eqref{Section_5_Base_4}) to estimate the second term, we deduce that
\begin{multline}
\label{Section_5_Step_7}
\|u^\eps - u\|_{X^{s,b}_{[(k+1)\delta,(k+2)\delta]}} 
\leq  C\|u^{\varepsilon}((k+1)\delta)-u((k+1)\delta)\|_{H^s}
\\
+C \delta^{\theta}\,\|w^{\varepsilon}-w\|_{L^1} \, \|w\|_{L^1}\, \bigl(\|u^{\varepsilon}\|_{X^{s,b}_{[(k+1)\delta,(k+2)\delta]}}^5 +\|u\|_{X^{s,b}_{[(k+1)\delta,(k+2)\delta]}}^5 \bigr)
\\
+C \delta^{\theta}\, \|w\|_{L^1}^2\, \|u^{\varepsilon} - u\|_{X^{s,b}_{[(k+1)\delta,(k+2)\delta]}}\, \bigl(\|u^{\varepsilon}\|_{X^{s,b}_{[(k+1)\delta,(k+2)\delta]}}^4 +\|u\|_{X^{s,b}_{[(k+1)\delta,(k+2)\delta]}}^4\bigr)
\,.
\end{multline}
We now use \eqref{Section_5_Step_4}--\eqref{Section_5_Step_5} and \eqref{delta_choice_Section_5} (for sufficiently small right-hand side) to deduce that 
\begin{multline}
\label{Section_5_Step_8}
\|u^\eps - u\|_{X^{s,b}_{[(k+1)\delta,(k+2)\delta]}} \leq C\|u^{\varepsilon}((k+1)\delta)-u((k+1)\delta\|_{H^s}
\\
+ C \delta^{\theta}\, \|w^{\varepsilon}-w\|_{L^1}\,\|w\|_{L^1}\, (\mathcal{A}+1)^5\,,
\end{multline}
for a different choice of $C$. From \eqref{Section_5_Step_8}, we conclude the induction.

Using \eqref{Section_5_Induction_1}--\eqref{Section_5_Induction_2}, Lemma \ref{X^{s,b}_space_properties} (i), and \eqref{w^{epsilon}_convergence_approximation}, it follows that
\begin{equation}
\label{Lemma_5.3_Conclusion_1}
\lim_{\varepsilon \rightarrow 0} \|u^{\varepsilon}-u\|_{L^{\infty}_{[0,n\delta]} H^s_x}=0\,,
\end{equation}
for $n$ as in \eqref{Section_5_choice_of_n}. By using $\|u\|_{L^{\infty}_{[n\delta,T]} H^s_x} \leq \mathcal{A}$, and by repeating the above argument, we also obtain
\begin{equation}
\label{Lemma_5.3_Conclusion_2}
\lim_{\varepsilon \rightarrow 0} \|u^{\varepsilon}-u\|_{L^{\infty}_{[n\delta,T]} H^s_x}=0\,.
\end{equation}
The lemma follows from \eqref{Lemma_5.3_Conclusion_1}--\eqref{Lemma_5.3_Conclusion_2}.
\end{proof}

\begin{remark}
The proof of Lemma \ref{q_Hartree_approximation_lemma} above is more involved than that of the analogous cubic result \cite[Proposition 5.1]{FKSS18}, \cite[Lemma 5.4]{RS22}. It requires full use of the set $\mathcal{G}$ of initial data leading to global solutions of \eqref{q_Hartree_defn}.
\end{remark}

Let us recall the following form of the diagonal argument \cite[Lemma 5.5]{FKSS18}.

\begin{lemma}[Diagonal argument]
\label{diagonal_argument_lemma}
Let $(\Gamma_k)$ be a sequence of sets with $\Gamma_k \subset \Gamma_{k+1}$ and let $\Gamma:=\cup_k \Gamma_k$.
Given $\varepsilon, \tau>0$, let $g,g^{\varepsilon}, g^{\varepsilon}_{\tau} : \Gamma \rightarrow \mathbb{C}$ 
be functions satisfying the following properties.
\begin{itemize}
\item[(i)] For $k \in \N$ and $\varepsilon>0$ fixed, we have $\lim_{\tau \rightarrow \infty} g^{\varepsilon}_{\tau}(\zeta)=g^{\varepsilon}(\zeta)$, uniformly in $\zeta \in \Gamma_k$.
\item[(ii)] For fixed $k \in \N$, we have $\lim_{\varepsilon \rightarrow 0} g^{\varepsilon}(\zeta)=g(\zeta)$, uniformly in $\zeta \in \Gamma_k$. 
\end{itemize}
Then, there exists a sequence $(\varepsilon_{\tau})$ converging to zero as $\tau \rightarrow \infty$ such that 
\begin{equation*}
\lim_{\tau \rightarrow \infty} g_{\tau}^{\varepsilon_{\tau}}(\zeta)=g(\zeta)\,.
\end{equation*}
\end{lemma}

We now have all of the necessary tools to prove Theorem \ref{q_L1_time_thm}.

\begin{proof}[Proof of Theorem \ref{q_L1_time_thm}]
Once we have Lemma \ref{q_Hartree_approximation_lemma} at our disposal, the proof is quite similar to that of \cite[Theorem 1.11]{RS22}. We just outline the main differences and refer the reader to \cite[Section 5.2]{RS22} for more details.
We adopt the convention that the superscript $\varepsilon$ denotes an object defined using the interaction potential $w^{\varepsilon}$. Using Theorem \ref{q_bounded_time_thm}, Lemma \ref{diagonal_argument_lemma}, and arguing analogously as in the proof of \cite[Theorem 1.11]{RS22}, the claim follows if we show that
\begin{equation}
\label{rho_eps_time_convergence_2}
\lim_{\varepsilon \to 0} \tilde{\rho}^{\varepsilon}_{1}\left( \Psi^{t_{1},\varepsilon}\Theta(\xi^{1}) \cdots \Psi^{t_{m},\varepsilon}\Theta(\xi^{m}) \right) = \tilde{\rho}_1\left( \Psi^{t_{1}}\Theta(\xi^{1}) \cdots \Psi^{t_{m}}\Theta(\xi^{m})\right)\,,
\end{equation}
uniformly in $\Gamma_k$
\begin{equation}
\label{Gamma_k}
\Gamma_{k}:= \bigl\{(m,t_i,p_i,\xi^i):\, m \leq k,\, |t_i| \leq k,\,p_i \in \mathbb{N}^{*}\,, p_i \leq k,\, \|\xi^i\| \leq k\,, 1 \leq i \leq m\bigr\}\,.
\end{equation}
In \eqref{Gamma_k}, we take $\xi^{i} \in \mathcal{L} (\mathfrak{h}^{(p_i)})$. We also recall \eqref{q_time_evolved_rv}, \eqref{classical_state_series}, and define the quantities with superscript $\varepsilon$ accordingly.

Let $\varphi$ denote the classical free field given by \eqref{q_class_free_field}. Recall that by \eqref{q_regularity_vph}, we have that for $s \in (0,\frac{1}{2})$, $\varphi \in H^s$ $\mu$-almost surely.
Let us recall the definition \eqref{q_random_var_defn} of $\Theta(\xi)$ and
the definition \eqref{q_evolution_defn} of the flow map $S_t(\cdot)$ of \eqref{q_Hartree_defn}. By suitably defining the quantities with superscript $\varepsilon$, we have for $\xi \in \mathcal{L}(\mathfrak{h}^{(k)})$
\begin{equation}
\label{Thm_unbounded_1}
{\Psi}^{t,\varepsilon}\Theta(\xi)=\left\langle \left(S^{\varepsilon}_{t} \,\varphi \right)^{\otimes_k}\,, \xi \left(S^{\varepsilon}_{t} \varphi \right)^{\otimes_k} \right\rangle_{\mathfrak{h}^{\otimes_k}}\,,\quad 
{\Psi}^{t}\Theta(\xi)= \left\langle \left(S_{t} \varphi \right)^{\otimes_k}\,, \xi \left(S_{t} \varphi \right)^{\otimes_k} \right\rangle_{\mathfrak{h}^{\otimes_k}}\,.
\end{equation}
When $\xi \in \mathcal{L}(\mathfrak{h}^{(k)})$,  Lemma \ref{q_Hartree_approximation_lemma} and \eqref{Thm_unbounded_1} imply that
\begin{equation}
\label{Thm_unbounded_2}
\lim_{\varepsilon \rightarrow 0} {\Psi}^{t,\varepsilon} {\Theta}(\xi) = {\Psi}^{t}\Theta(\xi)\,,
\end{equation}
$\mu$-almost surely.

Recalling \eqref{(4.5)_claim}, and using \eqref{Thm_unbounded_2}, it follows that 
\begin{equation}
\label{Thm_unbounded_3}
\lim_{\varepsilon \rightarrow 0} {\Psi}^{{t}_{1},\varepsilon} \Theta({\xi}^{1}) \cdots \Psi^{t_{m},\varepsilon} \Theta(\xi^{m}) \,\e^{-\mathcal{W}^{\varepsilon}} = \Psi^{t_{1}} {\Theta}({\xi}^{1})\cdots \Psi^{t_m} \Theta(\xi^m) \,\e^{-\mathcal{W}}\,,
\end{equation}
$\mu$-almost surely.

Using Lemma \ref{Multlinear_estimates_1}, Lemma \ref{q_random_var_estimate}, conservation of mass for \eqref{q_Hartree_approximation_cauchy_problem_2} and \eqref{q_Hartree_approximation_cauchy_problem_1}, as well as \eqref{w^{epsilon}_convergence_approximation}, we have the following bounds for $\varepsilon>0$ sufficiently small.
\begin{align}
\notag
\bigl|\Psi^{t_1,\varepsilon} {\Theta}({\xi}^{1}) \cdots \Psi^{t_m,\varepsilon} \Theta(\xi^m) \e^{-\mathcal{W}^{\varepsilon}} f\left(\mcN\right)\bigr| &\leq \Biggl(\prod_{j=1}^m \|\xi^j\| \|\varphi\|^{2p_j}_{\mathfrak{h}} \Biggr)\, \e^{c \|w\|_{L^1}^2 \|\varphi\|_{L^6}^6}\,f(\mathcal{N})
\\
\label{Thm_unbounded_4}
\bigl|\Psi^{t_1} \Theta(\xi^1) \cdots \Psi^{t_m} \Theta(\xi^m) e^{-\mcW} f(\mathcal{N})\bigr| &\leq \Biggl(\prod_{j=1}^m \|\xi^j\| \|\varphi\|^{2p_j}_{\mathfrak{h}} \Biggr) \e^{c \|w\|_{L^1}^2 \|\varphi\|_{L^6}^6}\,f(\mathcal{N})\,.
\end{align}
The claim now follows from \eqref{Thm_unbounded_3}--\eqref{Thm_unbounded_4}, by using Proposition \ref{Cauchy_problem_2} (i), Assumption \ref{f_assumption} (with $K$ sufficiently small) and the dominated convergence theorem.

We now consider $w$ as in Assumption \ref{w_assumption} (ii) and prove Theorem \ref{q_delta_time_thm}.
For simplicity of notation, throughout the sequel, we set $c=1$ in Assumption \ref{w_assumption} (ii). The general case follows in the same way.
Let us note the following analogue of Lemma \ref{q_Hartree_approximation_lemma}. 

\begin{lemma}
\label{q_Hartree_approximation_lemma_delta}
Let $s>0$ and let $\mathcal{G} \subset H^s(\Lambda)$ be the set constructed in Proposition \ref{Cauchy_problem_2} (ii). Let $w^\eps$ be given as in \eqref{w_epsilon_local_problem_definition}. Consider $T>0$ and $\psi \in \mathcal{G}$. The following claims hold.
\begin{enumerate}
\item[(i)]
For $\varepsilon>0$, the following Cauchy problem admits a solution on $[-T,T]$.
\begin{equation}
\label{q_Hartree_approximation_cauchy_problem_2_delta}
\begin{cases}
\mathrm{i} \partial_t u^\eps+(\Delta - \kappa)u^\eps =
\\
-\frac{1}{3}  \int \dd y \, \dd z \, w^\eps(x-y)\,w^\eps(x-z)\, |u^\eps(y)|^2\,|u^\eps(z)|^2\,u^\eps(x) 
\\
-\frac{2}{3}  \int \dd y \, \dd z \, w^\eps(x-y)\,w^\eps(y-z)\, |u^\eps(y)|^2\,|u^\eps(z)|^2\,u^\eps(x) 
\\
u^\eps|_{t=0} = \psi\,.
\end{cases}
\end{equation}
\item[(ii)]
We have
\begin{equation}
\label{Lemma_5.3_convergence_delta}
\lim_{\eps \to 0} \|u^{\varepsilon}- u\|_{L^\infty_{[-T,T]}H^s} = 0\,,
\end{equation}
where 
\begin{equation}
\label{q_Hartree_approximation_cauchy_problem_1_delta}
\begin{cases}
\mathrm{i}\partial_t u +(\Delta-\kappa)u=-|u|^4\,u
\\
u|_{t=0} = \psi\,.
\end{cases}
\end{equation}
\end{enumerate}
\end{lemma}
\end{proof}
The proof of Lemma \ref{q_Hartree_approximation_lemma_delta} is similar to that of Lemma \ref{q_Hartree_approximation_lemma} given above. The main difference is that we can no longer use the smallness of the factor  $\|w^{\varepsilon}-w\|_{L^1}$. Instead, we use a qualitative approach based on the dominated convergence theorem.

\begin{proof}
The proof proceeds as for Lemma \ref{q_Hartree_approximation_lemma}. We outline the main differences. As before, we reduce to considering $\kappa=0$ and positive times. We consider $\nu>0$ and 
use notation as in \eqref{G_union}--\eqref{K_eta,A} and \eqref{mathcal_A_choice}--\eqref{b_choice_Section_5} above. Throughout, we also consider $\varepsilon$ such that \eqref{w^epsilon_bound} holds. We take the initial data $\psi \in \mathcal{K}_{\vartheta,A}$. Furthermore, we consider $\delta>0$ such that 
\begin{equation}
\label{delta_choice_Section_5_delta}
\delta^{\theta} \, (\mathcal{A}+1)^4 \ll 1\,,
\end{equation}
where the smallness of the right-hand side of \eqref{delta_choice_Section_5_delta} will be determined later.
With $\delta$ as in \eqref{delta_choice_Section_5_delta}, we consider $n \in \N$ defined as in \eqref{Section_5_choice_of_n}.

We show that the following properties hold for $k=0,1,\ldots,n-1$.
\begin{itemize}
\item[(1)] For $\varepsilon>0$ small, we have $u^{\varepsilon} \in L^{\infty}_{[0,k\delta]}H^s_x$ and
\begin{equation}
\label{Section_5_Induction_1_delta}
\lim_{\varepsilon \rightarrow 0} \|u^{\varepsilon}(k\delta)-u(k\delta)\|_{H^s}=0\,.
\end{equation}
\item[(2)] There exists $\varepsilon_k>0$ such that for all $\varepsilon \in (0,\varepsilon_k)$, and for a constant $C>0$ independent of $k$, we have $u^{\varepsilon} \in X^{s,b}_{[k\delta,(k+1)\delta]}$ and

\begin{equation}
\label{Section_5_Induction_2_delta}
\|u^{\varepsilon}-u\|_{X^{s,b}_{[k\delta,(k+1)\delta]}} \leq C\|u^{\varepsilon}(k\delta)-u(k\delta)\|_{H^s}
+C\,\mathbf{F}(\varepsilon)\,(\mathcal{A}+1)^5\,.
\end{equation}
In \eqref{Section_5_Induction_2_delta} and throughout the sequel, $\mathbf{F}(\varepsilon)$ denotes a quantity such that $\lim_{\varepsilon \rightarrow 0} \mathbf{F}(\varepsilon)=0$.
\end{itemize}

We show \eqref{Section_5_Induction_1_delta}--\eqref{Section_5_Induction_2_delta} by induction on $k$.

\paragraph{\textbf{Base}} 
We consider $k=0$. The claim \eqref{Section_5_Induction_1_delta} is automatically satisfied since $u^{\varepsilon}(0)=u(0)=\psi$. 
%Let us note that by \eqref{w_epsilon_local_problem_definition}, we have 
%\begin{equation}
%\label{w^epsilon_bound_delta}
%\|w^{\varepsilon}\|_{L^1} \leq \|\Phi\|_{L^1}\,,
%\end{equation}
%uniformly in $\varepsilon>0$ small. 
Let us now show \eqref{Section_5_Induction_2_delta}. Using \eqref{w^epsilon_bound} and arguing as for \eqref{Section_5_Base_1}, we have 
\begin{equation}
\label{Section_5_Base_1_delta}
u^{\varepsilon} \in X^{s,b}_{[0,\delta]}\,, \qquad \|u^{\varepsilon}\|_{X^{s,b}_{[0,\delta]}} \lesssim \mathcal{A}\,.
\end{equation}
Similarly, for the local problem, we have
\begin{equation}
\label{Section_5_Base_2_delta}
u \in X^{s,b}_{[0,\delta]}\,, \qquad \|u\|_{X^{s,b}_{[0,\delta]}} \lesssim \mathcal{A}\,,
\end{equation}
(see \cite{Bou93}).

Instead of \eqref{Section_5_Base_3}, we write for $t \in [0,\delta]$
\begin{multline}
\label{Section_5_Base_3_delta}
u^\eps(x,t) - u(x,t) =
\\
\frac{\mathrm{i}}{3} \,\int_0^t \dd t' \, \e^{\mathrm{i}(t-t')\Delta} \biggl[\int \dd y \, \dd z
 \, w^\eps(\cdot-y)\,w^\eps(\cdot-z) |u^\eps(y,t')|^2 |u^\eps(z,t')|^2 u^\eps(\cdot,t') 
\\- |u(\cdot,t')|^4\,u(\cdot,t')\biggr](x)
\\
+\frac{2\mathrm{i}}{3} \,\int_0^t \dd t' \, \e^{\mathrm{i}(t-t')\Delta} \biggl[\int \dd y \, \dd z
\, w^\eps(\cdot-y)\, w^\eps (y-z)\, |u^\eps(y,t')|^2 |u^\eps(z,t')|^2 u^\eps(\cdot,t') 
\\
- |u(\cdot,t')|^4\,u(\cdot,t')\biggr](x)
\,.
\end{multline}
From \eqref{Section_5_Base_3_delta} and Lemma \ref{X^{s,b}_space_properties} (iii), we deduce that
\begin{multline}
\label{Section_5_Base_3_delta_2}
\|u^\eps - u\|_{X^{s,b}_{[0,\delta]}} \lesssim
\\
\Biggl\|\int \dd y \, \dd z
 \, w^\eps(x-y)\,w^\eps(x-z) |u^\eps(y,t)|^2 |u^\eps(z,t)|^2 u^\eps(x,t) 
- |u(x,t)|^4\,u(x,t)\Biggr\|_{X^{s,b-1}_{[0,\delta]}} 
\\
+
\Biggl\|\int \dd y \, \dd z
 \, w^\eps(x-y)\,w^\eps(y-z) |u^\eps(y,t)|^2 |u^\eps(z,t)|^2 u^\eps(x,t) 
- |u(x,t)|^4\,u(x,t)\Biggr\|_{X^{s,b-1}_{[0,\delta]}}
\\
=:I+II\,.
\end{multline}
We have 
\begin{equation}
\label{I_1+I_2}
I \leq I_1+I_2\,,
\end{equation}
where 
\begin{multline}
\label{I_1_definition_delta}
I_1:=\Biggl\|\int \dd y \, \dd z
 \, w^\eps(x-y)\,w^\eps(x-z) |u^\eps(y,t)|^2 |u^\eps(z,t)|^2 u^\eps(x,t) 
\\
-\int \dd y \, \dd z
 \, w^\eps(x-y)\,w^\eps(x-z) |u(y,t)|^2 |u(z,t)|^2 u(x,t) \Biggr\|_{X^{s,b-1}_{[0,\delta]}} 
\end{multline}
and
\begin{equation}
\label{I_2_definition_delta}
I_2:=\Biggl\|\int \dd y \, \dd z
 \, w^\eps(x-y)\,w^\eps(x-z) |u(y,t)|^2 |u(z,t)|^2 u(x,t) -|u(x,t)|^4\,u(x,t) \Biggr\|_{X^{s,b-1}_{[0,\delta]}}\,.
\end{equation}
By telescoping in \eqref{I_1_definition_delta}, and using Lemma \ref{Multlinear_estimates_2} as well as \eqref{w^epsilon_bound} and \eqref{Section_5_Base_1_delta}--\eqref{Section_5_Base_2_delta}, it follows that 
\begin{equation}
\label{I_1_bound_delta}
I_1 \lesssim \delta^{\theta}\, \|w^{\varepsilon}\|_{L^1}^2\, \|u^{\varepsilon} - u\|_{X^{s,b}_{[0,\delta]}}\, \bigl(\|u^{\varepsilon}\|_{X^{s,b}_{[0,\delta]}}^4 +\|u\|_{X^{s,b}_{[0,\delta]}}^4\bigr) \lesssim \delta^{\theta}\, \|u^{\varepsilon} - u\|_{X^{s,b}_{[0,\delta]}}\,\mathcal{A}^4\,.
\end{equation}
In order to estimate \eqref{I_2_definition_delta}, let us fix $U:\T \times \R \to \C$ such that 
\begin{equation}
\label{U_construction}
U|_{\T \times [0,\delta]} = u|_{\T \times [0,\delta]}\,,\quad \|U\|_{X^{s,b}} \sim \|u\|_{X^{s,b}_{[0,\delta]}}\lesssim \mathcal{A}\,.
\end{equation}
Here, we recalled \eqref{local_X^{s,b}_norm} and \eqref{Section_5_Base_2_delta}.
Let us define $G^{\eps}:\T \times \R \to \C$ as
\begin{multline}
\label{I_2_bound_delta_1_A}
G^{\eps}(x,t):=
\\
\int \dd y \, \dd z
 \, w^\eps(x-y)\,w^\eps(x-z) |U(y,t)|^2\, |U(z,t)|^2\, U(x,t) -|U(x,t)|^4\,U(x,t)\,.
\end{multline}
We argue as for \eqref{N_1_claim} to compute
\begin{multline}
\label{I_2_bound_delta_1}
(G^{\eps})\,\widetilde{\,}\, (k,\eta)=
%\\
%\mathfrak{F} \biggl(\int \dd y \, \dd z
%\, w^\eps(x-y)\,w^\eps(x-z) |U(y,t)|^2\, |U(z,t)|^2\, U(x,t) -|U(x,t)|^4\,U(x,t)\biggr)\,(k,\eta)
\\
= \sum_{k_1,\ldots,k_5} \int \dd \eta_1 \cdots \dd \eta_5\, \delta(k_1-k_2+k_3-k_4+k_5-k)\,\delta(\eta_1-\eta_2+\eta_3-\eta_4+\eta_5-\eta)\,
\\
\times \bigl\{\widehat{w}^{\eps}(k_1-k_2)\,\widehat{w}^{\eps}(k_3-k_4)-1\bigr\}\,\widetilde{U}(k_1,\eta_1)\,\overline{\widetilde{U}(k_2,\eta_2)}\,\widetilde{U}(k_3,\eta_3)\,\overline{\widetilde{U}(k_4,\eta_4)}\,\widetilde{U}(k_5,\eta_5)\,.
\end{multline}
From \eqref{w_epsilon_local_problem_definition}, it follows that there exists $C>0$ such that for all $k_1,k_2,k_3,k_4 \in \Z$ and $\varepsilon \in (0,1)$, we have
\begin{equation}
\label{w_epsilon_bound}
|\widehat{w}^{\eps}(k_1-k_2)\,\widehat{w}^{\eps}(k_3-k_4)-1| \leq C\,.
\end{equation}
Moreover, by \eqref{int_Phi} (and recalling that we had set $c=1$ for simplicity of notation), it follows that for all $k_1,k_2,k_3,k_4 \in \Z$
\begin{equation}
\label{w_epsilon_bound_2}
\lim_{\eps \rightarrow 0} \widehat{w}^{\eps}(k_1-k_2)\,\widehat{w}^{\eps}(k_3-k_4)=1\,.
\end{equation}
Now, let us define 
\begin{multline}
\label{I_2_bound_delta_2}
F(x,t):=
\\
\mathfrak{F}^{-1} \Biggl( \sum_{k_1,\ldots,k_5} \int \dd \eta_1 \cdots \dd \eta_5\, \delta(k_1-k_2+k_3-k_4+k_5-k)\,\delta(\eta_1-\eta_2+\eta_3-\eta_4+\eta_5-\eta)\,
\\
\times |\widetilde{U}(k_1,\eta_1)|\,|\widetilde{U}(k_2,\eta_2)|\,|\widetilde{U}(k_3,\eta_3)|\,|\widetilde{U}(k_4,\eta_4)|\,|\widetilde{U}(k_5,\eta_5)|\Biggr)(x,t)\,,
\end{multline}
where $\mathfrak{F}$ denotes the spacetime Fourier transform defined in \eqref{spacetime_Fourier_transform}. Then, arguing analogously as in the proof of Lemma \ref{Multlinear_estimates_2} (see \eqref{N_1_claim_A}--\eqref{F_j_choice_norm} above), we get that 
\begin{equation}
\label{I_2_bound_delta_3}
\|F\|_{X^{s,b-1}} \lesssim \|U\|_{X^{s,b}}^5 \lesssim \mathcal{A}^5\,.
\end{equation}
For the last inequality in \eqref{I_2_bound_delta_3}, we recalled \eqref{U_construction}.
By using Definition \ref{q_Xsb_defn} and \eqref{I_2_bound_delta_2}--\eqref{I_2_bound_delta_3}, it follows that for almost every $(k,\eta) \in \Z \times \R$ (with respect to the product of counting measure on $\Z$ and Lebesgue measure on $\R$), we have that 
\begin{multline}
\label{I_2_bound_delta_4}
|\widetilde{F}(k,\eta)|
\\
= \sum_{k_1,\ldots,k_5} \int \dd \eta_1 \cdots \dd \eta_5\, \delta(k_1-k_2+k_3-k_4+k_5-k)\,\delta(\eta_1-\eta_2+\eta_3-\eta_4+\eta_5-\eta)\,
\\
\times |\widetilde{U}(k_1,\eta_1)|\,|\widetilde{U}(k_2,\eta_2)|\,|\widetilde{U}(k_3,\eta_3)|\,|\widetilde{U}(k_4,\eta_4)|\,|\widetilde{U}(k_5,\eta_5)|<\infty\,.
\end{multline}
By \eqref{I_2_bound_delta_1}--\eqref{w_epsilon_bound}, \eqref{I_2_bound_delta_4}, and the dominated convergence theorem, it follows that for almost every $(k,\eta) \in \Z \times \R$, we have
\begin{equation}
\label{I_2_bound_delta_5}
\lim_{\varepsilon \rightarrow 0} (G^{\eps})\,\widetilde{\,}\, (k,\eta)=0\,.
\end{equation}
By \eqref{I_2_bound_delta_1}, \eqref{w_epsilon_bound}, and \eqref{I_2_bound_delta_2}, we have that for all $(k,\eta) \in \Z \times \R$ and $\varepsilon \in (0,1)$
\begin{equation}
\label{I_2_bound_delta_6}
|(G^{\eps})\,\widetilde{\,}\, (k,\eta)| \lesssim \widetilde{F}(k,\eta)\,.
\end{equation}
By using Definition \ref{q_Xsb_defn}, \eqref{I_2_bound_delta_3}, \eqref{I_2_bound_delta_5}--\eqref{I_2_bound_delta_6}, and the dominated convergence theorem, it follows that 
\begin{equation}
\label{I_2_bound_delta_7}
\lim_{\varepsilon \rightarrow 0} \|G^{\eps}\|_{X^{s,b-1}}=0\,.
\end{equation}
By \eqref{I_2_bound_delta_6} and \eqref{I_2_bound_delta_3}, we have that for all $\varepsilon \in (0,1)$
\begin{equation}
\label{I_2_bound_delta_8}
\|G^{\eps}\|_{X^{s,b-1}} \lesssim \mathcal{A}^5\,.
\end{equation}
Recalling \eqref{local_X^{s,b}_norm}, and applying \eqref{U_construction}--\eqref{I_2_bound_delta_1_A}, \eqref{I_2_bound_delta_7}--\eqref{I_2_bound_delta_8} in \eqref{I_2_definition_delta}, it follows that\footnote{Here, we could also absorb the factor of $\mathcal{A}^5$ into the $\mathbf{F}(\varepsilon)$, but we keep it instead to make the argument parallel with the proof of Lemma \ref{q_Hartree_approximation_lemma} above.}
\begin{equation}
\label{I_2_bound_delta_9}
I_2 \lesssim \mathbf{F}(\varepsilon)\,\mathcal{A}^5\,.
\end{equation}
From \eqref{I_1+I_2}, \eqref{I_1_bound_delta}, and \eqref{I_2_bound_delta_9}, it follows that 
\begin{equation}
\label{I_bound_delta}
I \lesssim \delta^{\theta}\, \|u^{\varepsilon} - u\|_{X^{s,b}_{[0,\delta]}}\,\mathcal{A}^4+ \mathbf{F}(\varepsilon)\,\mathcal{A}^5\,.
\end{equation}

We bound the term $II$ appearing in \eqref{Section_5_Base_3_delta_2} in a similar way. Instead of \eqref{I_1+I_2}, we write
\begin{equation}
\label{II_1+II_2}
II \leq II_1+II_2\,,
\end{equation}
where 
\begin{multline}
\label{II_1_definition_delta}
II_1:=\Biggl\|\int \dd y \, \dd z
 \, w^\eps(x-y)\,w^\eps(y-z) |u^\eps(y,t)|^2 |u^\eps(z,t)|^2 u^\eps(x,t) 
\\
-\int \dd y \, \dd z
 \, w^\eps(x-y)\,w^\eps(y-z) |u(y,t)|^2 |u(z,t)|^2 u(x,t) \Biggr\|_{X^{s,b-1}_{[0,\delta]}} 
\end{multline}
and
\begin{equation}
\label{II_2_definition_delta}
I_2:=\Biggl\|\int \dd y \, \dd z
 \, w^\eps(x-y)\,w^\eps(y-z) |u(y,t)|^2 |u(z,t)|^2 u(x,t) -|u(x,t)|^4\,u(x,t) \Biggr\|_{X^{s,b-1}_{[0,\delta]}}\,.
\end{equation}
Arguing analogously as for \eqref{I_1_bound_delta}, we get that 
\begin{equation}
\label{II_1_bound_delta}
II_1 %\lesssim \delta^{\theta}\, \|\Phi\|_{L^1}^2\, \|u^{\varepsilon} - u\|_{X^{s,b}_{[0,\delta]}}\, \bigl(\|u^{\varepsilon}\|_{X^{s,b}_{[0,\delta]}}^4 +\|u\|_{X^{s,b}_{[0,\delta]}}^4\bigr) 
\lesssim \delta^{\theta}\, \|u^{\varepsilon} - u\|_{X^{s,b}_{[0,\delta]}}\,\mathcal{A}^4\,.
\end{equation}
Let us define $H^{\eps}:\T \times \R \to \C$ as
\begin{multline}
\label{II_2_bound_delta_1_A}
H^{\eps}(x,t):=
\\
\int \dd y \, \dd z
 \, w^\eps(x-y)\,w^\eps(y-z) |U(y,t)|^2\, |U(z,t)|^2\, U(x,t) -|U(x,t)|^4\,U(x,t)\,.
\end{multline}
We argue as for \eqref{N_2_claim} to compute
\begin{multline}
\label{II_2_bound_delta_1}
%\\
%\mathfrak{F} \biggl(\int \dd y \, \dd z
%\, w^\eps(x-y)\,w^\eps(x-z) |U(y,t)|^2\, |U(z,t)|^2\, U(x,t) -|U(x,t)|^4\,U(x,t)\biggr)\,(k,\eta)
(H^{\eps})\,\widetilde{\,}\, (k,\eta)=\sum_{k_1,\ldots,k_5} \int \dd \eta_1 \cdots \dd \eta_5\, \delta(k_1-k_2+k_3-k_4+k_5-k)\,
\\ \times
\delta(\eta_1-\eta_2+\eta_3-\eta_4+\eta_5-\eta)\ \bigl\{\widehat{w}^{\eps}(k_1-k_2+k_3-k_4)\,\widehat{w}^{\eps}(k_3-k_4)-1\bigr\}\,
\\
\times
\widetilde{U}(k_1,\eta_1)\,\overline{\widetilde{U}(k_2,\eta_2)}\,\widetilde{U}(k_3,\eta_3)\,\overline{\widetilde{U}(k_4,\eta_4)}\,\widetilde{U}(k_5,\eta_5)\,.
\end{multline}
Using \eqref{II_2_bound_delta_1_A}--\eqref{II_2_bound_delta_1} instead of \eqref{I_2_bound_delta_1_A}--\eqref{I_2_bound_delta_1} and arguing analogously as for \eqref{II_2_bound_delta_9}, we obtain
\begin{equation}
\label{II_2_bound_delta_9}
II_2 \lesssim \mathbf{F}(\varepsilon)\,\mathcal{A}^5\,.
\end{equation}
From \eqref{II_1+II_2}, \eqref{II_1_bound_delta}, and \eqref{II_2_bound_delta_9}, it follows that 
\begin{equation}
\label{II_bound_delta}
II \lesssim \delta^{\theta}\, \|u^{\varepsilon} - u\|_{X^{s,b}_{[0,\delta]}}\,\mathcal{A}^4+ \mathbf{F}(\varepsilon)\,\mathcal{A}^5\,.
\end{equation}
The base of the induction follows from \eqref{Section_5_Base_3_delta_2}, \eqref{I_bound_delta}, and \eqref{II_bound_delta}.

\paragraph{\textbf{Inductive Step}}

Suppose that \eqref{Section_5_Induction_1_delta}--\eqref{Section_5_Induction_2_delta} hold for some $0 \leq k \leq n-2$. By Lemma \ref{X^{s,b}_space_properties} (i) and \eqref{Section_5_Induction_2_delta} for $k$, we have that for $\eps \in (0,\eps_k)$
\begin{multline}
\label{Section_5_Base_6_delta}
\|u^{\varepsilon}((k+1)\delta)-u((k+1)\delta)\|_{H^s} \lesssim \|u^{\varepsilon}-u\|_{X^{s,b}_{[k\delta,(k+1)\delta]}}
\\
\leq C \|u^{\varepsilon}(k\delta)-u(k\delta)\|_{H^s}+C\,\mathbf{F}(\varepsilon)\,(\mathcal{A}+1)^5\,.
\end{multline}
We deduce
\begin{equation}
\label{Section_5_Step_1_delta}
\lim_{\varepsilon \rightarrow 0} \|u^{\varepsilon}((k+1)\delta)-u((k+1)\delta)\|_{H^s}=0
\end{equation}
from \eqref{Section_5_Base_6_delta} and \eqref{Section_5_Induction_1_delta} for $k$. Hence, we obtain \eqref{Section_5_Induction_1_delta} for $k+1$.

By \eqref{Section_5_Step_1_delta}, we deduce that there exists $\varepsilon_{k+1} \in (0,\eps_k)$ small enough such that for all $\varepsilon \in (0,\varepsilon_{k+1})$, we have
\begin{equation}
\label{Section_5_Step_2_delta}
\|u^{\varepsilon}((k+1)\delta)-u((k+1)\delta)\|_{H^s} \leq 1\,.
\end{equation}
From \eqref{Section_5_Step_2_delta}, it follows that for all $\varepsilon \in (0,\varepsilon_{k+1})$, we have
\begin{equation}
\label{Section_5_Step_3_delta}
\|u^{\varepsilon}((k+1)\delta)\|_{H^s} \leq \|u((k+1)\delta)\|_{H^s}+1 \leq \mathcal{A}+1\,.
\end{equation}
%For \eqref{Section_5_Step_3}, we recalled \eqref{K_eta,A} and \eqref{mathcal_A_choice} above.
Using \eqref{Section_5_Step_3_delta} and arguing as for \eqref{Section_5_Step_4}, we deduce that
\begin{equation}
\label{Section_5_Step_4_delta}
u^{\varepsilon} \in X^{s,b}_{[(k+1)\delta,(k+2)\delta]}\,, \qquad \|u^{\varepsilon}\|_{X^{s,b}_{[(k+1)\delta,(k+2)\delta]}} \lesssim \mathcal{A}+1\,.
\end{equation}
Analogously as in \eqref{Section_5_Step_5}, we have 
\begin{equation}
\label{Section_5_Step_5_delta}
u \in X^{s,b}_{[(k+1)\delta,(k+2)\delta]}\,, \qquad \|u\|_{X^{s,b}_{[(k+1)\delta,(k+2)\delta]}} \lesssim \mathcal{A}\,.
\end{equation}
By Duhamel's principle, we have that for $t \in [(k+1)\delta,(k+2)\delta]$
\begin{multline}
\label{Section_5_Step_6_delta}
u^\eps(x,t) - u(x,t) =\e^{\mathrm{i}(t-(k+1)\delta) \Delta}\Bigl(u^{\varepsilon}(\cdot,(k+1)\delta)-u(\cdot,(k+1)\delta)\Bigr)(x)
\\
+\frac{\mathrm{i}}{3} \,\int_{(k+1)\delta}^t \dd t' \, \e^{\mathrm{i}(t-t')\Delta} \biggl[\int \dd y \, \dd z
        \, w^\eps(\cdot-y)\, w^\eps(\cdot-z) |u^\eps(y,t')|^2 \,|u^\eps(z,t')|^2\, u^\eps(\cdot,t') \\
        - |u(\cdot,t')|^4\,u(\cdot,t')\biggr](x)\\
+\frac{2\mathrm{i}}{3} \,\int_{(k+1)\delta}^t \dd t' \, \e^{\mathrm{i}(t-t')\Delta} \biggl[\int \dd y \, \dd z
        \, w^\eps(\cdot-y)\, w^\eps (y-z)\,|u^\eps(y,t')|^2 \,|u^\eps(z,t')|^2 \,u^\eps(\cdot,t') \\
        - |u(\cdot,t')|^4\,u(\cdot,t')\biggr](x)
	\end{multline}
We use Lemma \ref{X^{s,b}_space_properties} (ii) to estimate the first term on the right-hand side of \eqref{Section_5_Step_6_delta}. In order to estimate the second and third term, we use Lemma \ref{X^{s,b}_space_properties} (iii), \eqref{Section_5_Step_4_delta}, and \eqref{Section_5_Step_5_delta}, and apply analogous arguments to those in the proof of the base step. Putting everything together, we obtain  %and using telescoping, Lemma \ref{X^{s,b}_space_properties} (iii), Lemma \ref{Multlinear_estimates_2}, and \eqref{w^epsilon_bound} (as for \eqref{Section_5_Base_4}) to estimate the second term, we deduce that
\begin{equation*}
\|u^\eps - u\|_{X^{s,b}_{[(k+1)\delta,(k+2)\delta]}} 
\leq  C\|u^{\varepsilon}((k+1)\delta)-u((k+1)\delta)\|_{H^s}+C\,\mathbf{F}(\varepsilon)\,(\mathcal{A}+1)^5\,,
\end{equation*}
which concludes the proof of the inductive step. The claim of Lemma \ref{q_Hartree_approximation_lemma_delta} now follows from \eqref{Section_5_Induction_1_delta}--\eqref{Section_5_Induction_2_delta} analogously as in the proof of Lemma \ref{q_Hartree_approximation_lemma}.
\end{proof}

\begin{remark}
Let us note that, in the second term on the right-hand side of \eqref{Section_5_Induction_2_delta}, there is no factor of $\delta^{\theta}$ as there is in \eqref{Section_5_Induction_2}. This is because the analysis in the proof of \eqref{Section_5_Induction_1_delta}--\eqref{Section_5_Induction_2_delta} is based on the globally-defined function $U$ in \eqref{U_construction}. The $\delta^{\theta}$ factor is not necessary in order to prove Lemma  \ref{q_Hartree_approximation_lemma_delta}. Hence, this does not have any effect on the analysis.
\end{remark}

We now have all the necessary tools to prove Theorem \ref{q_delta_time_thm}.

\begin{proof}[Proof of Theorem \ref{q_delta_time_thm}]
The proof is analogous to that of Theorem \ref{q_L1_time_thm}. The only difference is that we apply Lemma \ref{q_Hartree_approximation_lemma_delta} instead of Lemma \ref{q_Hartree_approximation_lemma}.
\end{proof}

\appendix 

\section{An alternative model}
\label{Appendix A}
In this appendix, we introduce a different type of quintic Hartree (nonlocal quintic NLS) equation. We remark how Gibbs measures for the corresponding flow can be obtained as mean-field limits of Gibbs states for suitable three-body quantum interacting systems in sense studied above. The analysis is conceptually the same as before, but some of the estimates need to be done differently, due to the different algebraic structure of the nonlinearity. In the appendix, we only present the aforementioned estimates, and refer the reader to \cite[Chapter 4]{Rout_Thesis_2023} and \cite{RS23} for the full details.

Instead of \eqref{q_Hartree_defn}, we study

\begin{equation}
\label{A_q_Hartree_defn}
\mathrm{i}\partial_t u + (\Delta - \kappa)u =\int \dd y \, \dd z \, w(x-y)\,w(y-z)\,w(z-x)\,|u(y)|^2\,|u(z)|^2\,u(x)\,.
\end{equation}
When studying \eqref{A_q_Hartree_defn}, we consider interaction potentials satisfying the following assumption.
\begin{assumption}
\label{A_w_assumption}
Let $w: \Lambda \rightarrow \R$ be even and such that $w \in L^{\frac{3}{2}}(\Lambda)$.
\end{assumption}
In particular, unlike as in Assumption \ref{w_assumption}, we cannot consider $w=c\delta$, and thus we cannot consider the local problem in \eqref{A_q_Hartree_defn}.
By arguing as for Lemma \ref{Hamiltonian_Lemma}, we have the following result.
\begin{lemma}
\label{A_Hamiltonian_Lemma}
With Poisson structure given by \eqref{Poisson_structure}, we have that \eqref{q_Hartree_defn} corresponds to the Hamiltonian equation of motion associated with Hamiltonian
\begin{multline*}
H(u) = \int \dd x\, \left(|\nabla u(x)|^2 + \kappa |u(x)|^2\right) \\ 
+ \frac{1}{3}\int \dd x \, \dd y \, \dd z \, w(x-y)\,w(y-z)\,w(z-x)\,|u(x)|^2\,|u(y)|^2\,|u(z)|^2\,.
\end{multline*}
\end{lemma}
With $w$ as in Assumption \ref{A_w_assumption}, the classical interaction $\mathcal{W} \equiv \mathcal{W}^{\omega}$ is given by
\begin{equation}
\label{A_q_classical_interaction}
\mc{W} := \frac{1}{3} \int \dd x \, \dd y \, \dd z \, w(x-y)\,w(y-z)\,w(z-x)\, |\vph(x)|^2\,|\vph(y)|^2\,|\vph(z)|^2\,.
\end{equation}
With \eqref{q_classical_interaction} replaced by \eqref{A_q_classical_interaction}, we define the classical Gibbs measure and $p$-particle correlation functions as in \eqref{q_gibbs_cutoff_defn_rigorous}--\eqref{q_classical_correln_fn_defn} above. Moreover, instead of \eqref{W_kernel}, let us now consider
\begin{multline*}
W(x_1,x_2,x_3;y_1,y_2,y_3)
\\
:=w(x_1-x_2)\,w(x_2-x_3)\,w(x_3-x_1)\,\delta(x_1-y_1)\,\delta(x_2-y_2)\,\delta(x_3-y_3)\,,
\end{multline*}
(again for $w$ as in Assumption \ref{A_w_assumption}). With this modification, we define the time-evolved observable as in \eqref{q_time_evolved_rv}, where now $S_t$ denotes the flow of \eqref{A_q_Hartree_defn}.

In the many-body quantum setting, we consider $w$ as in Assumption \ref{w_assumption_bounded} and define the Hamiltonian $H_{\tau}$ by \eqref{q_quantum_interacting_Hamn}, where now instead of by \eqref{Quantum_interaction}, the interaction is given by
\begin{equation*}
\mc{W}_\tau := \frac{1}{3} \int \dd x \, \dd y \, \dd z \, w(x-y)\,w(y-z)\,w(z-x)\,\vphts(x)\vphts(y)\vphts(z)\vpht(x)\vpht(y)\vpht(z)\,.
\end{equation*}
In particular, we have 
\begin{equation*}
H_{\tau}=\bigoplus_{n=0}^{\infty} H_{\tau}^{(n)}\,,
\end{equation*}
where 
\begin{equation*}
H^{(n)}_{\tau} =
\frac{1}{\tau}\,\sum_{i=1}^n (-\Delta_i + \kappa) + \frac{1}{3\tau^3} \mathop{\sum_{i,j,k}^{n}}_{i \neq j \neq k \neq i} w(x_i-x_j)\,w(x_j-x_k)\,w(x_k-x_i)\,.
\end{equation*}
With these modifications, we consider the objects defined as in \eqref{q_quantum_state_defn}--\eqref{q_quantum_time_evolution} above. 

We can now state our results in this alternative model. Let us first state the time-independent results.
\begin{theorem}[Convergence for bounded interaction potentials]
\label{A_q_bounded_state_convergence_thm}
Let $w$ be as in Assumption \ref{w_assumption_bounded}. With objects as defined in the appendix, we have that \eqref{correlation_function_convergence_bounded} and \eqref{partition_function_convergence_bounded} hold.
\end{theorem}
\begin{theorem}[Convergence for $L^\frac{3}{2}$ interaction potentials]
\label{q_L32_state_convergence_thm}
Let $w$ be as in Assumption \ref{w_assumption}. Suppose that $w^\eps$ is a sequence of interaction potentials as in Assumption \ref{w_assumption_bounded} converging to $w$ in $L^{\frac{3}{2}}$. %Let $\gamma_{p}$, $\gamma^\eps_{\tau,p}$, $z$, and $\mc{Z}^\eps_{\tau}$ be defined as in \eqref{q_classical_correln_fn_defn}, \eqref{q_quantum_correlation_fn_defn}, \eqref{q_classical_partn_functn}, and \eqref{q_relative quantum_partition_fn} respectively. 
With objects defined analogously as for Theorem \ref{A_q_bounded_state_convergence_thm}, there exists a sequence $\eps_\tau \rightarrow 0$ as $\tau \to \infty$ such for any $p \in \mathbb{N}^*$, we have
\begin{equation*}
%\label{correlation_function_convergence_L32}
\lim_{\tau \to \infty}\|\gamma^{\varepsilon_\tau}_{\tau,p} - \gamma_p\|_{\mfS^1(\mfhp)} = 0\,,
\end{equation*}
and such that
\begin{equation*}
%\label{partition_function_convergence_L32}
\lim_{\tau \to \infty} \mc{Z}^{\varepsilon_\tau}_{\tau} = z\,.
\end{equation*}
\end{theorem}
We can also obtain the following time-dependent results.
\begin{theorem}[Convergence for bounded potentials]
\label{A_q_bounded_time_thm}
Let $w$ be as in Assumption \ref{w_assumption_bounded}. Let $m \in \mathbb{N}^*$, $p_1,\ldots,p_m \in \N^*$, $\xi_1 \in \mc{L}(\mathfrak{h}^{(p_1)}),\ldots,\xi_m \in \mc{L}(\mathfrak{h}^{(p_m)})$, and $t_1,\ldots,t_m \in \R$ be given. With objects as defined in the appendix, we have that \eqref{Theorem_1.12} holds.
\end{theorem}
\begin{theorem}[Convergence for $L^{\frac{3}{2}}$ potentials]
\label{q_L32_time_thm}
Let $w$ be as in Assumption \ref{w_assumption}. Let $w^\eps$ be defined as in Theorem \ref{q_L32_state_convergence_thm} above. With objects defined as for Theorem \ref{A_q_bounded_time_thm}, there is a sequence $\varepsilon_\tau \rightarrow 0$ as $\tau \to \infty$ such that, for all $m \in \mathbb{N}^*$, $p_1,\ldots,p_m \in \N^*$, $\xi_1 \in \mc{L}(\mathfrak{h}^{(p_1)}),\ldots,\xi_m \in \mc{L}(\mathfrak{h}^{(p_m)})$, and $t_1,\ldots,t_m \in \R$, we have
\begin{equation*}
\lim_{\tau \rightarrow \infty} \rt^{\eps_\tau}(\Psi^{t_1}_\tau\Th_\tau(\xi_1) \ldots \Psi^{t_m}_\tau\Th_\tau(\xi_m)) = \rho(\Psi^{t_1}\Th(\xi_1) \ldots \Psi^{t_m}\Th(\xi_m))\,.
\end{equation*}
\end{theorem}

Let us note the following results, which are analogues of Propositions \ref{Cauchy_problem_1} and \ref{Cauchy_problem_2} for \eqref{A_q_Hartree_defn} respectively.

\begin{proposition}[Deterministic local existence for \eqref{A_q_Hartree_defn}]
\label{A_Cauchy_problem_1}
%Consider $w \in L^{\frac{3}{2}}(\Lambda)$ even. Then 
The Cauchy problem \eqref{A_q_Hartree_defn} is locally well-posed in $H^s(\La)$ for $s>0$.
\end{proposition}

\begin{proposition}[Invariance of truncated Gibbs measure and almost sure global existence for \eqref{A_q_Hartree_defn}]
\label{A_Cauchy_problem_2}
The following claims hold.
\begin{itemize}
\item[(i)]  Recall the probability space $(\mathbb{C}, \mc{G}, \mu)$ defined in \eqref{q_rigorous_gibbs_defn}, $\vph \equiv \vph^\om$ defined in \eqref{q_class_free_field}, and $\mathcal{N}$ defined in \eqref{q_mass}. For $B>0$ sufficiently small, we have 
\begin{equation*}
\e^{-\frac{1}{3} \int \dd x \, \dd y \, \dd z \, w(x-y)\,w(y-z)\,w(x-z) \,|\vph(x)|^2|\vph(y)|^2|\vph(z)|^2}  \chi_{(\mcN \leq B)} \in L^1(\dd\mu)\,.
\end{equation*}
In particular, taking $K=B$ in \eqref{q_cut-off_support}, we get that the probability measure $\mbP^f_{\mathrm{Gibbs}}$ is well-defined.
\item[(ii)] Consider $s \in (0,\frac{1}{2})$. The measure $\mbP^f_{\mathrm{Gibbs}}$ is invariant under the flow of \eqref{A_q_Hartree_defn}. Furthermore, \eqref{A_q_Hartree_defn} admits global solutions for $\mbP^f_{\mathrm{Gibbs}}$-almost every $u_0 \in H^s(\Lambda)$.
\end{itemize}
\end{proposition}
In order to obtain Proposition \ref{A_Cauchy_problem_1} and \ref{A_Cauchy_problem_2}, it is necessary to prove the following two multilinear estimates, which are analogues of Lemmas \ref{Multlinear_estimates_1} and \ref{Multlinear_estimates_2} respectively.

\begin{lemma}
\label{A_Multlinear_estimates_1}
Consider $w_1,w_2,w_3 \in L^{\frac{3}{2}}(\La)$.
Given $q \in H^{\frac{1}{3}}(\La)$, let 
\begin{equation}
\label{A_N_1_definition}
\mathcal{M}_1(q):=\int \dd x\,\dd y\,\dd z\, w_1(x-y)\,w_2(y-z)\,w_3(z-x)\,|q(x)|^2\,|q(y)|^2\,|q(z)|^2\,.
\end{equation}
Then, we have 
\begin{equation*}
%\label{Multlinear_estimates_1_claim}
|\mathcal{M}_1(q)| \lesssim \|w_1\|_{L^{\frac{3}{2}}}\,\|w_2\|_{L^{\frac{3}{2}}}\,\|w_3\|_{L^{\frac{3}{2}}}\,\|q\|_{H^{\frac{1}{3}}}^6\,.
\end{equation*}
\end{lemma}

\begin{lemma}
\label{A_Multlinear_estimates_2}
Consider $w_1,w_2,w_3 \in L^{\frac{3}{2}}(\La)$.
Given $s,\varepsilon>0$ and $v_j \in X^{s,1/2-\varepsilon}$, for $j=1,\ldots,5$, we let
\begin{multline}
\label{A_N_2_definition}
\mathcal{M}_2(v_1,v_2,v_3,v_4,v_5)(x,t):=
\\
\int \dd y\,\dd z\,w_1(x-y)\,w_2(y-z)\,w_3(z-x)\,v_1(y,t)\,\overline{v_2(y,t)}\,v_3(z,t)\,\overline{v_4(z,t)}\,v_5(x,t)\,.
\end{multline}
For $\varepsilon>0$ sufficiently small, we have that for all $t_0 \in \R$ and $\delta>0$ small
\begin{equation}
\label{A_Multlinear_estimates_2_claim_2}
\|\mathcal{M}_2(v_1,v_2,v_3,v_4,v_5)\|_{X^{s,-\frac{1}{2}+\varepsilon}_{[t_0,t_0+\delta]}} \lesssim_{s,\varepsilon} \delta^{\varepsilon}\, \|w_1\|_{L^{\frac{3}{2}}}\,\|w_2\|_{L^{\frac{3}{2}}}\,\|w_3\|_{L^{\frac{3}{2}}}\,\prod_{j=1}^{5}\|v_j\|_{X^{s,\frac{1}{2}+\varepsilon}_{[t_0,t_0+\delta]}}\,.
\end{equation}
Here, we recall \eqref{local_X^{s,b}_norm}.
\end{lemma}

We now present the details of the proofs Lemmas \ref{A_Multlinear_estimates_1} and \ref{A_Multlinear_estimates_2}. Let us note that the arguments are quite different, and more involved, than those used to prove Lemmas \ref{Multlinear_estimates_1} and \ref{Multlinear_estimates_2} above. The reason for this is that the nonlinearities \eqref{A_N_1_definition} and \eqref{A_N_2_definition} no longer have a convolution structure.

\begin{proof}[Proof of Lemma \ref{A_Multlinear_estimates_1}]
By a density argument, we may assume without loss of generality that $w_1,w_2,w_3,q$ are smooth functions, thus making the calculations that follow rigorous. Writing each integrand as a Fourier series, we compute
\begin{multline}
\label{A_N_1_1}
\mathcal{M}_1(q) =
\\
\sum_{k_1,\ldots,k_6} \sum_{\zeta_1,\zeta_2,\zeta_3} \int \dd x\,\dd y\,\dd z\,
\widehat{w}_1(\zeta_1)\,\widehat{w}_2(\zeta_2)\,\widehat{w}_3(\zeta_3)\,\widehat{q}\,(k_1)\,\overline{\widehat{q}\,(k_2)}\,\widehat{q}\,(k_3)\,\overline{\widehat{q}\,(k_4)}\,\widehat{q}\,(k_5)\,\overline{\widehat{q}\,(k_6)}\,
\\ 
\times \e^{2\pi \mathrm{i} x(k_1-k_2+\zeta_1-\zeta_3)}\,
\e^{2\pi \mathrm{i} y(k_3-k_4-\zeta_1+\zeta_2)}\,\e^{2\pi \mathrm{i}z (k_5-k_6-\zeta_2+\zeta_3)}
\, 
\end{multline}
The summations in \eqref{A_N_1_1} and in the sequel are taken over $\Z$.
By integrating in $x,y,z$ in \eqref{A_N_1_1} and taking absolute values, we deduce that
\begin{multline}
\label{A_N_1_2}
|\mathcal{M}_1(q)| \leq 
\\
\sum_{k_1,\ldots,k_6} \sum_{\zeta_1,\zeta_2,\zeta_3} 
|\widehat{w}_1(\zeta_1)|\,|\widehat{w}_2(\zeta_2)|\,|\widehat{w}_3(\zeta_3)|\,|\widehat{q}\,(k_1)|\,|\widehat{q}\,(k_2)|\,|\widehat{q}\,(k_3)|\,|\widehat{q}\,(k_4)|\,|\widehat{q}\,(k_5)|\,|\widehat{q}\,(k_6)|\,
\\ 
\times \delta(k_1-k_2+\zeta_1-\zeta_3)\,
\delta(k_3-k_4-\zeta_1+\zeta_2)\,\delta(k_5-k_6-\zeta_2+\zeta_3)
\,.
\end{multline}
By the constraints on the summands, we can rewrite the expression on the right-hand side of \eqref{A_N_1_2} as
\begin{multline}
\label{A_N_1_3}
\sum_{k_1,\ldots,k_6} \biggl\{ \sum_{\zeta_1} 
|\widehat{w}_1(\zeta_1)|\,|\widehat{w}_2(\zeta_1-k_3+k_4)|\,|\widehat{w}_3(\zeta_1+k_1-k_2)|\biggr\}\,
\\
\times
|\widehat{q}\,(k_1)|\,|\widehat{q}\,(k_2)|\,|\widehat{q}\,(k_3)|\,|\widehat{q}\,(k_4)|\,|\widehat{q}\,(k_5)|\,|\widehat{q}\,(k_6)|\,
\\ 
\times \delta(k_1-k_2+k_3-k_4+k_5-k_6)
\,,
\end{multline}
which by applying H\"{o}lder's inequality in $\zeta_1$ in the curly brackets in \eqref{A_N_1_3} is
\begin{multline}
\label{A_N_1_4}
\leq \|\widehat{w}_1\|_{L^3}\, \|\widehat{w}_2\|_{L^3}\,\|\widehat{w}_3\|_{L^3}
\, \sum_{k_1,\ldots,k_6}
|\widehat{q}\,(k_1)|\,|\widehat{q}\,(k_2)|\,|\widehat{q}\,(k_3)|\,|\widehat{q}\,(k_4)|\,|\widehat{q}\,(k_5)|\,|\widehat{q}\,(k_6)|\,
\\ 
\times \delta(k_1-k_2+k_3-k_4+k_5-k_6)
\,.
\end{multline}
By the Hausdorff-Young inequality and Parseval's identity, we have that 
\begin{equation}
\label{A_N_1_5}
\eqref{A_N_1_4} \lesssim \|w_1\|_{L^{\frac{3}{2}}}\,\|w_2\|_{L^{\frac{3}{2}}}\,\|w_3\|_{L^{\frac{3}{2}}}\,
\int \dd x \,|F(x)|^6\,,
\end{equation}
where $F$ is chosen such that 
\begin{equation}
\label{A_N_1_6}
\widehat{F}=|\widehat{q}\,|\,.
\end{equation}
By applying Sobolev embedding, we deduce that
\begin{multline}
\label{A_N_1_7}
\eqref{A_N_1_5} =  \|w_1\|_{L^{\frac{3}{2}}}\,\|w_2\|_{L^{\frac{3}{2}}}\,\|w_3\|_{L^{\frac{3}{2}}}\,\|F\|_{L^6}^6 \lesssim 
\|w_1\|_{L^{\frac{3}{2}}}\,\|w_2\|_{L^{\frac{3}{2}}}\,\|w_3\|_{L^{\frac{3}{2}}}\,\|F\|_{H^{\frac{1}{3}}}^6
\\
= \|w_1\|_{L^{\frac{3}{2}}}\,\|w_2\|_{L^{\frac{3}{2}}}\,\|w_3\|_{L^{\frac{3}{2}}}\,\|q\|_{H^{\frac{1}{3}}}^6\,.
\end{multline}
For the last equality in \eqref{A_N_1_7}, we used that $\|F\|_{H^{\frac{1}{3}}}=\|q\|_{H^{\frac{1}{3}}}$, which follows by \eqref{A_N_1_6}.
\end{proof}

\begin{proof}[Proof of Lemma \ref{A_Multlinear_estimates_2}]
Arguing analogously as in the proof of Lemma \ref{Multlinear_estimates_2}, it suffices to prove the following global version of the estimate \eqref{A_Multlinear_estimates_2_claim_2}.
\begin{equation*}
\|\mathcal{M}_2(v_1,v_2,v_3,v_4,v_5)\|_{X^{s,-\frac{1}{2}+\varepsilon}} \lesssim_{s,\varepsilon} \|w_1\|_{L^{\frac{3}{2}}}\,\|w_2\|_{L^{\frac{3}{2}}}\,\|w_3\|_{L^{\frac{3}{2}}}\,\prod_{j=1}^{5}\|v_j\|_{X^{s,\frac{1}{2}-\varepsilon}}\,.
\end{equation*}
Arguing as for \eqref{N_2_claim}, it suffices to show that for $w_i,\, i=1,2,3$ and $v_j,\, j=1,\ldots,5$  and $k \in \Z, \eta \in \R$, we have
\begin{multline}
\label{A_N_2_claim}
\big|\bigl(\mathcal{N}_2(v_1,v_2,v_3,v_4,v_5)\bigr)\,\widetilde{\,}\,(k,\eta)\big|
\lesssim \|w_1\|_{L^{\frac{3}{2}}}\,\|w_2\|_{L^{\frac{3}{2}}}\,\|w_3\|_{L^{\frac{3}{2}}}
\\
\times \sum_{k_1,\ldots,k_5} \int \dd \eta_1 \cdots \dd \eta_5\, \delta(k_1-k_2+k_3-k_4+k_5-k)\,\delta(\eta_1-\eta_2+\eta_3-\eta_4+\eta_5-\eta)\,
\\
\times |\widetilde{v}_1(k_1,\eta_1)|\,|\widetilde{v}_2(k_2,\eta_2)|\,|\widetilde{v}_3(k_3,\eta_3)|\,|\widetilde{v}_4(k_4,\eta_4)|\,|\widetilde{v}_5(k_5,\eta_5)|\,.
\end{multline}
By expanding all of the integrands as a Fourier series and arguing analogously as in the proof of Lemma \ref{A_Multlinear_estimates_1}, we compute

\begin{multline}
\label{A_N_2_1}
\mathcal{N}_2(v_1,v_2,v_3,v_4,v_5)(x,t)=
\\
\sum_{k_1,\ldots,k_5} \sum_{\zeta_1,\zeta_2,\zeta_3}
\int \dd y\,\dd z\,\int \dd \eta_1 \cdots \dd \eta_5\,
\widehat{w}_1(\zeta_1)\,\widehat{w}_2(\zeta_2)\,\widehat{w}_3(\zeta_3)\,\widetilde{v}_1(k_1,\eta_1)\,\overline{\widetilde{v}_2(k_2,\eta_2)}\,
\\
\times \widetilde{v}_3(k_3,\eta_3)\,\overline{\widetilde{v}_4(k_4,\eta_4)}\,\widetilde{v}_5(k_5,\eta_5)\,\e^{2\pi \mathrm{i} x(\zeta_1-\zeta_3+k_5)}\,\e^{2\pi \mathrm{i} y (-\zeta_1+\zeta_2+k_1-k_2)}\, 
\\
\times 
\e^{2\pi \mathrm{i} z(-\zeta_2+\zeta_3+k_3-k_4)}\,\e^{2\pi \mathrm{i} t (\eta_1-\eta_2+\eta_3-\eta_4+\eta_5)}\,.
\end{multline}
From \eqref{A_N_2_1}, we hence deduce that
\begin{multline}
\label{A_N_2_2}
\bigl(\mathcal{N}_2(v_1,v_2,v_3,v_4,v_5)\bigr)\,\widetilde{\,}\,(k,\eta)=
\\
\sum_{k_1,\ldots,k_5} \sum_{\zeta_1,\zeta_2,\zeta_3} \int \dd \eta_1 \cdots \dd \eta_5\,
\widehat{w}_1(\zeta_1)\,\widehat{w}_2(\zeta_2)\,\widehat{w}_3(\zeta_3)\,\widetilde{v}_1(k_1,\eta_1)\,\overline{\widetilde{v}_2(k_2,\eta_2)}\,
\\
\times \widetilde{v}_3(k_3,\eta_3)\,\overline{\widetilde{v}_4(k_4,\eta_4)}\,\widetilde{v}_5(k_5,\eta_5)\,\delta(\zeta_1-\zeta_3+k_5-k)\,\delta(-\zeta_1+\zeta_2+k_1-k_2)\, 
\\
\times 
\delta(-\zeta_2+\zeta_3+k_3-k_4)\,\delta(\eta_1-\eta_2+\eta_3-\eta_4+\eta_5-\eta)
\\
=\sum_{k_1,\ldots,k_5} \sum_{\zeta_1} \int \dd \eta_1 \cdots \dd \eta_5\,
\widehat{w}_1(\zeta_1)\,\widehat{w}_2(\zeta_1-k_1+k_2)\,\widehat{w}_3(\zeta_1+k_5-k)\,
\\
\times
\widetilde{v}_1(k_1,\eta_1)\,\overline{\widetilde{v}_2(k_2,\eta_2)}\, \widetilde{v}_3(k_3,\eta_3)\,\overline{\widetilde{v}_4(k_4,\eta_4)}\,\widetilde{v}_5(k_5,\eta_5)\,
\\
\times \delta(k_1-k_2+k_3-k_4+k_5-k)\,
\delta(\eta_1-\eta_2+\eta_3-\eta_4+\eta_5-\eta)\,.
\end{multline}
We now deduce \eqref{A_N_2_claim} from \eqref{A_N_2_2} by arguing analogously as in the proof of Lemma \ref{A_Multlinear_estimates_1}. 
\end{proof}

One then obtains Proposition \ref{A_Cauchy_problem_1} by arguing analogously as in the proof of Proposition \ref{Cauchy_problem_1} above. The only difference is that one uses Lemma \ref{A_Multlinear_estimates_2} instead of Lemma \ref{Multlinear_estimates_2}. For details, see \cite[Proof of Proposition 4.1.1]{Rout_Thesis_2023} and \cite[Proof of Proposition 2.1]{RS23}. We now prove Proposition \ref{A_Cauchy_problem_2}(i). Given Proposition \ref{A_Cauchy_problem_2}(i), the proof of Proposition \ref{A_Cauchy_problem_2}(ii) is similar to that of Proposition \ref{Cauchy_problem_2}(ii). For details, see \cite[Proof of Proposition 4.1.2 (ii)]{Rout_Thesis_2023} and \cite[Proof of Proposition 2.2 (ii)]{RS23}.

\begin{proof}[Proof of Proposition \ref{A_Cauchy_problem_2} (i)]
The analysis is similar to that of \cite[Lemma 3.10]{Bou94}, which was recalled in \cite[Appendix A]{RS22}. We follow the exposition in \cite[Appendix A]{RS22} and explain the main differences.
By using Proposition \ref{A_Multlinear_estimates_1} with $w_1=w_2=w_3$ and $q=\vph$, it follows that 
\begin{equation}
\label{A_Cauchy_problem_2_i_1}
\biggl|-\frac{1}{3}\,\int \dd x \, \dd y \, \dd z \, w(x-y)\,w(y-z)\,w(z-x)\,|\vph(x)|^2\,|\vph(y)|^2\,|\vph(z)|^2\biggr|
\lesssim \|w\|_{L^{\frac{3}{2}}}^3\,\|\vph\|_{H^{\frac{1}{3}}}^6\,.
\end{equation}

The estimate \eqref{A_Cauchy_problem_2_i_1} is the key reduction to the proof of \cite[Lemma 3.10]{Bou94}. 
In particular, given $c_0>0$, it suffices to show that for $B>0$ sufficiently small, depending on $c_0$, we have
\begin{equation}
\label{A_Cauchy_problem_2_i_1_B}
\e^{c_0\|\varphi\|^6_{H^{1/3}}}\,\chi (\mathcal{N} \leq B) \in L^1(\dd \mu)\,.
\end{equation}
%We follow the summary of this proof given in \cite[Appendix A]{RS22}. Here, we only emphasise the main changes and refer the reader to \cite[Appendix A]{RS22} for more precise details and definitions.
By arguing analogously as for \cite[(A.4)]{RS22} (see also \cite[(3.11)]{Bou94}), we reduce the proof of \eqref{A_Cauchy_problem_2_i_1_B} to showing that there exists $c>0$ such that for large enough $\lambda>0$, we have
\begin{equation}
\label{A_Cauchy_problem_2_i_2}
\mu \Biggl[\Biggl\|\sum_{k \in \Z} \frac{\omega_k}{\sqrt{\lambda_k}}\,\e^{2\pi \mathrm{i} k x}\Biggr\|_{H^{\frac{1}{3}}}>\lambda\,, \,\, \Biggl(\sum_{k \in \Z} \frac{|\omega_k|^2}{\lambda_k}\Biggr) \leq B \Biggr] \lesssim \mathrm{exp}\bigl(-c M_0^{4/3}\,\lambda^2\bigr)\,,
\end{equation}
where 
\begin{equation}
\label{A_M_0}
M_0 \sim \biggl(\frac{\lambda}{B}\biggr)^3\,.
\end{equation}
Here, we recall \eqref{q_class_free_field}. In other words, we reduce to the analysis from \cite[Lemma 3.10]{Bou94} and \cite[Appendix A]{RS22} with $p=6$, with $\|\cdot\|_{L^6}$ norms replaced by $\|\cdot\|_{H^{\frac{1}{3}}}$ norms\footnote{These norms are linked by the Sobolev embedding $H^{\frac{1}{3}} \subset L^6$.} Since the norms are now defined in Fourier space, the analysis in fact simplifies.
In particular, the analogue of \cite[(A.7)]{RS22}, which is now given by
\begin{equation*} 
\Biggl\|\sum_{|k| \sim M} a_k
 \e^{2\pi \mathrm{i} k x}\Biggr\|_{H^{\frac{1}{3}}} \lesssim M^{\frac{1}{3}}\,\Biggl\|\sum_{|k| \sim M} a_k \e^{2\pi \mathrm{i} k x}\Biggr\|_{L^2}\,,
\end{equation*}
follows by definition of  $\|\cdot\|_{H^{\frac{1}{3}}}$. Here, $M$ is a dyadic integer and $|k| \sim M$ means that $\frac{3M}{4} \leq |k| < \frac{3M}{2}$.
Similarly, the analogue of \cite[(A.14)]{RS22} which is now given by 
\begin{equation*} 
\Biggl\|\sum_{|k| \sim M} \frac{\omega_k}{\sqrt{\lambda_k}} \,\e^{2\pi \mathrm{i} k x}\Biggr\|_{H^{\frac{1}{3}}} \lesssim \frac{1}{M}\,\Biggl\|\sum_{|k| \sim M} \omega_k \e^{2\pi \mathrm{i} k x}\Biggr\|_{H^{\frac{1}{3}}}\,,
\end{equation*}
follows by definition of  $\|\cdot\|_{H^{\frac{1}{3}}}$.

With $M_0$ as in \eqref{M_0} and  
\begin{equation*}
\sigma_M \sim M^{-1/6}+(M_0/M)^{1/2}\,, \quad M>M_0\,,
\end{equation*}
as in \cite[(A.17)]{RS22}, the above modifications allow us to deduce that there is $M>M_0$ a dyadic integer such that 
\begin{equation}
\label{A.17'}
\Biggl\|\sum_{|k| \sim M}\omega_k\, \e^{2\pi \mathrm{i} k x}\Biggr\|_{H^{\frac{1}{3}}}>\sigma_M M \lambda\,.
\end{equation}
We now conclude the proof of \eqref{A_Cauchy_problem_2_i_2} by using \eqref{A.17'} and a union bound as in \cite[(A.18)--(A.24)]{RS22}.
In order to apply the above union bound argument, one needs to show that, given a dyadic integer $M$, and
\begin{equation*}
\mathcal{S} \equiv \mathcal{S}_M:=\Biggl\{\sum_{|k| \sim M} a_k \,\e^{2\pi \mathrm{i} k x}\,, a_k \in \C\Biggr\}\,,
\end{equation*}
there exists a set $\Xi \equiv \Xi_M$ contained in the unit sphere of $H^{-\frac{1}{3}}$ satisfying the following properties.
\begin{itemize}
\item[(1)] $\mathrm{max}_{\phi \in \Xi} |\langle g, \phi \rangle| \geq \frac{1}{2}\,\|g\|_{H^{\frac{1}{3}}}$ for all $g \in \mathcal{S}$.
\item[(2)] $\|\phi\|_{L^2} \lesssim M^{\frac{1}{3}}$ for all $\phi \in \Xi$.
\item[(3)] $\log |\Xi| \lesssim M$.
\end{itemize}
The property (2) above follows from the frequency localisation of $\phi$ and the assumption that $\|\phi\|_{H^{-1/3}}=1$. Properties (1) and (2) follow from the analogue of \cite[Lemma A.4]{RS22} with the spaces $L^p$ and $L^{p'}$ replaced with $H^{\frac{1}{3}}$ and $H^{-\frac{1}{3}}$ respectively. The modified claim follows from the same proof. We refer the reader to \cite[Appendix A]{RS22} for the full details.
\end{proof}

The rest of the analysis leading to the results stated in Theorems \ref{A_q_bounded_state_convergence_thm}--\ref{q_L32_time_thm} is then quite similar as before. The graph structure is slightly different, see \cite[Section 4.2.6]{Rout_Thesis_2023} and \cite[Section 3.6]{RS23}, but this does not affect the above arguments in a substantial way. We refer the reader to \cite[Chapter 4]{Rout_Thesis_2023} and \cite{RS23} for the full details of the proofs.

\bigskip
\emph{Acknowledgements:}
The authors would like to thank Zied Ammari, J\"{u}rg Fr\"{o}hlich, Sebastian Herr, Antti Knowles, Jinyeop Lee, Phan Th\`{a}nh Nam, Nicolas Rougerie, Mathieu Lewin, Renato Luc\`{a}, Benjamin Schlein, Arnaud Triay, Nikolay Tzvetkov, and Daniel Ueltschi for useful discussions and comments during various stages of the project. 
A.R. is supported by the Warwick Mathematics Institute Centre for Doctoral Training, and gratefully acknowledges funding from the University of Warwick. V.S. acknowledges support of the EPSRC New Investigator Award grant EP/T027975/1.

\end{document}